\documentclass[letterpaper]{article}
\usepackage[vmargin=2.5cm,hmargin=2cm]{geometry}
\usepackage{amsmath,amssymb,amsthm,mathtools}
\usepackage{mathrsfs,bm,color}
\usepackage{hyperref,fancyhdr,lastpage}
\usepackage{enumitem,calc}
\usepackage{graphicx}
\usepackage{tikz,tikz-cd}
\usetikzlibrary{arrows,trees,positioning,shapes.geometric,shapes.misc}
\usepackage{pgfplots}
\usepackage{bbm}
\usepackage{caption}
\usepackage{subcaption}
\usepackage{algorithm}
\usepackage[noend]{algpseudocode}
\usepackage{indentfirst}
\usepackage{array}
\usepackage{enumitem}
\usepackage[toc,page]{appendix}
\usepackage{imakeidx}

\makeatletter
\def\BState{\State\hskip-\ALG@thistlm}
\makeatother

\newtheorem{theorem}{Theorem}[section]
\newtheorem{corollary}[theorem]{Corollary}
\newtheorem{lemma}[theorem]{Lemma}

\newtheorem{definition}[theorem]{Definition}
\newtheorem{observation}[theorem]{Observation}

\newcommand{\un}[1]{\ensuremath{u_{\textsc{#1}}}} 
\newcommand{\nn}[1]{\ensuremath{n_{\textsc{#1}}}} 
\newcommand{\tabitem}{\textbullet\space}

\def\uint{\kern7pt\raise14.8pt
\hbox{\vrule height.8pt width4pt}\kern-11pt\int}
\def\lint{\kern1pt\lower10.3pt
\hbox{\vrule height.7pt width4pt}\kern-5pt\int}

\algdef{SE}[SUBALG]{Indent}{EndIndent}{}{\algorithmicend\ }%
\algtext*{Indent}
\algtext*{EndIndent}

\tikzset{
	treenode/.style = {align=center, inner sep=0pt, text centered,
		font=\sffamily},
	arn_b/.style = {treenode, circle, white, font=\sffamily\bfseries, draw=black, fill=black, text width=1em},
	arn_g/.style = {treenode, circle, black, draw=black, 
		text width=1em, very thick, fill=lightgray},
	arn_w/.style = {treenode, circle, black, draw=black, 
		text width=1em, very thick}, 
	arn_ws/.style = {treenode, circle, black, draw=black, very thick, minimum size=0.15cm}, 
	arn_wx/.style = {treenode, rectangle, black, draw=black, 
		minimum width=1em, minimum height=1em, very thick}, 
	arn_gx/.style = {treenode, rectangle, black, draw=black, 
		minimum width=1em, minimum height=1em, very thick, fill=lightgray}, 
	arn_bx/.style = {treenode, rectangle, draw=black,
		minimum width=1em, minimum height=1em, very thick}, 
		arn_bs/.style = {treenode, circle, white, font=\sffamily\bfseries, draw=black,
		fill=black, scale=0.8, text width=1em}, 
	arn_gs/.style = {treenode, circle, black, draw=black, 
		text width=1em, thick, scale=0.8, fill=lightgray}, 
	arn_ws/.style = {treenode, circle, black, draw=black, 
		text width=1em, thick, scale=0.8, minimum height=1em}, 
	arn_wxs/.style = {treenode, rectangle, black, draw=black, 
		minimum width=1em, minimum height=1em, scale=0.8, thick}, 
	arn_gxs/.style = {treenode, rectangle, scale=0.8,  black, draw=black, 
		minimum width=1em, minimum height=1em, thick, fill=lightgray}, 
	arn_bxs/.style = {treenode, rectangle, draw=black,
		minimum width=1em, minimum height=1em, scale=0.8, thick}, 
	sd/.style args = {#1} {sibling distance=#1}
}

\tikzstyle{box} = [draw, rectangle, rounded corners, thick, node distance=7em, text width=6em, text centered, minimum height=3.5em]
\tikzstyle{container} = [draw, rectangle, dashed, inner sep=2em]
\tikzstyle{line} = [draw, thick, -latex']

\hypersetup{
	colorlinks,
	linkcolor={red!50!black},
	citecolor={green!50!black},
	urlcolor={blue!80!black}
}

\makeindex
\title{The Amortized Analysis of a Non-blocking Chromatic Tree}
\author{Jeremy Ko, University of Toronto, jerko@cs.toronto.edu}

\begin{document}
\maketitle

\begin{abstract}
	A non-blocking chromatic tree is a type of balanced binary search tree where multiple processes can concurrently perform search and update operations. We prove that a certain implementation has amortized cost $O(\dot{c} + \log n)$ for each operation, where $\dot{c}$ is the maximum number of concurrent operations during the execution and $n$ is the maximum number of keys in the tree during the operation. This amortized analysis presents new challenges compared to existing analyses of other non-blocking data structures. 
\end{abstract}

\section{Introduction}\label{section_introduction}

A \textit{concurrent data structure} is one that can be concurrently updated by multiple processes.  A \textit{wait-free}\index{wait-free} implementation of a concurrent data structure guarantees that every operation completes within a finite number of steps by the process that invoked the operation. A \textit{non-blocking}\index{non-blocking} implementation of a concurrent data structure guarantees that whenever there are active operations, one operation will eventually complete in a finite number of steps. Since a particular operation may not complete in a finite number of steps, it not possible to perform a worst-case analysis for non-blocking data structures. However, such implementations are desirable since they are typically less complicated than wait-free implementations and often perform well in practice. Amortized analysis gives an upper bound on the worst-case number of steps performed during a sequence of operations in an execution, rather than on the worst-case step complexity of a single operation. This type of analysis is useful for expressing the efficiency of non-blocking data structures.

In this paper, we present an amortized analysis of a non-blocking chromatic tree. The chromatic tree was introduced by Nurmi and Soisalon-Soininen \cite{DBLP:journals/acta/NurmiS96} as a generalization of a red-black tree with relaxed balance conditions, allowing insertions and deletions to be performed independently of rebalancing. Boyar, Fagerberg, and Larsen \cite{DBLP:conf/wads/BoyarFL95} showed that the amortized number of rebalancing transformations per update is constant, provided each update and each rebalancing transformation is performed without interference from other operations.

Brown, Ellen, and Ruppert \cite{DBLP:conf/ppopp/BrownER14} gave a non-blocking implementation using LLX and SCX primitives \cite{DBLP:conf/podc/BrownER13}, which themselves can be implemented using CAS. Unfortunately, this implementation does not have good amortized complexity. We prove that a slightly modified implementation has amortized cost $O(\dot{c}(\alpha) + \log n(op))$ for each operation $op$ in an execution $\alpha$, where $\dot{c}(\alpha)$ is the point contention of $\alpha$ and $n(op)$\index{$n(op)$} is the maximum number of nodes in the chromatic tree during $op$'s execution interval.  

Our amortized analysis for the chromatic tree is based on the amortized analysis for the unbalanced binary search tree by Ellen, Fatourou, Helga, and Ruppert \cite{DBLP:conf/podc/EllenFHR13}. Several challenges make our analysis more difficult than their analysis. 

The unbalanced binary search tree only has one transformation to perform insertions and one transformation to perform deletions. The chromatic tree also has 11 different rebalancing transformations (not including their mirror images). Our analysis has the feature that it does not require a separate case for each type of transformation performed. A very similar analysis should give the same amortized step complexity for other types of balanced binary search trees, especially for those implemented using LLX and SCX \cite{BrownThesis17}.

Unlike update transformations, which only occur at leaf nodes, rebalancing transformations may occur anywhere in the tree. Furthermore, operations may perform rebalancing transformations on behalf of other operations, or may cause them to perform additional rebalancing transformations they would otherwise not perform. For an amortized analysis done using the accounting method, these facts make it much more difficult to determine the steps in which dollars should be deposited into bank accounts to pay for future steps.

Much like the analysis of the unbalanced binary search tree, many overlapping transformations can cause all but one to fail. Many operations may also be working together to perform the same rebalancing transformations. The amortized analysis of the chromatic tree is complicated by the fact that the constant upper bound on the number of rebalancing transformations per operation is amortized, even when each update and rebalancing transformation is performed atomically. Since the number of rebalancing transformations performed by each operation may differ, a bad configuration in which many rebalancing transformations can occur may also have high contention. Executions including such configurations must be accounted for in the amortized analysis.

The remainder of this paper is organized as follows. In Section~\ref{section_model}, we present the asynchronous shared memory model assumed in our analysis. In Section~\ref{section_related}, we give an overview of related work done on the amortized analysis of concurrent data structures. In Section~\ref{section_background}, we give an overview of chromatic trees, the LLX and SCX primitives, and the implementation of the chromatic tree by Brown, Ellen, and Ruppert. In Section~\ref{section_modifications}, we make optimizations to the chromatic tree implementation to improve its amortized step complexity. Finally, in Section~\ref{section_amortized}, we prove the amortized step complexity of the modified chromatic tree. 

\section{Model}\label{section_model}

Throughout this paper, we use an asynchronous shared memory model. Shared memory consists of a collection of shared variables accessible by all processes in a system. Processes access or modify shared variables using primitive instructions that are performed atomically: \textsc{Write}$(r, value)$, which stores $value$ into the shared variable $r$, \textsc{Read}$(r)$, which returns the value stored in $r$, and compare-and-swap (CAS)\index{CAS}. CAS$(r, old, new)$ compares the value stored in shared variable $r$ with the value $old$. If the two values are the same, the value stored in $r$ is replaced with the value $new$ and \textsc{True} is returned; otherwise \textsc{False} is returned.

A \textit{configuration}\index{configuration} of a system consists of the values of all shared variables and the states of all processes. A \textit{step}\index{step} by a process either accesses or modifies a shared variable, and can also change the state of the process. An \textit{execution}\index{execution} is an alternating sequence of configurations and steps, starting with a configuration. A \textit{solo execution}\index{solo execution} from a configuration $C$ by a process $P$ is an execution starting from $C$ in which all steps are performed by $P$.

An \textit{abstract data type} is a collection of objects and operations that satisfies certain properties. A \textit{concurrent data structure} for the abstract data type provides representations of the objects in shared memory and algorithms for the processes to perform the operations. An operation on a data structure by a process becomes \textit{active}\index{active} when the process performs the first step of its algorithm. The operation becomes \textit{inactive}\index{inactive} after the last step of the algorithm is performed by the process. The \textit{execution interval}\index{execution interval} of the operation consists of all configurations in which it is active. In the initial configuration, the data structure is empty and there are no active operations. 

An execution $\alpha$ is \textit{linearizable}\index{linearizable} if one can assign linearization points to each completed operation and a subset of the uncompleted operations in $\alpha$ with the following properties. First, the linearization point of an operation is within its execution interval. Second, the return value of each operation in $\alpha$ must be the same as in the execution in which the same operations are performed atomically in order of their linearization points. A concurrent data structure is considered to be correct if all executions of its operations are linearizable.

A \textit{non-blocking}\index{non-blocking} implementation of a concurrent data structure guarantees that whenever there are active operations, one operation will eventually complete in a finite number of steps. However, the execution interval of any particular operation in an execution may be unbounded, provided other operations are completed. A \textit{wait-free}\index{wait-free} implementation of a concurrent data structure guarantees that every operation completes within a finite number of steps by the process that invoked the operation.

The \textit{step complexity} of an operation $op$ invoked by process $P$ is the number of steps performed by $P$ in the execution interval of $op$. The \textit{worst-case step complexity}\index{worst-case step complexity} of an operation is the maximum number of steps taken by a process to perform any instance of this operation in any execution. The \textit{amortized step complexity}\index{amortized step complexity} of a data structure is the maximum number of steps in any execution consisting of operations on the data structure, divided by the number operations invoked in the execution. One can determine an upper bound on the amortized step complexity by assigning an \textit{amortized cost}\index{amortized cost} to each operation, such that for all possible executions $\alpha$ on the data structure, the total number of steps taken in $\alpha$ is at most the sum of the amortized costs of the operations in $\alpha$. The amortized cost of a concurrent operation is often expressed as a function of contention. For a configuration $C$, we define its \textit{contention}\index{contention} $\dot{c}(C)$ to be the number of active operations in $C$. For an operation $op$ in an execution $\alpha$, we define its \textit{point contention}\index{point contention} $\dot{c}(op)$\index{$\dot{c}(op)$} to be the maximum number of active operations in a single configuration during the execution interval of $op$. Finally, for an execution $\alpha$, we define its point contention $\dot{c}(\alpha)$\index{$\dot{c}(\alpha)$} to be the maximum number of active operations in a single configuration during $\alpha$.

\section{Related Work}\label{section_related}

In this section, we give an overview of the amortized analyses of various non-blocking data structures. While there are many papers that give implementations of non-blocking data structures, there are relatively few that give the amortized step complexity of their implementations.

Fomitchev and Ruppert \cite{DBLP:conf/podc/FomitchevR04} modify an implementation of a singly linked list by Harris \cite{DBLP:conf/wdag/Harris01} with the addition of \textit{backlinks}. A backlink is a pointer to the predecessor of a node that has been \textit{marked} for deletion. Whenever an update operation fails, the backlinks can be followed until an unmarked node is reached, rather than restarting searches from the head of the linked list. They prove an amortized step complexity of $O(n(op) + \dot{c}(op))$ for their improved implementation, where $n(op)$ is the number of elements in the linked list when operation $op$ is invoked. The amortized analysis is done by charging each failed CAS and each backlink traversal performed by an operation $op$ to a successful CAS of a concurrent operation in the execution interval of $op$. Gibson and Gramoli \cite{DBLP:conf/wdag/GibsonG15} give a simplified implementation of a linked list and claim it has the same amortized step complexity.

Ellen, Fatourou, Ruppert, and van Breugel \cite{DBLP:conf/podc/EllenFRB10} give the first provably correct non-blocking implementation of an unbalanced binary search tree using CAS primitives. Ellen, Fatourou, Helga, and Ruppert \cite{DBLP:conf/podc/EllenFHR13} show that this implementation has poor amortized step complexity. They improve the implementation by allowing failed update attempts to backtrack to a suitable internal node in the tree, rather than restarting attempts from the root. They prove for their implementation, each operation $op$ has amortized cost $O(h(op) + \dot{c}(op))$, where $h(op)$ is the height of the binary search tree when $op$ is invoked. The amortized analysis is difficult because a single \textsc{Delete} operation can cause a sequence of overlapping \textsc{Delete} operations by other processes to fail. The accounting method is used to show how all failed attempts in an execution can be paid for. Chatterjee, Nguyen and Tsigas \cite{DBLP:conf/podc/ChatterjeeDT13} give a different implementation of an unbalanced binary search tree, and give a sketch that the amortized cost of each operation in their implementation is also $O(h(op) + \dot{c}(op))$.

Shafiei \cite{DBLP:conf/opodis/Shafiei15} gives an implementation of a non-blocking doubly-linked list. She proves an amortized cost of $O(\dot{c}(op))$ for update operations $op$. The analysis closely follows the analysis of the unbalanced binary search tree \cite{DBLP:conf/podc/EllenFHR13}, except modified to use the potential method.


\section{Background}\label{section_background}
\subsection{The Chromatic Tree Definition}

A chromatic tree is a data structure for a dynamic set of elements, each with a distinct key from an ordered universe. It supports the following three operations:
\begin{itemize}
	\item \textsc{Insert}$(k)$, which adds an element with key $k$ into the set and returns \textsc{True} if the set does not contain an element with key $k$; otherwise it returns \textsc{False},
	\item \textsc{Delete}$(k)$, which removes the element with key $k$ from the set and returns \textsc{True} if there is such an element; otherwise it returns \textsc{False}, and
	\item \textsc{Find}$(k)$, which returns \textsc{True} if there is an element in the set with key $k$ and \textsc{False} otherwise.
\end{itemize}

A chromatic tree is a generalization of a red-black tree with relaxed balance conditions. It is \textit{leaf-oriented}\index{leaf-oriented}, so every element in the dynamic set represented by the data structure corresponds to a leaf and all nodes have 0 or 2 children. Each node in a chromatic tree has a non-negative weight. We call a node \textit{red} if it has weight 0, \textit{black} if it has weight 1, and \textit{overweight} if it has weight greater than 1. The \textit{weighted level} of a node $x$ is the sum of node weights along the path from the root node to $x$. The balance conditions of red-black trees and chromatic trees as described in \cite{DBLP:conf/wads/BoyarFL95} are as follows.
\begin{definition}
	A \textit{red-black tree}\index{red-back tree} $T$ satisfies following balance conditions:
	\begin{enumerate}
		\item[B1.] The leaves of $T$ are black.
		\item[B2.] All leaves of $T$ have the same weighted level.
		\item[B3.] No path from $T$'s root to a leaf contains two consecutive red nodes.
		\item[B4.] $T$ has only red and black nodes.
	\end{enumerate}
\end{definition}

\begin{definition}\label{def_chromatic}
	A \textit{chromatic tree}\index{chromatic tree} $T$ satisfies the following balance conditions:
	\begin{enumerate}
		\item[C1.] The leaves of $T$ are not red.
		\item[C2.] All leaves of $T$ have the same weighted level.
	\end{enumerate}
\end{definition}

Let $x$ be a node that is in the chromatic tree with weight $x.w$. If $x$ and its parent both have weight 0, then a \textit{red-red violation}\index{red-red violation} occurs at $x$. If $x.w > 1$, then $x.w-1$ \textit{overweight violations}\index{overweight violation} occur at $x$. When a chromatic tree contains no violations, it satisfies the balance conditions of a red-black tree. The balance conditions of a red-black tree ensure its height is $O(\log n)$, where $n$ is the number of nodes in the tree.

An example of a chromatic tree is shown in Figure~\ref{fig_chromatic_tree}. The number inside each node is its key and the number to its right is its weight. The weighted level of each leaf is 3. There is an overweight violation at the leaf with key 1, and a red-red violation at the internal node with key 6.

\begin{figure}[t]
	\centering
	\begin{tikzpicture}[level/.style={sibling distance = 1.5cm, level distance = 0.6cm}] 
	\node [arn_w, label=right:{1}] {3}
	child{ node [arn_wx, label=right:{2}] {1}
	}
	child{ node [arn_w, label=right:{1}] {4}
		child{ node [arn_wx, label=right:{1}] {3}}
		child{ node [arn_w, label=right:{0}] {5}
			child{ node [arn_wx, label=right:{1}] {4}}
			child{ node [arn_w, label=right:{0}] {6}
				child{ node [arn_wx, label=right:{1}] {5}}
				child{ node [arn_wx, label=right:{1}] {6}}
			}
		}
	}; 
	\end{tikzpicture}
	\caption{A chromatic tree containing elements with keys 1, 3, 4, 5 and 6.}
	\label{fig_chromatic_tree}
\end{figure}
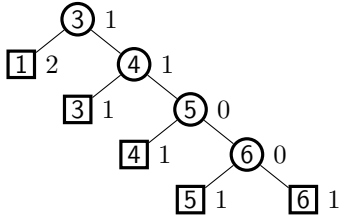

Rebalancing is done independently of insertions and deletions. This allows rebalancing to be done at times when there are fewer processes that want to perform operations. However, the chromatic tree may not be height balanced at all times. The following lemma proven by Boyar, Fagerberg, and Larsen \cite{DBLP:conf/wads/BoyarFL95} gives an upper bound on the total number of rebalancing transformations that can occur.  
\begin{theorem}\label{thm_rebal}
	If $i > 0$ insertions and $d$ deletions are performed on an initially empty chromatic tree, then at most $3i + d - 2$ rebalancing transformations can occur. 
\end{theorem}
In the non-blocking implementation by Brown, Ellen, and Ruppert \cite{DBLP:conf/ppopp/BrownER14}, rebalancing is done as a part of \textsc{Insert} and \textsc{Delete} operations. They proved that the imbalance in the chromatic tree depends on the number of active operations. 
\begin{lemma}\label{thm_height}
	If there are $c$ active \textsc{Insert} and \textsc{Delete} operations and the chromatic tree contains $n$ elements, then its height is $O(c + \log n)$.
\end{lemma}

\subsection{An Overview of LLX and SCX}
The LLX and SCX primitives, introduced in \cite{DBLP:conf/podc/BrownER13}, are themselves implemented using CAS primitives. It is therefore necessary to understand the step complexity of LLX and SCX to analyze the chromatic tree. LLX and SCX are generalizations of the LL and SC primitives. They operate on a set of Data-records, rather than a single shared variable. A Data-record\index{Data-record} is a collection of fields (denoted $m_1,\dots,m_y,i_1,\dots,i_z$) used to represent a natural piece of a data structure. The fields $m_1,\dots,m_y$ are \textit{mutable} and can be changed after the Data-record is created. The fields $i_1,\dots,i_z$ are \textit{immutable} and cannot be changed after the Data-record is created.

A successful linearized LLX$(r)$\index{LLX} on Data-record $r$ returns a snapshot of the fields of $r$. A process can only perform SCX$(V,R,fld,new)$ if it has previously performed a successful LLX on each Data-record in $V$ since it last performed SCX (or the beginning of the execution if it has never performed SCX). The last successful LLX by a particular process on each Data-record in $V$ is \textit{linked}\index{linked} to the SCX. A successful linearized SCX$(V,R,fld,new)$\index{SCX} \textit{finalizes} all Data-records in $R \subseteq V$. It also updates $fld$, one field of a single Data-record in $V$, to the value $new$. Finalized\index{finalized} Data-records can no longer be modified and are removed from the data structure. An unsuccessful SCX returns \textsc{False}, which only occurs when there is a concurrent SCX$(V',R',fld',new')$ such that $V \cap V' \neq \emptyset$. Likewise an LLX can return \textsc{Fail} if there is a concurrent SCX$(V',R',fld',new')$ such that $r \in V'$. An LLX$(r)$ returns \textsc{Finalized} if $r$ has been previously finalized by an SCX. In both of these cases, the LLX is unsuccessful.  


\begin{figure}[htbp]
	\begin{algorithmic}[1]
		\State \textsc{LLX}($r$) \text{by process $P$}
		\Indent
		\State $\triangleright$ Precondition: $r \neq \textsc{Nil}$ \Comment{order of lines 3-6 matters}
		\State $marked_1 := r.marked$ \label{ln:llx:mark1}
		\State $\mathit{rinfo} := r.\mathit{info}$ \label{ln:llx:rinfo}
		\State $state := \mathit{rinfo}.state$ \label{ln:llx:state}
		\State $marked_2 := r.marked$ \label{ln:llx:mark2}
		\If {$state = \text{Aborted}$ \textbf{or} $(state = \text{Committed}$ \textbf{and not} $marked_2$)} \Comment{if $r$ was not frozen at line 5} \label{ln:llx:main_if}
		\State \textbf{read} $r.m_1,\dots,r.m_y$ and record values in local variables $m_1,\dots,m_y$
		\If {$r.\mathit{info} = \mathit{rinfo}$} \Comment{if $r.\mathit{info}$ unchanged since line 4} \label{ln:llx:rinfo_change}
		\State Store $\langle r, \mathit{rinfo}, \langle m_1,\dots,m_y \rangle \rangle$ in $P$'s local table
		\State \Return $\langle m_1,\dots,m_y \rangle$
		\EndIf
		\EndIf 
		\If {$(\mathit{rinfo}.state = \text{Committed}$ \textbf{or} $(\mathit{rinfo}.state = \text{InProgress}$ \textbf{and} $\textsc{Help}(\mathit{rinfo})))$ \textbf{and} $marked_1$} \label{ln:llx:final_if}
		\State \Return \textsc{Finalized} 
		\Else
		\If {$\mathit{r.info}.state = \text{InProgress}$} \textsc{Help}$(\mathit{r.info})$ \EndIf \label{ln:llx:help}
		\State \Return \textsc{Fail}
		\EndIf
		\EndIndent
		\item[]
		\State SCX$(V,R,fld,new)$ by process $P$
		\Indent
		\State $\triangleright$ Preconditions: (1) for each $r$ in $V$, $P$ has performed an invocation $I_r$ of LLX$(r)$ linked to this SCX 
		\item[] \hspace{30.5mm} (2) $new$ is not the initial value of $fld$
		\item[] \hspace{30.5mm} (3) for each $r$ in $V$, no SCX$(V',R',fld,new)$ was linearized before $I_r$ was linearized
		\State Let $\mathit{infoFields}$ be a pointer to a newly created table in shared memory containing
		\item[] \hspace{1cm} for each $r$ in $V$, a copy of $r$'s $\mathit{info}$ value in $P$'s local table of LLX results
		\State Let $old$ be the value for $fld$ stored in $P$'s local table of LLX result
		\State \Return \textsc{Help}$(\text{pointer to new SCX-record}(V,R,fld,new,old,\text{InProgress},\textsc{False},\mathit{infoFields}))$
		\EndIndent
		\item[] 
		\State \textsc{Help}$(scxPtr)$
		\Indent
		\State $\triangleright$ Freeze all Data-records in $scxPtr.V$ to protect their mutable fields from being charged by other SCXs
		\State\textbf{for each} $r$ in $scxPtr.V$ enumerated in order \textbf{do}
		\Indent
		\State Let $\mathit{rinfo}$ be the pointer indexed by $r$ in $scxPtr.\mathit{infoFields}$
		\If {\textbf{not} CAS$(\mathit{r.info, rinfo, scxPtr})$} \Comment{Freezing CAS} \label{ln:help:freezing_cas}
		\State $\triangleright$ Could not freeze $r$ because it is frozen for another SCX
		\If {$\mathit{r.info \neq scxPtr}$}
		\If {$scxPtr.allFrozen = \textsc{True}$} \Comment{Frozen check step}
		\State $\triangleright$ the SCX has already completed successfully
		\State \Return $\textsc{True}$
		\Else
		\State $\triangleright$ Atomically unfreeze all Data-records frozen for this SCX
		\State $scxPtr.state := \text{Aborted}$ \Comment{Abort step}
		\State \Return $\textsc{False}$
		\EndIf
		\EndIf
		\EndIf
		\EndIndent
		\item[] 
		\State $\triangleright$ Finished freezing Data-records (Assert: $state \in \{\text{InProgress, Committed} \}$)
		\State $scxPtr.allFrozen := \textsc{True}$  \Comment{Frozen step}
		\State \textbf{for each} $r$ in $scxPtr.R$ \textbf{do} $r.marked := \textsc{True}$ \Comment{Mark step}
		\State CAS$(\mathit{scxPtr.fld, scxPtr.old, scxPtr.new})$ \Comment{Update CAS}
		\State $\triangleright$ Finalized all $r$ in $R$, and unfreeze $r$ in $V$ that are not in $R$
		\State $scxPtr.state = \text{Committed}$ \Comment{Commit step}
		\State \Return $\textsc{True}$
		\EndIndent
	\end{algorithmic}
	\caption{Pseudocode for LLX and SCX given in \cite{DBLP:conf/podc/BrownER13}.}
	\label{llx_and_scx_Code}
\end{figure}

Next we give a simplified overview of the implementation of LLX and SCX. The pseudocode for both LLX and SCX are given in Figure~\ref{llx_and_scx_Code}. In the implementation, every Data-record contains $marked$ bit, and $\mathit{info}$ pointer. The marked bit is used to finalize nodes. A node is only removed from the chromatic tree after it is marked. Once a marked bit is set to \textsc{True}, it cannot be changed back to \textsc{False}. An SCX-record\index{SCX-record} contains fields of information required for processes to finish an SCX on another process's behalf. It also contains a \textit{state} field whose value is either \text{InProgress}, \text{Aborted}, or \text{Committed}, and is initially \text{InProgress}.

Consider an SCX$(V,R,fld,new)$ performed by a process $P$. $P$ first creates a new SCX-record $U$ for this SCX. A Data-record $r$ is \textit{frozen}\index{frozen} for $U$ if $r.\mathit{info}$ points to $U$ and either $U.state = \text{InProgress}$, or $U.state = \text{Committed}$ and $r$ is marked. Freezing acts like a lock on a Data-record held by a particular operation. For each Data-record $v \in V$, $P$ attempts to freeze $v$ by setting $v.\textit{info}$ to point to $U$ using a \textit{freezing}\index{freezing CAS} CAS. This only succeeds if $v$ has not been frozen since $P$'s last LLX on $v$. If the freezing CAS fails and no other process has frozen $v$ for $U$, then $P$ performs an \textit{abort step} and returns \textsc{False}. An abort step atomically unfreezes all Data-records frozen for $U$ by setting $U$ to the \text{Aborted} state. 

Once all Data-records in $V$ are successfully frozen for $U$, the SCX can no longer fail. Then each Data-record in $R \subseteq V$ is \textit{marked} for removal and an \textit{update} CAS sets $fld$ to the value $new$. Finally, a \textit{commit step}\index{commit step} is performed, which atomically unfreezes all nodes in $V - R$ by setting $U$ to the \text{Committed} state. The Data-records in $R$ remain frozen. 

Consider an LLX$(r)$ by a process $P$. If $r$ is marked when the LLX is invoked, $P$ will help the SCX that marked $r$ (if it is not yet complete) and return \textsc{Finalized}. Otherwise $P$ attempts to take a snapshot of the fields of $r$. If $r$ is not frozen while $P$ reads each of its fields, the LLX is successful and returns the fields of $r$. If a concurrent SCX operation freezes $r$ sometime during $P$'s LLX, $P$ may help the concurrent SCX complete before returning \textsc{Fail}.

LLX and SCX guarantee non-blocking progress. If SCX is performed infinitely often, then an infinite number succeed. Much of our amortized analysis will require a mechanism to charge failed LLX and SCX of one operation to the other operation that caused it to fail.

\subsection{An Implementation of a Chromatic Tree using LLX and SCX}
We next give an overview of the non-blocking implementation of a chromatic tree using LLX and SCX primitives by Brown, Ellen, and Ruppert \cite{DBLP:conf/ppopp/BrownER14}. The \textit{nodes}\index{node} of the chromatic tree are represented as Data-records with fields as shown in Figure~\ref{data_records}. 
\begin{figure}[htbp!]
	\begin{algorithmic}[1]
		\State \textbf{Data-record} Node
		\Indent
		\State $\triangleright$ Fields used for the chromatic tree
		\State $\mathit{left}$, $right$ \Comment{(Mutable) Node pointers to left and right children}
		\State $k, w$ \Comment{(Immutable) key and weight}
		\State $\triangleright$ Fields used by LLX/SCX algorithm
		\State $\mathit{info}$ \Comment{pointer to an SCX-record}
		\State $marked$ \Comment{Boolean flag to mark node to be \textsc{Finalized}}
		\EndIndent
	\end{algorithmic}
	\caption{Data-record used to represent each node in the chromatic tree.}
	\label{data_records}
\end{figure}

\noindent The implementation uses additional $sentinel$\index{sentinel node} nodes with keys $\infty$ and weight 1. (See Figure~\ref{fig_sentinel_nodes}.) The sentinel nodes help prevent edge cases when the $root$ of a chromatic tree is involved in transformations. The key $\infty$ is greater than all keys in the chromatic tree, and is not the key of any element in the dynamic set. We do not consider a sentinel node which is a leaf to be in violation of balance condition (C2) from Definition~\ref{def_chromatic}. An empty chromatic tree is represented by one internal node $entry$\index{$entry$} with two children, all with keys $\infty$. A non-empty chromatic tree has two internal sentinel nodes at the parent and grandparent of the chromatic tree, each of which also has a sentinel node as its right child. A node $x$ is \textit{in the chromatic tree} in a configuration $C$ if $x$ is reachable from $entry$ in $C$.

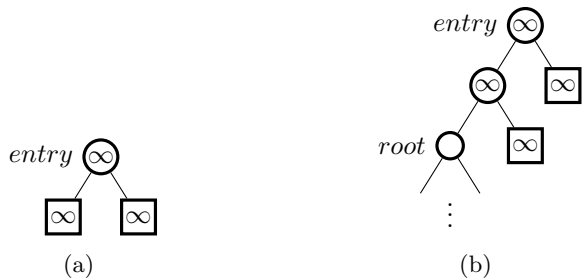
\begin{figure}[H]
	\centering
	\begin{subfigure}[b]{0.3\textwidth}\centering
		\begin{tikzpicture}[-,>=stealth',
		level distance=0.9cm,
		level 1/.style={sibling distance=1cm, level distance = 0.8cm}
		] 
		\node [arn_w, minimum size=0.45cm, text width=10pt, label=left:{$entry$}]{ $\infty$ } 
		child{node [arn_wx, minimum size=0.45cm, text width=10pt, label=right:{}] { $\infty$ }  }
		child{node [arn_wx, minimum size=0.45cm, text width=10pt] {  $\infty$ }  
		}; 
		\end{tikzpicture}
		\caption{}
	\end{subfigure}%
	\begin{subfigure}[b]{0.3\textwidth}\centering
		\begin{tikzpicture}[-,>=stealth',
		level distance=0.9cm,
		level 1/.style={sibling distance=1cm, level distance = 0.8cm},
		] 
		\node [arn_w, minimum size=0.45cm, text width=10pt, label=left:{$entry$}]{ $\infty$ } 
		child{node [arn_w, minimum size=0.45cm, text width=10pt, label=right:{}] { $\infty$ }  
			child{ node [arn_w, label=left:{$root$}] { } 
				child{ node [arn_w, draw=none, label={[xshift=0.5cm, yshift=-0.6cm]\vdots}] { } }
				child{ node [arn_w, draw=none, label=above:{}] { } }
			}
			child{ node [arn_wx, minimum size=0.45cm, text width=10pt, label=right:{}] { $\infty$  } }
		}
		child{node [arn_wx, minimum size=0.45cm, text width=10pt] {  $\infty$ }
		}; 
		\end{tikzpicture}
		\caption{}
	\end{subfigure}
	\caption{The location of sentinel nodes for (a) an empty tree and (b) a non-empty tree. Internal nodes are denoted with circles, while leaves are denoted with squares.}
	\label{fig_sentinel_nodes}
\end{figure}

We give an overview of the algorithm to \textsc{Insert} an element with key $k$ into a chromatic tree. The \textsc{Insert} operation is divided into an \textit{update phase}\index{update phase} and a \textit{cleanup phase}\index{cleanup phase}. An update phase repeatedly performs update attempts until one is successful. An update attempt begins by invoking \textsc{Search}$(k)$. It searches for key $k$ starting at $entry$, until it reaches some leaf $l$. Let $p$ denote the node visited immediately before reaching $l$. If $l.key = k$, then the key is already in the chromatic tree, and so \textsc{Insert} returns \textsc{False}. Otherwise \textsc{TryInsert}$(p,l)$ is invoked, which attempts to perform an LLX on the nodes $p$ and $l$, followed by a single SCX to apply the \textsc{Insert} transformation (as shown in Figure~\ref{fig_transformations}). The invocation of \textsc{TryInsert} fails if any LLX or SCX fails, in which case a new update attempt begins. Once a successful \textsc{TryInsert} is performed, the operation enters its cleanup phase if a violation was created; otherwise the \textsc{Insert} terminates. 

During the cleanup phase, a search for $k$ is again performed starting at $entry$. If a violation is encountered at a node $v$ during this search, \textsc{TryRebalance}$(ggp,gp,p,v)$ attempts to remove the violation, where $p$, $gp$, and $ggp$ are the last 3 nodes visited on the search to $v$. Whether or not \textsc{TryRebalance} succeeds, the cleanup restarts its search at $entry$. The cleanup phase terminates when the search from $entry$ reaches a leaf without encountering a violation. 

In any configuration, the \textit{search path for a key $k$ from a node $r$} \index{search path} is the path that is followed when searching for $k$ starting at $r$. If $r = entry$, this is called the \textit{search path for $k$}. The rebalancing transformations designed by Brown, Ellen, and Ruppert \cite{DBLP:conf/ppopp/BrownER14} have the following key property. 
\begin{lemma}\normalfont\label{viol}
	If a violation is on the search path for key $k$ before a rebalancing transformation, then the violation is still on the search path for $k$ after the rebalancing transformation, or it has been removed. 
\end{lemma}
\noindent This guarantees that when a process reaches a leaf node in its cleanup phase, the violation created during its update phase has been removed.

The algorithm to \textsc{Delete} an element is similar to the algorithm for \textsc{Insert}, except with \textsc{TryInsert}$(p,l)$ replaced with \textsc{TryDelete}$(gp,p,l)$. \textsc{TryDelete}$(gp,p,l)$ performs LLX on the nodes $gp$, $p$, $l$, and the sibling of $l$ before performing an SCX to complete the \textsc{Delete} transformation (as shown in Figure~\ref{fig_transformations}). Likewise, the routine \textsc{TryRebalance} determines a suitable rebalancing transformation to perform by a series of LLXs, followed by a single SCX. \textsc{TryRebalance} has an extra failure condition called a \textsc{Nil} check\index{Nil check}, which is used to guarantee \textsc{Nil} pointers are not traversed. These checks may fail when a nearby node is modified by a concurrent operation.

Each chromatic tree transformation is outlined in Figure~\ref{fig_transformations}. Each transformation has a symmetric version with left and right interchanged. The rebalancing transformations preserve the weighted level of all leaves. Notice that all transformations can be performed by using SCX$(V,R,fld,n)$ to update a single child pointer of the node labeled $u$ to point to the new node $n$. We use the notation \un{x} to denote any child of the node \un{}, \un{l} to denote the left child of \un{}, and \un{r} to denote the right child of \un{}. Any weight restrictions are to the right of each node. Each node in the set $V$ of this SCX is shaded in gray. Nodes marked with $\times$ are the nodes in the set $R$ and are finalized, while the nodes marked with $+$ are newly created nodes. Each transformation is marked with a \textit{center}\index{center} node, which is highlighted with a bold outline. The only exception is BLK, whose center node for BLK is the first node in the sequence (\un{xll}, \un{xlr}, \un{xrl}, \un{xrr}) with weight 0. A rebalancing transformation is \textit{centered} at a violation if the violation occurs at the center node of the transformation. The center node will be used to identify an update operation that can be charged for performing each transformation. 

A transformation may create, remove, or move violations. The first two columns of Figure~\ref{violations_move} describe how violations are \textit{created}\index{create violation} or \textit{removed}\index{remove violation} in each transformation. Any violation in $R$ that is not removed is \textit{moved}\index{move violation} to one of the newly added nodes, as described in the third column. Note that for RB1 and BLK, the red-red violation at the center node is removed by changing its parent to a node with weight 1. Also note that for W7, an overweight from \un{xr} is moved to \nn{} if $\un{x}.w > 0$ (rather than saying 1 violation is created at \nn{}, and 1 is removed from each of \un{xl} and \un{xr}).

All rebalancing transformations remove 1 violation at the center node, except for BLK and PUSH. When BLK and PUSH do not remove a violation at the center node, they instead \textit{elevate}\index{elevate violation} the violation, moving it to a new node higher in the tree. In particular, if $\un{x}.w = 1$ and $\un{}.w = 0$, then a BLK transformation centered at \un{xxx} \textit{elevates} the red-red violation at \un{xxx} to the node \nn{}. If $\un{x}.w > 0$, then a PUSH transformation \textit{elevates} 1 overweight violation at \un{xl} to the node \nn{}. 

For every violation $v$, let the \textit{creator} of $v$ be the process that invoked the SCX that created $v$ during a \textsc{TryInsert} or \textsc{TryDelete}. Whenever a violation $v$ is moved or elevated to a new node, the creator of the violation does not change. 

\begin{figure}[phbt]
	\begin{subfigure}{\textwidth}
		\centering
		\begin{tikzpicture}[-,>=stealth',level/.style={sibling distance = 4cm/#1, level distance = 0.6cm}] 
		\node [arn_gs, label=left:{$\un{}$}] { }
		child{ node [arn_gxs, ultra thick, label=left:{$\un{x}$}] {$\large{\times}$} }; 
		\end{tikzpicture}
		\quad
		{\raisebox{1\height}{\scalebox{1}{$\stackrel{\textsc{Insert}}{\xrightarrow{\makebox[0.5cm]{}}}$}}}
		\begin{tikzpicture}[-,>=stealth',
		level distance = 0.6cm,
		level 1/.style={sibling distance=1cm, level distance = 0.6cm},
		] 
		\node [arn_gs, label=left:{$\un{}$}] {} 
		child{node [arn_ws, label=left:{$n$}, label=right:{$\un{x}.w$-1}] {\large{+}}  
			child{ node [arn_wxs, label=left:{}, label=right:{1}] {\large{+}} }
			child{ node [arn_wxs, label=left:{}, label=right:{1}] {\large{+}} }
		}; 
		\end{tikzpicture}
	\end{subfigure}

	\vspace*{0.2cm}
	\begin{subfigure}{0.5\textwidth}\centering
		\begin{tikzpicture}[-,>=stealth',
		level distance = 0.6cm,
		level 1/.style={sibling distance=1.2cm, level distance = 0.6cm},
		] 
		\node [arn_gs, label=left:{$\un{}$}]{ } 
		child{node [arn_gs, label=left:{$\un{x}$}] {$\large{\times}$}  
			child{ node [arn_gxs, ultra thick, label=left:{$\un{xl}$}] {$\large{\times}$} }
			child{ node [arn_gs, label=left:{$\un{xr}$}] {$\large{\times}$} 
				child{ node [arn_ws, label=below:{$\un{xrl}$}] { } }
				child{ node [arn_ws, label=below:{$\un{xrr}$}] { } }
			}
		}; 
		\end{tikzpicture}
		{\raisebox{2.7\height}{\scalebox{1}{$\stackrel{\textsc{Delete}}{\xrightarrow{\makebox[0.5cm]{}}}$}}}
		\begin{tikzpicture}[-,>=stealth',
		level distance = 0.6cm,
		level 1/.style={sibling distance=1cm, level distance = 0.6cm},
		] 
		\node [arn_gs, label=left:{$\un{}$}]{ } 
		child{node [arn_ws, label=left:{$n$}, label=right:{$\un{x}.w$+$\un{xr}.w$}] {\large{+}}  
			child{ node [arn_ws, label=below:{$\un{xrl}$}] {} }
			child{ node [arn_ws, label=below:{$\un{xrr}$}] {} }
		}; 
		\end{tikzpicture}
	\end{subfigure}
	\begin{subfigure}{0.5\textwidth}\centering
		\begin{tikzpicture}[-,>=stealth',
		level distance = 0.6cm,
		level 2/.style={sibling distance=1.6cm, level distance = 0.6cm},
		level 3/.style={sibling distance=0.8cm, level distance = 0.6cm},
		] 
		\node [arn_gs, label=left:{$\un{}$}]{ } 
		child{node [arn_gs, label=left:{$\un{x}$}, label=right:{$>$0}] {$\large{\times}$}  
			child{ node [arn_gs, label=right:{$0$}] {$\large{\times}$} 
				child{ node [arn_ws, ultra thick, label=below:{$\un{xll}$}, label={[shift={(0.9,-1.3)}]at least 1 has weight 0}] { } }
				child{ node [arn_ws, ultra thick, label=below:{$\un{xlr}$}] { } }
			}
			child{ node [arn_gs, label=right:{$0$}] {$\large{\times}$} 
				child{ node [arn_ws, ultra thick, label=below:{$\un{xrl}$}] { }}
				child{ node [arn_ws, ultra thick, label=below:{$\un{xrr}$}] { } }
			}
		}; 
		\end{tikzpicture}
		{\raisebox{2.7\height}{\scalebox{1}{$\stackrel{\textsc{Blk}}{\xrightarrow{\makebox[0.5cm]{}}}$}}}
		\begin{tikzpicture}[-,>=stealth',
		level distance = 0.6cm,
		level 2/.style={sibling distance=1.6cm, level distance = 0.6cm},
		level 3/.style={sibling distance=0.8cm, level distance = 0.6cm},
		] 
		\node [arn_gs, label=left:{$\un{}$}]{ } 
		child{node [arn_ws, label=left:{$n$}, label=right:{$\un{x}.w$-$1$}] {\large{+}}  
			child{ node [arn_ws, label=right:{$1$}] {\large{+}} 
				child{ node [arn_ws, label=below:{$\un{xll}$}] { } }
				child{ node [arn_ws, label=below:{$\un{xlr}$}] { } }
			}
			child{ node [arn_ws, label=right:{$1$}] {\large{+}} 
				child{ node [arn_ws, label=below:{$\un{xrl}$}] { }}
				child{ node [arn_ws, label=below:{$\un{xrr}$}] { } }
			}
		}; 
		\end{tikzpicture}
	\end{subfigure}
	\vspace*{1mm}
	\begin{subfigure}{0.5\textwidth}\centering
		\begin{tikzpicture}[-,>=stealth',
		level distance = 0.6cm,
		level 1/.style={sibling distance=1cm, level distance = 0.6cm},
		] 
		\node [arn_gs, label=left:{$\un{}$}]{ } 
		child{node [arn_gs, label=left:{$\un{x}$}, label=right:{$>$0}] {$\large{\times}$}  
			child{ node [arn_gs, label=left:{$\un{xl}$}, label=right:{$0$}] {$\large{\times}$}
				child{ node [arn_ws, ultra thick, label=right:{$0$}, label=below:{$\un{xll}$}] { } }
				child{ node [arn_ws, label=below:{$\un{xlr}$}] { } }
			}
			child{ node [arn_ws, label=below:{$\un{xr}$}, label=right:{$>$0}] { } }
		}; 
		\end{tikzpicture}
		{\raisebox{2.7\height}{\scalebox{1}{$\stackrel{\textsc{RB1}}{\xrightarrow{\makebox[0.5cm]{}}}$}}}
		\begin{tikzpicture}[-,>=stealth',
		level distance = 0.6cm,
		level 1/.style={sibling distance=1cm, level distance = 0.6cm},
		] 
		\node [arn_gs, label=left:{$\un{}$}]{ } 
		child{node [arn_ws, label=left:{$n$}, label=right:{$\un{x}.w$}] {\large{+}}  
			child{ node [arn_ws, label=below:{$\un{xll}$}] { } }
			child{ node [arn_ws, label=right:{$0$}] {\large{+}} 
				child{ node [arn_ws, label=below:{$\un{xlr}$}] { } }
				child{ node [arn_ws, label=below:{$\un{xr}$}] { } }
			}
		}; 
		\end{tikzpicture}
	\end{subfigure}
	\begin{subfigure}{0.5\textwidth}\centering
		\begin{tikzpicture}[-,>=stealth',
		level distance = 0.6cm,
		level 1/.style={sibling distance=1.5cm, level distance = 0.6cm},
		level 2/.style={sibling distance=1.4cm, level distance = 0.6cm},
		level 3/.style={sibling distance=1.4cm, level distance = 0.6cm},
		level 4/.style={sibling distance=1cm, level distance = 0.6cm},
		] 
		\node [arn_gs, label=left:{$\un{}$}]{ } 
		child{node [arn_gs, label=left:{$\un{x}$}, label=right:{$>$0}] {$\large{\times}$}  
			child{ node [arn_gs, label=left:{$\un{xl}$}, label=right:{$0$}] {$\large{\times}$}
				child{ node [arn_ws, label=below:{$\un{xll}$}, label=left:{}] { } }
				child{ node [arn_gs, ultra thick, label=left:{$\un{xlr}$}, label=right:{$0$}] {$\large{\times}$}
					child{ node [arn_ws, label=below:{$\un{xlrl}$}] { } }
					child{ node [arn_ws, label=below:{$\un{xlrr}$}] { } }
				}
			}
			child{ node [arn_ws, label=below:{$\un{xr}$}, label=right:{$>0$}] { } }
		}; 
		\end{tikzpicture}
		\quad
		{\raisebox{2.7\height}{\scalebox{1}{$\stackrel{\textsc{RB2}}{\xrightarrow{\makebox[0.5cm]{}}}$}}}
		\begin{tikzpicture}[-,>=stealth',
		level distance = 0.6cm,
		level 2/.style={sibling distance=1.6cm, level distance = 0.6cm},
		level 3/.style={sibling distance=0.8cm, level distance = 0.6cm},
		] 
		\node [arn_gs, label=left:{$\un{}$}]{ } 
		child{node [arn_ws, label=left:{$n$}, label=right:{$\un{x}.w$}] {\large{+}}  
			child{ node [arn_ws, label=right:{$0$}] {\large{+}} 
				child{ node [arn_ws, label=below:{$\un{xll}$}] { } }
				child{ node [arn_ws, label=below:{$\un{xlrl}$}] { } }
			}
			child{ node [arn_ws, label=right:{$0$}] {\large{+}} 
				child{ node [arn_ws, label=below:{$\un{xlrr}$}] { } }
				child{ node [arn_ws, label=below:{$\un{xr}$}] { } }
			}
		}; 
		\end{tikzpicture}
	\end{subfigure}
	\vspace*{1mm}
	\begin{subfigure}{0.52\textwidth}\centering
		\begin{tikzpicture}[-,>=stealth',
		level distance = 0.6cm,
		level 2/.style={sibling distance=1.8cm, level distance = 0.6cm},
		level 3/.style={sibling distance=1cm, level distance = 0.6cm},
		] 
		\node [arn_gs, label=left:{$\un{}$}]{ } 
		child{node [arn_gs, label=left:{$\un{x}$}, label=right:{$>$0}] {$\large{\times}$}  
			child{ node [arn_gs, ultra thick, label=left:{$\un{xl}$}, label=right:{$>$1}] {$\large{\times}$} 
				child{ node [arn_ws, label=below:{$\un{xll}$}] { } }
				child{ node [arn_ws, label=below:{$\un{xlr}$}] { } }
			}
			child{ node [arn_gs, label=right:{$0$}] {$\large{\times}$} 
				child{ node [arn_gs, label=right:{$>$1}] {$\large{\times}$}
					child{ node [arn_ws, label=below:{$\un{xrll}$}] { } }
					child{ node [arn_ws, label=below:{$\un{xrlr}$}] { } }
				}
				child{ node [arn_ws, label=below:{$\un{xrr}$}] { } }
			}
		}; 
		\end{tikzpicture}
		{\raisebox{2.7\height}{\scalebox{1}{$\stackrel{\textsc{W1}}{\xrightarrow{\makebox[0.5cm]{}}}$}}}
		\begin{tikzpicture}[-,>=stealth',
		level distance = 0.6cm,
		level 2/.style={sibling distance=1cm, level distance = 0.6cm},
		level 3/.style={sibling distance=1.8cm, level distance = 0.6cm},
		level 4/.style={sibling distance=0.8cm, level distance = 0.6cm},
		] 
		\node [arn_gs, label=left:{$\un{}$}]{ } 
		child{node [arn_ws, label=left:{$n$}, label=right:{$\un{x}.w$}] {\large{+}}  
			child{ node [arn_ws, label=right:{$1$}] {\large{+}} 
				child{ node [arn_ws, label=right:{$\un{xl}.w$-1}] {\large{+}}
					child{ node [arn_ws, label=below:{$\un{xll}$}] { } }
					child{ node [arn_ws, label=below:{$\un{xlr}$}] { } }
				}
				child{ node [arn_ws, label=right:{$\un{xrl}.w$-1}] {\large{+}} 
					child{ node [arn_ws, label=below:{$\un{xrll}$}] { } }
					child{ node [arn_ws, label=below:{$\un{xrlr}$}] { } }
				}
			}
			child{ node [arn_ws, label=right:{$\un{xrr}$}] { } }
		}; 
		\end{tikzpicture}
	\end{subfigure}
	\begin{subfigure}{0.5\textwidth}\centering
		\begin{tikzpicture}[-,>=stealth',
		level distance = 0.6cm,
		level 2/.style={sibling distance=1.6cm, level distance = 0.6cm},
		level 3/.style={sibling distance=0.8cm, level distance = 0.6cm},
		level 4/.style={sibling distance=1cm, level distance = 0.8cm},
		] 
		\node [arn_gs, label=left:{$\un{}$}]{} 
		child{node [arn_gs, label=left:{$\un{x}$}, label=right:{$>$0}] {$\large{\times}$}  
			child{ node [arn_gs, ultra thick, label=left:{$\un{xl}$}, label=right:{$>1$}] {$\large{\times}$} 
				child{ node [arn_ws, label=below:{$\un{xll}$}] { } }
				child{ node [arn_ws, label=below:{$\un{xlr}$}] { } }
			}
			child{ node [arn_gs, label=right:{$0$}] {$\large{\times}$} 
				child{ node [arn_gs, label=right:{$1$}] {$\large{\times}$}
					child{ node [arn_ws, label=below:{$\un{xrll}$}, label=right:{$>$0}] { } }
					child{ node [arn_ws, label=below:{$\un{xrlr}$}, label=right:{$>$0}] { } }
				}
				child{ node [arn_ws, label=below:{$\un{xrr}$}] { } }
			}
		}; 
		\end{tikzpicture}
		{\raisebox{2.7\height}{\scalebox{1}{$\stackrel{\textsc{W2}}{\xrightarrow{\makebox[0.5cm]{}}}$}}}
		\begin{tikzpicture}[-,>=stealth',
		level distance = 0.6cm,
		level 2/.style={sibling distance=1.6cm, level distance = 0.6cm},
		level 3/.style={sibling distance=1.8cm, level distance = 0.6cm},
		level 4/.style={sibling distance=1cm, level distance = 0.6cm},
		] 
		\node [arn_gs, label=left:{$\un{}$}]{ } 
		child{node [arn_ws, label=left:{$n$}, label=right:{$\un{x}.w$}] {\large{+}}  
			child{ node [arn_ws, label=right:{$1$}] {\large{+}} 
				child{ node [arn_ws, label=right:{$\un{xl}.w$-1}] {\large{+}}
					child{ node [arn_ws, label=below:{$\un{xll}$}] { } }
					child{ node [arn_ws, label=below:{$\un{xlr}$}] { } }
				}
				child{ node [arn_ws, label=right:{$0$}] {\large{+}} 
					child{ node [arn_ws, label=below:{$\un{xrll}$}] { } }
					child{ node [arn_ws, label=below:{$\un{xrlr}$}] { } }
				}
			}
			child{ node [arn_ws, label=below:{$\un{xrr}$}] { } }
		}; 
		\end{tikzpicture}
	\end{subfigure}
	\vspace*{2mm}
	\begin{subfigure}{0.53\textwidth}\centering
		\begin{tikzpicture}[-,>=stealth',
		level distance = 0.6cm,
		level 2/.style={sibling distance=1.6cm, level distance = 0.6cm},
		level 3/.style={sibling distance=1cm, level distance = 0.55cm},
		level 4/.style={sibling distance=1cm, level distance = 0.75cm},
		level 5/.style={sibling distance=1cm, level distance = 0.5cm},
		] 
		\node [arn_gs, label=left:{$\un{}$}]{ } 
		child{node [arn_gs, label=left:{$\un{x}$}, label=right:{$>$0}] {$\large{\times}$}  
			child{ node [arn_gs, ultra thick, label=left:{$\un{xl}$}, label=right:{$>$1}] {$\large{\times}$} 
				child{ node [arn_ws, label=below:{$\un{xll}$}] { } }
				child{ node [arn_ws, label=below:{$\un{xlr}$}] { } }
			}
			child{ node [arn_gs, label=right:{$0$}] {$\large{\times}$} 
				child{ node [arn_gs, label=right:{$1$}] {$\large{\times}$}
					child{ node [arn_gs, label=right:{$0$}] {$\large{\times}$}
						child{ node [arn_ws, label=below:{$\un{xrlll}$}] { } }
						child{ node [arn_ws, label=below:{$\un{xrllr}$}] { } } 
					}
					child{ node [arn_ws, label=below:{$\un{xrlr}$}] { } }
				}
				child{ node [arn_ws, label=below:{$\un{xrr}$}] { } }
			}
		}; 
		\end{tikzpicture}
		{\raisebox{3\height}{\scalebox{1}{$\stackrel{\textsc{W3}}{\xrightarrow{\makebox[0.5cm]{}}}$}}}
		\begin{tikzpicture}[-,>=stealth',
		level distance = 0.6cm,
		level 3/.style={level distance = 0.8cm},
		level 4/.style={level distance = 0.6cm},
		] 
		\node [arn_gs, label=left:{$\un{}$}]{ } 
		child[sd=1.6cm]{node [arn_ws, label=left:{$n$}, label=right:{$\un{x}.w$}] {\large{+}}  
			child[sd=0.9cm]{ node [arn_ws, label=right:{$0$}] {\large{+}} 
				child[sd=1.8cm]{ node [arn_ws, label=right:{ $1$ }] {\large{+}}
					child[sd=2cm]{ node [arn_ws, label=right:{$\un{xl}.w$-1}] {\large{+}}
						child[sd=1cm]{ node [arn_ws, label=below:{$\un{xll}$}] { } }
						child[sd=1cm]{ node [arn_ws, label=below:{$\un{xlr}$}] { } } 
					}
					child[sd=1.6cm]{ node [arn_ws, label={[shift={(-0.48,-0.74)}]$\un{xrlll}$}] { } }
				}
				child[sd=1.8cm]{ node [arn_ws, label=right:{1}] {\large{+}} 
					child[sd=0.8cm]{ node [arn_ws, label=below:{$\un{xrllr}$}] { } }
					child[sd=0.8cm]{ node [arn_ws, label=below:{$\un{xrlr}$}] { } }
				}
			}
			child[sd=0.9cm]{ node [arn_ws, label=below:{$\un{xrr}$}] { } }
		}; 
		\end{tikzpicture}
	\end{subfigure}
	\begin{subfigure}{0.51\textwidth}\centering
		\begin{tikzpicture}[-,>=stealth',
		level distance = 0.6cm,
		level 2/.style={sibling distance=1.6cm, level distance = 0.6cm},
		level 3/.style={sibling distance=1cm, level distance = 0.6cm},
		level 4/.style={sibling distance=1cm, level distance = 0.8cm},
		level 5/.style={sibling distance=1cm, level distance = 0.6cm},
		] 
		\node [arn_gs, label=left:{$\un{}$}]{ } 
		child{node [arn_gs, label=left:{$\un{x}$}, label=right:{$>$0}] {$\large{\times}$}  
			child{ node [arn_gs, ultra thick, label=left:{$\un{xl}$}, label=right:{$>$1}] {$\large{\times}$} 
				child{ node [arn_ws, label=below:{$\un{xll}$}] { } }
				child{ node [arn_ws, label=below:{$\un{xlr}$}] { } }
			}
			child{ node [arn_gs, label=right:{$0$}] {$\large{\times}$} 
				child{ node [arn_gs, label=right:{$1$}] {$\large{\times}$}
					child{ node [arn_ws, label=below:{$\un{xrll}$}] { } }
					child{ node [arn_gs, label=right:{$0$}] {$\large{\times}$}
						child{ node [arn_ws, label=below:{$\un{xrlrl}$}] { } }
						child{ node [arn_ws, label=below:{$\un{xrlrr}$}] { } }
					}
				}
				child{ node [arn_ws, label=below:{$\un{xrr}$}] { } }
			}
		}; 
		\end{tikzpicture}
		{\raisebox{3\height}{\scalebox{1}{$\stackrel{\textsc{W4}}{\xrightarrow{\makebox[0.5cm]{}}}$}}}
		\begin{tikzpicture}[-,>=stealth',
		level distance = 0.6cm,
		level 4/.style={level distance = 0.7cm},
		] 
		\node [arn_gs, label=left:{$\un{}$}]{ } 
		child{node [arn_ws, label=left:{$n$}, label=right:{$\un{x}.w$}] {\large{+}}  
			child[sd=1.8cm]{ node [arn_ws, label=right:{$1$}] {\large{+}} 
				child[sd=1.7cm]{ node [arn_ws, label=right:{$\un{xl}.w$-1}] {\large{+}}
					child[sd=0.8cm]{ node [arn_ws, label=below:{$\un{xll}$}] { } }
					child[sd=0.8cm]{ node [arn_ws, label=below:{$\un{xlr}$}] { } }
				}
				child[sd=1.7cm]{ node [arn_ws, label=below:{$\un{xrll}$}] { } }
			}
			child[sd=1.8cm]{ node [arn_ws, label=right:{$0$}] {\large{+}}
				child[sd=0.8cm]{ node [arn_ws, label=right:{$1$}] {\large{+}}
					child[sd=0.8cm]{ node [arn_ws, label=below:{$\un{xrlrl}$}] { } }
					child[sd=0.8cm]{ node [arn_ws, label=below:{$\un{xrlrr}$}] { } }
				}
				child[sd=0.8cm]{ node [arn_ws, label=below:{$\un{xrr}$}] { } }
			}
		}; 
		\end{tikzpicture}
	\end{subfigure}
	\vspace*{2mm}
	\begin{subfigure}{0.53\textwidth}\centering
		\begin{tikzpicture}[-,>=stealth',
		level distance = 0.6cm,
		level 2/.style={sibling distance=1.6cm, level distance = 0.6cm},
		level 3/.style={sibling distance=0.9cm, level distance = 0.6cm},
		] 
		\node [arn_gs, label=left:{$\un{}$}]{ } 
		child{node [arn_gs, label=left:{$\un{x}$}] {$\large{\times}$}  
			child{ node [arn_gs, ultra thick, label=right:{$>$1}] {$\large{\times}$} 
				child{ node [arn_ws, label=below:{$\un{xll}$}] { } }
				child{ node [arn_ws, label=below:{$\un{xlr}$}] { } }
			}
			child{ node [arn_gs, label=right:{$1$}] {$\large{\times}$} 
				child{ node [arn_ws, label=below:{$\un{xrl}$}] { }}
				child{ node [arn_gs, label=right:{$0$}] {$\large{\times}$}
					child{ node [arn_ws, label=below:{$\un{xrrl}$}] { } }
					child{ node [arn_ws, label=below:{$\un{xrrr}$}] { } }
				}
			}
		}; 
		\end{tikzpicture}
		{\raisebox{3\height}{\scalebox{1}{$\stackrel{\textsc{W5}}{\xrightarrow{\makebox[0.5cm]{}}}$}}}
		\begin{tikzpicture}[-,>=stealth',
		level distance = 0.6cm,
		] 
		\node [arn_gs, label=left:{$\un{}$}]{ } 
		child{node [arn_ws, label=left:{$n$}, label=right:{$\un{x}.w$}] {\large{+}}  
			child[sd=1.8cm]{ node [arn_ws, label=right:{$1$}] {\large{+}} 
				child[sd=1.8cm]{ node [arn_ws, label=right:{$\un{xl}.w$-1}] {\large{+}}
					child[sd=0.8cm]{ node [arn_ws, label=below:{$\un{xll}$}] { } }
					child[sd=0.8cm]{ node [arn_ws, label=below:{$\un{xlr}$}] { } }
				}
				child[sd=1.8cm]{ node [arn_ws, label={[shift={(-0.3,-0.74)}]$\un{xrl}$}] { } }
			}
			child[sd=1.8cm]{ node [arn_ws, label=right:{$1$}] {\large{+}}
				child[sd=0.8cm]{ node [arn_ws, label=below:{$\un{xrrl}$}] { }}
				child[sd=0.8cm]{ node [arn_ws, label=below:{$\un{xrrr}$}] { } }
			}
		}; 
		\end{tikzpicture}
	\end{subfigure}
	\begin{subfigure}{0.5\textwidth}\centering
		\begin{tikzpicture}[-,>=stealth',
		level distance = 0.6cm,
		level 2/.style={sibling distance=1.6cm, level distance = 0.6cm},
		level 3/.style={sibling distance=0.8cm, level distance = 0.6cm},
		] 
		\node [arn_gs, label=left:{$\un{}$}]{ } 
		child{node [arn_gs, label=left:{$\un{x}$}] {$\large{\times}$}  
			child{ node [arn_gs, ultra thick, label=right:{$>$1}] {$\large{\times}$} 
				child{ node [arn_ws, label=below:{$\un{xll}$}] { } }
				child{ node [arn_ws, label=below:{$\un{xlr}$}] { } }
			}
			child{ node [arn_gs, label=right:{$1$}] {$\large{\times}$} 
				child{ node [arn_gs, label=right:{$0$}] {$\large{\times}$}
					child{ node [arn_ws, label=below:{$\un{xrll}$}] { } }
					child{ node [arn_ws, label=below:{$\un{xrlr}$}] { } }
				}
				child{ node [arn_ws, label=below:{$\un{xrr}$}] { } }
			}
		}; 
		\end{tikzpicture}
		{\raisebox{3\height}{\scalebox{1}{$\stackrel{\textsc{W6}}{\xrightarrow{\makebox[0.5cm]{}}}$}}}
		\begin{tikzpicture}[-,>=stealth',
		level distance = 0.6cm,
		] 
		\node [arn_gs, label=left:{$\un{}$}]{ } 
		child{node [arn_ws, label=left:{$n$}, label=right:{$\un{x}.w$}] {\large{+}}  
			child[sd=1.8cm]{ node [arn_ws, label=right:{$1$}] {\large{+}} 
				child[sd=1.8cm]{ node [arn_ws, label=right:{$\un{xl}.w$-1}] {\large{+}}
					child[sd=0.8cm]{ node [arn_ws, label=below:{$\un{xll}$}] { } }
					child[sd=0.8cm]{ node [arn_ws, label=below:{$\un{xlr}$}] { } }
				}
				child[sd=1.8cm]{ node [arn_ws, label={[shift={(-0.3,-0.74)}]$\un{xrll}$}] { } }
			}
			child[sd=1.8cm]{ node [arn_ws, label=right:{$1$}] {\large{+}}
				child[sd=0.8cm]{ node [arn_ws, label=below:{$\un{xrlr}$}] { }}
				child[sd=0.8cm]{ node [arn_ws, label=below:{$\un{xrr}$}] { } }
			}
		}; 
		\end{tikzpicture}
	\end{subfigure}
	\vspace*{1mm}
	\begin{subfigure}{0.5\textwidth}\centering
		\begin{tikzpicture}[-,>=stealth',
		level distance = 0.6cm,
		level 2/.style={sibling distance=1.6cm, level distance = 0.6cm},
		level 3/.style={sibling distance=0.8cm, level distance = 0.6cm},
		] 
		\node [arn_gs, label=left:{$\un{}$}]{ } 
		child{node [arn_gs, label=left:{$\un{x}$}] {$\large{\times}$}  
			child{ node [arn_gs, ultra thick, label=right:{$>$1}] {$\large{\times}$} 
				child{ node [arn_ws, label=below:{$\un{xll}$}] { } }
				child{ node [arn_ws, label=below:{$\un{xlr}$}] { } }
			}
			child{ node [arn_gs, label=right:{$>$1}] {$\large{\times}$} 
				child{ node [arn_ws, label=below:{$\un{xrl}$}] { }}
				child{ node [arn_ws, label=below:{$\un{xrr}$}] { } }
			}
		}; 
		\end{tikzpicture}
		{\raisebox{2.7\height}{\scalebox{1}{$\stackrel{\textsc{W7}}{\xrightarrow{\makebox[0.5cm]{}}}$}}}
		\begin{tikzpicture}[-,>=stealth',
		level distance = 0.6cm,
		level 2/.style={sibling distance=1.8cm, level distance = 0.6cm},
		level 3/.style={sibling distance=0.9cm, level distance = 0.6cm},
		] 
		\node [arn_gs, label=left:{$\un{}$}]{ } 
		child{node [arn_ws, label=left:{$n$}, label=right:{$\un{x}.w$+1}] {\large{+}}  
			child{ node [arn_ws, label=right:{$\un{xl}.w$-1}] {\large{+}} 
				child{ node [arn_ws, label=below:{$\un{xll}$}] { } }
				child{ node [arn_ws, label=below:{$\un{xlr}$}] { } }
			}
			child{ node [arn_ws, label=right:{$\un{xr}.w$-1}] {\large{+}} 
				child{ node [arn_ws, label=below:{$\un{xrl}$}] { }}
				child{ node [arn_ws, label=below:{$\un{xrr}$}] { } }
			}
		}; 
		\end{tikzpicture}
	\end{subfigure}
	\begin{subfigure}{0.5\textwidth}\centering
		\begin{tikzpicture}[-,>=stealth',
		level distance = 0.6cm,
		level 2/.style={sibling distance=1.7cm, level distance = 0.6cm},
		level 3/.style={sibling distance=1cm, level distance = 0.6cm},
		] 
		\node [arn_gs, label=left:{$\un{}$}]{ } 
		child{node [arn_gs, label=left:{$\un{x}$}] {$\large{\times}$}  
			child{ node [arn_gs, ultra thick, label=right:{$>$1}] {$\large{\times}$} 
				child{ node [arn_ws, label=below:{$\un{xll}$}] { } }
				child{ node [arn_ws, label=below:{$\un{xlr}$}] { } }
			}
			child{ node [arn_gs, label=right:{$1$}] {$\large{\times}$} 
				child{ node [arn_ws, label=below:{$\un{xrl}$}, label=right:{$>$0}] { }}
				child{ node [arn_ws, label=below:{$\un{xrr}$}, label=right:{$>$0}] { } }
			}
		}; 
		\end{tikzpicture}
		{\raisebox{2.7\height}{\scalebox{1}{$\stackrel{\textsc{Push}}{\xrightarrow{\makebox[0.5cm]{}}}$}}}
		\begin{tikzpicture}[-,>=stealth',
		level distance = 0.6cm,
		level 2/.style={sibling distance=1.9cm, level distance = 0.6cm},
		level 3/.style={sibling distance=0.8cm, level distance = 0.6cm},
		] 
		\node [arn_gs, label=left:{$\un{}$}]{ } 
		child{node [arn_ws, label=left:{$n$}, label=right:{$\un{x}.w$+1}] {\large{+}}  
			child{ node [arn_ws, label=right:{$\un{xl}.w$-1}] {\large{+}} 
				child{ node [arn_ws, label=below:{$\un{xll}$}] { } }
				child{ node [arn_ws, label=below:{$\un{xlr}$}] { } }
			}
			child{ node [arn_ws, label=right:{$0$}] {\large{+}} 
				child{ node [arn_ws, label=below:{$\un{xrl}$}] { }}
				child{ node [arn_ws, label=below:{$\un{xrr}$}] { } }
			}
		}; 
		\end{tikzpicture}
	\end{subfigure}
	\caption{The chromatic tree transformations.}
	\label{fig_transformations} 
\end{figure}

\begin{figure}[htbp!]
	\begin{center}
		\begin{tabular}{ | p{2.2cm} | p{4.5cm} | p{4.5cm} | p{4.5cm} |}
			\hline
			& \textbf{Violations Created} & \textbf{Violations Removed} & \textbf{Violations Moved} \\ \hline
			\textsc{Insert} & \tabitem red-red at \nn{} (if $\un{x}.w = 1$ and $\un{}.w = 0$) & \tabitem 1 overweight at \un{x} (if $\un{x}.w > 1$) & \tabitem from \un{x} to \nn{} \\ \hline
			\textsc{Delete} & \tabitem 1 overweight at \nn{} (if $\un{x}.w > 0$ and  $\un{xr}.w > 0$) & \tabitem red-red at \un{x} (if  $\un{xr}.w > 0$)  &  \tabitem from \un{x} and \un{xr} to \nn{} \\ 
			\space & \space & \tabitem red-red at \un{xr} (if any) & \space \\ 
			\space & \space & \tabitem red-red at \un{xrl} or \un{xrr} (if $\un{x}.w + \un{xr}.w > 0$) & \space \\  
			\space & \space & \tabitem any overweight at \un{xl} & \space \\ \hline 
			\textsc{BLK} (center at \un{xll}) & \tabitem None & \tabitem red-red at \un{xll} (if $\un{x}.w \neq 1$ or $\un{}.w \neq 0$) & \tabitem from \un{x} to \nn{} \\
			\space & \space  & \tabitem 1 overweight at \un{x} (if $\un{x}.w > 0$) & \tabitem red-red at \un{xll} to \nn{} (if $\un{x}.w = 1$ and $\un{}.w = 0$) \\ 
			\space & \space  & \tabitem red-red at \un{xlr}, \un{xrl}, or \un{xrr} (if any) & \space \\ \hline
			\textsc{RB1} & \tabitem None & \tabitem red-red at \un{xll} & \tabitem from \un{x} to \nn{} \\ \hline
			\textsc{RB2} & \tabitem None & \tabitem red-red at \un{xlr} & \tabitem from \un{x} to \nn{} \\ \hline
			\textsc{W1} & \tabitem None & \tabitem 1 overweight at \un{xl} & \tabitem from \un{x} to \nn{} \\ 
			\space & \space & \tabitem 1 overweight at \un{xlr} & \tabitem $\un{xl}.w-2$ overweight from \un{xl} to \nn{ll} \\
			\space & \space  & \tabitem red-red at \un{xrr} (if any) & \tabitem $\un{xrl}.w-2$ overweight from \un{xrl} to \nn{lr} \\ \hline
			\textsc{W2} & \tabitem None & \tabitem 1 overweight at \un{xl} & \tabitem from \un{x} to \nn{} \\
			\space & \space & \tabitem red-red at \un{xrr} (if any)  & \tabitem $\un{xl}.w-2$ overweight from \un{xl} to \nn{ll} \\ \hline
			\textsc{W3} & \tabitem None & \tabitem 1 overweight at \un{xl} & \tabitem from \un{x} to \nn{}  \\
			\space & \space & \tabitem red-red at \un{xrlll}, \un{xrllr}, or \un{xrr} (if any) & \tabitem $\un{xl}.w-2$ overweight from \un{xl} to \nn{lll} \\ \hline
			\textsc{W4} & \tabitem None & \tabitem 1 overweight at \un{xl} & \tabitem from \un{x} to \nn{} \\
			\space & \space & \tabitem red-red at \un{xrlrl} or \un{xrlrr} (if any) & \tabitem $\un{xl}.w-2$ overweight from \un{xl} to \nn{ll} \\ \hline
			\textsc{W5} & \tabitem None & \tabitem 1 overweight at \un{xl} & \tabitem from \un{x} to \nn{} \\
			\space & \space & \tabitem red-red at \un{xrrl} or \un{xrrr} (if any) & \tabitem $\un{xl}.w-2$ overweight from \un{xl} to \nn{ll}  \\ \hline
			\textsc{W6} & \tabitem None & \tabitem 1 overweight at \un{xl} & \tabitem from \un{x} to \nn{} \\
			\space & \space & \tabitem red-red at \un{xrll} or \un{xrlr} (if any) & \tabitem $\un{xl}.w-2$ overweight from \un{xl} to \nn{ll} \\ \hline
			\textsc{W7} & \tabitem None & \tabitem 1 overweight at \un{xl} & \tabitem from \un{x} to \nn{} \\
			\space & \space & \tabitem 1 overweight at \un{xr} (if $\un{x}.w = 0$) & \tabitem $\un{xl}.w-2$ overweight from \un{xl} to \nn{l} \\
			\space & \space & \tabitem red-red at \un{x} (if $\un{}.w = \un{x}.w = 0$) & \tabitem $\un{xr}.w-2$ overweight from \un{xr} to \nn{r}, and 1 overweight from \un{xr} to \nn{} (if $\un{x}.w > 0$) \\ \hline
			\textsc{PUSH} & \tabitem None & \tabitem 1 overweight at \un{xl} (if $\un{x}.w = 0$) & \tabitem from \un{x} to \nn{}  \\
			\space & \space & \tabitem red-red at \un{x} (if $\un{x}.w = \un{}.w = 0$) & \tabitem 1 overweight from \un{xl} to \nn{} (if $\un{x}.w > 0$) \\ 
			\space & \space & \space & \tabitem $\un{xl}.w-2$ overweight from \un{xl} to \nn{l} \\ \hline
		\end{tabular}
	\end{center}
	\caption{Violations created, removed, or moved after each transformation.}\label{violations_move}
\end{figure}

Inspection of the rules outlined in Figure~\ref{violations_move} gives the following facts:
\begin{enumerate}
	\item Only \textsc{Insert} and \textsc{Delete} transformations that create a violation can increase the total number of violations in the chromatic tree.
	\item A rebalancing transformation that elevates a violation does not increase the total number of violations in the chromatic tree.
	\item A rebalancing transformation that does not elevate a violation decreases the total number of violations in the chromatic tree.
	\item The violation at the center node of a rebalancing transformation is either removed or elevated.
\end{enumerate}

When the SCX is performed, nodes are frozen top-down, starting with the nodes on the path from $u$ to the center node. For an SCX$(V,R,fld,new)$, we classify each node in $V$ as either a downwards node or a cross node, which depends on the order that the nodes are frozen.
\begin{definition}\normalfont\label{def_downward_node}
	Consider an invocation $S$ of SCX$(V,R,fld,new)$, where $V = \{v_1, \dots, v_k\}$, enumerated in the order they are frozen. For $1 \leq i < k$, $v_i$ is a \textit{cross node}\index{cross node} for $S$ if $v_i$ is the sibling of $v_{i+1}$, otherwise $v_i$ is a \textit{downwards node}\index{downwards node} for $S$. We define $v_k$ to be a downwards node for $S$.
\end{definition}
\noindent Only the \textsc{Delete}, W1-W7, PUSH, and BLK transformations contain cross nodes. If BLK is centered at either \un{xrl} or \un{xrr}, then \un{xr} is the only cross node. In the rest of these transformations, \un{xl} is the only cross node. All other gray nodes in Figure~\ref{fig_transformations} are downwards nodes.

\section{Chromatic Tree Modifications}\label{section_modifications}
The chromatic tree implementation as described by Brown, Ellen, and Ruppert \cite{DBLP:conf/ppopp/BrownER14} has poor amortized step complexity. For example, consider the following execution, which begins with $N$ processes invoking \textsc{Insert}$(k)$, where $k$ is the key of an element to be inserted at a node with depth $h$. The following three sequences of steps are repeated:
\begin{enumerate}
	\item All $N$ processes traverse the tree from $entry$ to the insertion point of $k$ and successfully perform all required LLXs.
	\item One process $P$ performs a successful SCX and completes its \textsc{Insert}$(k)$. Each other process fails its SCX and begins a new attempt of \textsc{Insert}$(k)$.
	\item  Process $P$ performs \textsc{Delete}$(k)$, and then invokes \textsc{Insert}$(k)$.
\end{enumerate}
Each time we repeat steps 1-3, a single \textsc{Insert} and a single \textsc{Delete} are completed by process $P$. However, the $N$ processes traverse the tree from $entry$ to the insertion point of $k$ a total of $N+1$ times. Even if we assume no rebalancing is necessary, $\Omega(h \cdot N)$ steps are taken for each completed operation in the execution. This results from processes restarting their search from $entry$ each time an insertion attempt fails. In this section, we show how the implementation can be modified so that it has $O(h + N)$ amortized step complexity.

\subsection{Modifying the Search Routine}
Instead of restarting the search routine from $entry$ after a failed update attempt, we backtrack through the path of visited nodes until an unmarked node is reached. Each process maintains a local stack of nodes visited along its search path. As discussed in Section~\ref{section_related}, this technique was used previously to give good amortized step complexity for the update operations of a binary search tree \cite{DBLP:conf/podc/EllenFHR13}.
\begin{figure}[H]
\begin{algorithmic}[1]
\State \textsc{BacktrackingSearch}($k$, $stack$)
\Indent
\State \textbf{if} $stack$ is empty \textbf{then} $l := entry$ \label{ln:search:entry}
\State \textbf{else} $l := \textsc{Pop}(stack)$ \label{ln:search:pop_l}
\While {$l.marked$} \Comment{Backtrack while nodes are $marked$} \label{ln:search:backtracking_start}
	\If {$l.\mathit{info}.state = \textsc{InProgress}$} \textsc{Help}$(l.\mathit{info})$ \Comment{Help finalize and remove $l$} \EndIf \label{ln:search:help}
	\State $l := \textsc{Pop}(stack)$ \label{ln:search:pop}
\EndWhile \label{ln:search:backtracking_end}
\While {$l$ is not a leaf} \Comment{Search for leaf containing $k$} \label{ln:search:forward_start}
	\State $\textsc{Push}(stack, l)$ \label{ln:search:push}
	\State \textbf{if} $k < l.key$ \textbf{then} $l := l.\mathit{left}$ \textbf{else} $l := l.right$  \label{ln:search:leaf_update}
\EndWhile \label{ln:search:forward_end}
\State $p := \textsc{Pop}(stack)$ \label{ln:search:pop_p}
\State $gp := \textsc{Top}(stack)$ \Comment{Necessary for \textsc{Delete} only} \label{ln:search:pop_gp}
\State \Return $\langle gp, p, l \rangle$
\EndIndent
\end{algorithmic}
\caption{Modified \textsc{Search}$(k)$ routine using backtracking instead of restarting searches at the $entry$ node.}
\label{search}
\end{figure}
In the first attempt of an update operation, an empty stack is initialized for \textsc{BacktrackingSearch}. Invocations of \textsc{BacktrackingSearch} in subsequent attempts use the same stack. A process is \textit{backtracking}\index{backtracking} between the time when a node is popped from the stack on line~\ref{ln:search:pop_l} and when an unmarked node is found on line~\ref{ln:search:backtracking_start}. Otherwise the process is \textit{looking for a leaf}\index{looking for a leaf}. The stack allows searches to restart at the first unmarked node encountered in the stack after a failed attempt, rather than from $entry$. The following facts about using a stack for backtracking in a binary search tree were shown in \cite{DBLP:conf/podc/EllenFHR13}. 

\begin{observation}\normalfont\label{stack_obs}
\hfill
\begin{enumerate}
	\item \label{stack_obs:entry} The first node that is pushed onto the stack is $entry$.
	\item\label{stack_obs:empty} The node $entry$ is never popped from the stack.
	\item\label{stack_obs:config} Only nodes that are or have been in the chromatic tree are pushed onto a process's stack.
	\item\label{stack_obs:connected} If a node $x$ appears immediately below node $y$ in a process's stack, then $x$ was the parent of $y$ in some previous configuration.
\end{enumerate}
\end{observation}

\subsubsection{Proof of Correctness}
We argue that the first unmarked node $x$ visited by a process $P$ during backtracking after a failed \textsc{TryInsert}$(k)$ or \textsc{TryDelete}$(k)$ is on the search path for key $k$ starting from $entry$. This implies that a solo run of \textsc{Search}$(k)$ will pass through $x$, and so $x$ is a valid restarting point for \textsc{BacktrackingSearch}$(k, stack)$. We first show that unmarked nodes are in the chromatic tree after they are added by some SCX.

\begin{lemma}\normalfont\label{unmarked}
	Any unmarked node in configuration $C$ that has been previously added to the chromatic tree is in the chromatic tree in $C$.
\end{lemma}

\begin{proof}
Consider any node unmarked node $x$ in $C$ that has been added to the chromatic tree in some configuration before $C$. Recall that nodes must be marked before they can be removed from the chromatic tree, and that marked nodes cannot be unmarked.  Since $x$ is unmarked in $C$, it could not have been removed from the chromatic tree since the configuration in which it was added. Therefore, $x$ is in the chromatic tree.
\end{proof}

Since updates are performed in the same way as in the original chromatic tree implementation, the following fact proven in \cite{DBLP:conf/ppopp/BrownER14} still holds. 

\begin{lemma}\normalfont\label{hindsight}
	If a node $x$ is in the data structure in some configuration $C$ and $x$ was on the search path for key $k$ in some earlier configuration, then $x$ is on the search path for $k$ in $C$. 
\end{lemma}

A process in its update phase \textit{visits}\index{visit} a node $x$ if its local variable $l$ in an instance of \textsc{BacktrackingSearch} points to $x$. We use the previous result to show that any nodes visited during \textsc{BacktrackingSearch}$(k,stack)$ when looking for a leaf were on the search path for $k$ in some earlier configuration.

\begin{lemma}\normalfont\label{reach_SP_backtracking}
	Consider an instance $I$ of \textsc{BacktrackingSearch}$(k, stack)$, and let $C$ be a configuration during $I$ sometime after backtracking. Let $C'$ be the configuration immediately after the last node popped during the backtracking of $I$, or the configuration after $I$ is invoked if no such pop occurred. If $x$ is the node pointed to by $I$'s local variable $l$ in $C$, then there is a configuration between $C'$ and $C$ in which $x$ is on the search path for $k$. 
\end{lemma}

\begin{proof}
We prove by induction on the configurations of an execution. We assume the lemma is true for all configurations before $C$, and show it is true in $C$.

Suppose $I$ has not yet executed line~\ref{ln:search:leaf_update} since the end of backtracking. If $l$ points to $entry$, then the lemma holds because $entry$ is always in the chromatic tree and is on the search path for every key. So suppose $l$ points to the last node popped from the stack during backtracking in $C$. By the check on line~\ref{ln:search:backtracking_start}, $x$ is unmarked in $C'$ (since the check on line~\ref{ln:search:backtracking_start} occurs after $x$ is popped and marked nodes do not become unmarked). By Observation~\ref{stack_obs}.\ref{stack_obs:config}, $x$ was in the chromatic tree in some configuration before it was pushed onto the stack. Therefore, by Lemma~\ref{unmarked}, $x$ is in the chromatic tree in $C$. By the induction hypothesis, $x$ was on the search path for $k$ during some earlier invocation of \textsc{BacktrackingSearch}$(k,stack)$ that pushed $x$ onto the stack. Since $x$ is in the chromatic tree in $C$ and on the search path for $k$ in a configuration before $C$, by Lemma~\ref{hindsight}, $x$ is on the search path for $k$ in $C$.

So suppose $I$ has executed line~\ref{ln:search:leaf_update} at least once since the end of backtracking. Since $l$ is only updated on line~\ref{ln:search:leaf_update} after backtracking, we show the lemma holds in the configuration $C$ immediately after $l$ is updated on this line. Let $y$ be the node pointed to by $l$ in the configuration $C^-$ immediately before $C$. By the induction hypothesis, there is a configuration $C''$ between $C'$ and $C^-$ such that $y$ is on the search path for $k$ in $C''$. From the code, line~\ref{ln:search:leaf_update} updates $l$ to the next node on the search path for $k$. So if $x$ is a child of $y$ in $C''$, then $x$ is on the search path for $k$ in $C''$. Since $C''$ is between $C'$ and $C$, the lemma holds. If $x$ is not a child of $y$ in $C''$, then there is a successful SCX that changes a child pointer of $y$ to point to $x$ in some configuration between $C''$ and $C$. Let  $C'''$ be the configuration immediately after this SCX. The node whose pointer is updated by an SCX is in the chromatic tree before and after the SCX, so $y$ is in the chromatic tree in $C'''$. By Lemma~\ref{hindsight}, $y$ is on the search path for $k$ in $C'''$. Since $x$ is the child of $y$ in $C'''$,  $x$ is on the search path for $k$ in $C'''$. Since $C'''$ is between $C'$ and $C$, the lemma holds.
\end{proof}

It was shown that each successful \textsc{Insert} and \textsc{Delete} operation implemented using \textsc{Search} can be linearized at its successful SCX in \textsc{TryInsert} and \textsc{TryDelete}, respectively \cite{DBLP:conf/ppopp/BrownER14}. This is used to show \textsc{Insert} and \textsc{Delete} operations using \textsc{BacktrackingSearch} can also be linearized at the same places.

\begin{lemma}\normalfont\label{search_correct}
Consider an invocation $S$ of \textsc{BacktrackingSearch}$(k,stack)$ that terminates in configuration $C$. Let the sequence of nodes returned by $S$ be $\langle gp, p, l \rangle$. Suppose an atomic invocation $S'$ of \textsc{Search}$(k)$ is performed in a configuration $C'$ after $C$. If the nodes $gp$, $p$ and $l$ are in the chromatic tree in $C'$, then $S'$ returns $\langle gp, p, l \rangle$.
\end{lemma}

\begin{proof}
To prove $S'$ returns $\langle gp,p,l \rangle$, we show the nodes $gp$, $p$, and $l$ must be last three nodes on the search path for $k$ from $entry$ in $C'$.

We first argue $l$ is a leaf and on the search path for $k$ in $C'$. By the check on line~\ref{ln:cleanup:leaf}, $S$ only returns after visiting a leaf, so $l$ is a leaf in $C$. There are no chromatic tree transformations that update the child pointers of leaves, so $l$ is a leaf in $C'$. By Lemma~\ref{reach_SP_backtracking}, there exists a configuration before $C$ in which $l$ was on the search path for $k$. Since $l$ is in the chromatic tree in $C'$, by Lemma~\ref{hindsight}, $l$ is on the search path for $k$ in configuration $C'$.

Next, we argue $p$ is the parent of $l$ and $gp$ is the grandparent of $l$ in $C'$. Let $C_l$ be the configuration immediately after $S$ visits $l$. By assumption, $S$ returns $\langle gp,p,l \rangle$, so the top two nodes on $stack$ in $C_l$ must be $gp$ and $p$. Therefore, by Observation~\ref{stack_obs}.\ref{stack_obs:connected}, there exists a configuration in which $gp$ is the parent of $p$, and a configuration in which $p$ is the parent of $l$. Inspection of the chromatic tree transformations shows that whenever a pointer of a node is modified, the node referenced by the pointer is removed from the chromatic tree. Therefore, no node can be inserted between $p$ and $l$ without removing $l$ from the chromatic tree. So $p$ is the parent of $l$ in $C'$. Similarly, $gp$ is the parent of $p$ in $C'$. 

Therefore, the nodes $gp$, $p$, and $l$ are the last three nodes on the search path for $k$ from $entry$ in $C'$. This proves $S'$ will return the nodes $\langle gp,p,l \rangle$.
\end{proof}

\begin{theorem}\label{search_correct_2}
An \textsc{Insert} or \textsc{Delete} operation implemented using \textsc{BacktrackingSearch} can be linearized at its successful SCX in \textsc{TryInsert} or \textsc{TryDelete}, respectively.
\end{theorem}

\begin{proof}
Consider a successful invocation of \textsc{TryDelete}$(gp,p,l)$. Let $C'$ be the configuration immediately before its invocation of SCX. This SCX guarantees the nodes $gp$, $p$, and $l$ are in the chromatic tree in $C'$. Consider the last invocation of \textsc{BacktrackingSearch} before $C'$. It returned $\langle gp, p, l \rangle$. By Lemma~\ref{search_correct}, this invocation of \textsc{BacktrackingSearch} can be replaced with an atomic invocation of \textsc{Search} in $C'$ that also returns $\langle gp, p, l \rangle$. Since a \textsc{Delete} operation using \textsc{Search} can be linearized at its successful SCX in \textsc{TryDelete}, this implies a \textsc{Delete} operation using \textsc{BacktrackingSearch} can also be linearized at its successful SCX in \textsc{TryDelete}. A similar argument can be made for \textsc{Insert}.
\end{proof}

\subsection{Modifying the Cleanup Phase}
In the original cleanup phase, operations also restart searches from $entry$ after rebalancing attempts, regardless of whether the rebalancing succeeds or fails. To improve the amortized step complexity, backtracking can also be added to the search in the cleanup phase. Our modified cleanup phase is outlined in Figure \ref{cleanup}. An empty local stack is initialized at the start of the cleanup, and is reused until the cleanup phase terminates. Nodes are pushed onto the stack when traversing the search path for $k$. When a violation is reached, a rebalancing attempt is performed. Regardless of whether this attempt succeeds or fails, a process backtracks to the first unmarked node on its stack before resuming its traversal to the leaf on the search path for $k$. A process performing \textsc{BacktrackingCleanup} is \textit{backtracking}\index{backtracking} between the time when a node is popped from the stack on line~\ref{ln:cleanup:first_pop} and when an unmarked node is found on line~\ref{ln:cleanup:backtracking_start}. If a process is not backtracking, it is \textit{looking for a violation}\index{looking for a violation}. We divide an invocation $I$ of \textsc{BacktrackingCleanup}$(k)$ into a series of \textit{attempts}. The first attempt of $I$ begins when $I$ is invoked, and all other attempts of $I$ begin with backtracking. The last attempt of $I$ ends on line~\ref{ln:cleanup:leaf} when a leaf is found, and all other attempts of $I$ end on line~\ref{ln:cleanup:exit_loop}. A process in its cleanup phase \textit{visits}\index{visit} a node $x$ if its local variable $l$ in \textsc{BacktrackingCleanup} points to $x$.

\begin{figure}[H]
\begin{algorithmic}[1]
\State \textsc{BacktrackingCleanup}($k$)
\Indent
	\State $stack$ := new empty stack for pointers to nodes
	\While {\textsc{True}} \label{ln:cleanup:main_loop}
		\State \textbf{if} $stack$ is empty \textbf{then} $l := entry$ \label{ln:cleanup:entry} 
		\State \textbf{else} $l := \textsc{Pop}(stack)$ \label{ln:cleanup:first_pop} 
		\While {$l.marked$} \Comment{Backtrack while nodes are no longer in tree} \label{ln:cleanup:backtracking_start}
			\If {$l.\mathit{info}.state = \textsc{InProgress}$} \textsc{Help}$(l.\mathit{info})$ \EndIf \label{ln:cleanup:help} 
			\State $l := \textsc{Pop}(stack)$ \label{ln:cleanup:pop}
		\EndWhile \label{ln:cleanup:backtracking_end}
		\While {\textsc{True}} \label{ln:cleanup:forward_start}
			\If {$l.w > 1$ \textbf{or} ($\textsc{Top}(stack).w = 0$ \textbf{and} $l.w = 0$)}  \Comment{Overweight or red-red violation at $l$} \label{ln:cleanup:viol_check}
				\State $p := \textsc{Pop}(stack)$ \label{ln:cleanup:p_pop}
				\State $gp := \textsc{Pop}(stack)$ \label{ln:cleanup:gp_pop}
				\State $ggp := \textsc{Top}(stack)$ \label{ln:cleanup:ggp_top}
				\State \textsc{TryRebalance}($gpp, gp, p, l$)
				\State \textbf{exit loop} \Comment{Go back to line 4} \label{ln:cleanup:exit_loop}
			\EndIf
			\If {$l$ is a leaf} \Return \EndIf \label{ln:cleanup:leaf}
			\State $\textsc{Push}(stack, l)$ \label{ln:cleanup:push}
			\State \textbf{if} $k < l.key$ \textbf{then} $l := l.\mathit{left}$ \textbf{else} $l := l.right$ \label{ln:cleanup:leaf_update}  \Comment{Search for leaf containing $k$}
		\EndWhile \label{ln:cleanup:forward_end}
	\EndWhile
\EndIndent
\end{algorithmic}
\caption{Modified \textsc{Cleanup}$(k)$ routine with backtracking.}
\label{cleanup}
\end{figure}

\subsubsection{Bounding the Height of the Tree}
Unlike insertions and deletions, rebalancing transformations do not change the keys of elements in the dynamic set and are not linearized. The original cleanup algorithm does has the property that all violations in the chromatic tree are eventually cleaned up by a process in its cleanup phase. This gives a height bound on the chromatic tree as a function of the number of active \textsc{Insert} or \textsc{Delete} operations. In this section, we show \textsc{BacktrackingCleanup} also has this property.

We first show basic properties about violations and the nodes on a process's local stack used for backtracking. A violation $v$ is \textit{on a search path} if $v$ is located at a node on the search path. A violation $v$ is \textit{reachable} from a node $x$ if one could follow child pointers starting from $x$ and visit the node containing $v$.
\begin{lemma}\normalfont\label{v_path_entry}
A violation created during an \textsc{Insert}$(k)$ or \textsc{Delete}$(k)$ operation is on the search path for $k$ until it is removed.
\end{lemma}

\begin{proof}
Recall that by inspection of Figure~\ref{violations_move}, only \textsc{Insert} and \textsc{Delete} transformations create violations. By inspection of Figure~\ref{fig_transformations}, if an \textsc{Insert} transformation adds a leaf with key $k$, or if a \textsc{Delete} transformation removes a leaf with key $k$, then any violations $v$ created by the transformation are on the search path for $k$. By Lemma~\ref{viol}, $v$ is on the search path for $k$ until it is removed.
\end{proof}

\begin{lemma}\normalfont\label{v_reachable}

Consider a node $x$ in the chromatic tree that is on the search path for $k$ in a configuration $C$. If $x$ is reachable from a node $y$ in the chromatic tree in $C$, then $x$ is on the search path for $k$ from $y$ in $C$.
\end{lemma}

\begin{proof}

Since $x$ is reachable from $y$, $y$ is on the unique path of nodes in the chromatic tree from $entry$ to $x$ in $C$. So $y$ is also on the search path for $k$ in $C$. Since $y$ appears before $x$ on the search path for $k$ in $C$, $x$ is on the search path for $k$ from $y$ in $C$.
\end{proof}

\begin{lemma}\normalfont\label{red_red_child}
Let $v$ be a node with weight 0. If $v$ has a parent with non-zero weight at some point in an execution after $v$ is added to the chromatic tree, then it does not gain a parent with weight 0.
\end{lemma}

\begin{proof}
Recall that weights are immutable. So the weight of $v$'s parent can only change if its parent is replaced with a new node. By inspection of each chromatic tree transformation in Figure~\ref{fig_transformations}, if the parent of $v$ is replaced with a new node with weight 0, then $v$'s parent had weight 0 before the transformation.
\end{proof}

\begin{lemma}\normalfont\label{no_viol_added}
	A node cannot gain violations while it is in the chromatic tree, and the only violation it can lose is a red-red violation, which occurs as the result of an update CAS.
\end{lemma}

\begin{proof}
Weights are immutable, so nodes cannot gain or lose overweight violations after they are added to the chromatic tree. By Lemma~\ref{red_red_child}, a node with weight 0 that does not contain a red-red violation cannot gain a parent with weight 0. Thus, nodes do not gain violations while they are in the chromatic tree. 

A node with weight 0 may lose a red-red violation if its parent with weight 0 is removed from the chromatic tree and replaced with a node with non-zero weight. Nodes are only removed from the chromatic tree as the result of an update CAS.
\end{proof}

\begin{lemma}\normalfont\label{no_viol_stack}
Any node on a process's stack in an invocation of \textsc{BacktrackingCleanup}$(k)$ contains no violations.
\end{lemma}

\begin{proof}
Consider a node $x$ pushed onto $stack$ on line~\ref{ln:cleanup:push}.  Before any node is pushed onto $stack$, it must fail the check for violations on line~\ref{ln:cleanup:viol_check}. If $x.w > 1$, the check on line~\ref{ln:cleanup:viol_check} will pass and so $x$ will not be pushed onto $stack$. If $x.w = 1$, then $x$ does not contain any violations. So suppose $x.w = 0$. Let $y$ be the node pointed to by \textsc{Top}$(stack)$ on line~\ref{ln:cleanup:viol_check}, where $y.w \neq 0$. By Observation~\ref{stack_obs}.\ref{stack_obs:connected}, $y$ was once a parent of $x$. Since $x$ has a parent with non-zero weight 0 at some point, by Lemma~\ref{red_red_child}, $x$ cannot gain a parent with weight 0. Therefore, there are no violations at $x$.
\end{proof}

Consider a chromatic tree transformation $T$ that updates the pointer of a single node $u$, which removes a connected subtree of nodes $R$ and replaces it with a new connected subtree of nodes $A$. The following facts follow by inspection of Figure~\ref{fig_transformations}:
\begin{observation}\normalfont\label{trans_obs}
\hfill
	\begin{enumerate}
		\item \label{trans_obs:u} The node $u$ is in the chromatic tree before and after $T$.
		
		\item \label{trans_obs:R} A node is removed from the chromatic tree if and only if it is a node in $R$.
		\item \label{trans_obs:A} A node is added to the chromatic tree if and only if it is a node in $A$.
		\item \label{trans_obs:reach_before} If a node is reachable from $u$ before $T$ and is not a node in $R$, then the node is reachable from $u$ after $T$.
		\item \label{trans_obs:reach_after} If a node is reachable from $u$ after $T$ and is not a node in $A$, then the node was reachable from $u$ before $T$.
	\end{enumerate}
\end{observation}

\begin{definition}\normalfont\label{global_graph}
For any execution $\alpha$, let its \textit{global graph}, denoted $G_\alpha = (V_\alpha,E_\alpha)$, be a directed graph whose vertices are the union of nodes in the chromatic tree at all configurations in $\alpha$. There is an edge from vertex $x$ to vertex $y$ if and only if $y$ is the child of $x$ sometime during $\alpha$.
\end{definition}

\begin{lemma}\normalfont\label{g_alpha}
	For every execution $\alpha$, $G_\alpha$ contains no cycles.
\end{lemma}

\begin{proof}
We show by induction on the prefixes of the execution $\alpha$ that $G_\alpha$ contains no cycles. Let $\phi$ denote the empty execution. $G_\phi$ only contains the initial sentinel nodes, which do not form a cycle. Suppose $\alpha = \alpha's$, where $s$ is a step by some process. We assume $G_{\alpha'}$ contains no cycles, and show $G_\alpha$ contains no cycles.

The only step $s$ that changes the structure of the chromatic tree is an SCX. Let $u$ be the node whose pointer is updated by $s$, which removes a connect set of nodes $R$ and replaces it with a new connected set of nodes $A$. Only nodes that were not in the chromatic tree in any prior configuration are added to the chromatic tree, and so $A \cap V_{\alpha'} = \emptyset$ and $V_\alpha = V_{\alpha'} \cup A$.

Consider the edges in $G_\alpha$ that were added by $s$. There is one edge from $u$ to the root of $A$. No other edge in $E_\alpha$ connects a node in $V_{\alpha'}$ to a node in $A$. Besides the edges between nodes in $A$, the only other edges added by $s$ are those from nodes in $A$ to nodes in $V_{\alpha'}$. 

If a path in $G_\alpha$ only contains vertices in $V_{\alpha'}$, then by the induction hypothesis, this path does not form a cycle. If a path in $G_\alpha$ only contains vertices in $A$, then since $A$ is a tree, the path does not form a cycle. So consider a path containing both vertices in $V_{\alpha'}$ and $A$. The path contains the edge from $u$ to the root of $A$, since this is the only edge from a vertex in $V_{\alpha'}$ to a vertex in $A$. Since $A$ is a tree, the path can only form a cycle if it also passes through an edge connecting a vertex in $A$ to a vertex in $V_{\alpha'}$. Let this edge be $(x,y)$, where $x \in A$ and $y \in  V_{\alpha'}$.

Since $A$ is a tree, there is an unique path from $u$ to $x$ in $G_\alpha$. By Observation~\ref{trans_obs}.\ref{trans_obs:reach_after}, since $y$ is reachable from $u$ after $s$ and $y \notin A$, there is a path from $u$ to $y$ in $G_{\alpha'}$. Therefore, since $G_{\alpha'}$ contains no cycles, the path from $u$ to $y$ containing the edge $(x,y)$ in $G_\alpha$ does not form a cycle. Since no new path created by $s$ creates a cycle, $G_\alpha$ does not contain a cycle.
\end{proof}

Consider any process $P$ in its cleanup phase for key $k$. Let $location(P,C)$\index{$location(P,C)$} be the node pointed to by $P$'s local variable $l$ in configuration $C$. We let \textit{$P$'s search path from $r$} be the search path for $k$ starting from the node $r$. If $r$ is not mentioned, we assume the search path is starting from $entry$.

\begin{lemma}\normalfont\label{no_cycle_claim}
For every node $x$ on $P$'s stack in configuration $C$ during an execution $\alpha$, there is a path in $G_\alpha$ from $x$ to $location(P,C)$. 
\end{lemma}

\begin{proof}
We prove by induction on the sequence of configurations of $\alpha$. We assume the lemma holds for all configurations before an arbitrary step $s$, and show the lemma holds in the configuration $C$ immediately after $s$. Let $C^-$ be the configuration immediately before $C$. In the initial configuration, there are no nodes on $P$'s stack, and so the lemma holds trivially. The types of steps $s$ that affect the truth of lemma are those that update $P$'s local variable $l$, push nodes onto $P$'s stack, or pop nodes off $P$'s stack. For a process $P$ in \textsc{BacktrackingCleanup}, such steps occur on lines~\ref{ln:cleanup:entry}, \ref{ln:cleanup:first_pop}, \ref{ln:cleanup:pop}, \ref{ln:cleanup:p_pop}, \ref{ln:cleanup:gp_pop}, \ref{ln:cleanup:push}, and \ref{ln:cleanup:leaf_update}. (For a process $P$ in \textsc{BacktrackingSearch}, the steps occur on lines~\ref{ln:search:entry}, \ref{ln:search:pop_l}, \ref{ln:search:pop}, \ref{ln:search:push}, \ref{ln:search:leaf_update}, \ref{ln:search:pop_p}.)
\begin{itemize}
	\item  Suppose $s$ is a step that updates $P$'s local variable $l$ to point to $entry$ on line~\ref{ln:cleanup:entry} (on line~\ref{ln:search:pop_l} for \textsc{BacktrackingSearch}). From the code, this only occurs when $P$'s stack is empty, so the lemma is vacuously true.
	
	\item Suppose $s$ is a step that pops a node $m$ off of $P$'s stack on line~\ref{ln:cleanup:first_pop} or \ref{ln:cleanup:pop} (on line~\ref{ln:search:pop_l} or \ref{ln:search:pop} for \textsc{BacktrackingSearch}). Consider the configuration $C_m$ immediately before $m$ was last pushed onto the stack. A node is only pushed onto $P$'s stack on line~\ref{ln:cleanup:push} (on line~\ref{ln:search:push} for \textsc{BacktrackingSearch}) when pointed to by the local variable $l$, so $location(P,C_m) = m$. By definition of $C_m$, the nodes on $P$'s stack are the same in $C_m$ and $C$. Since $l$ is updated to the node popped by $s$, $location(P,C) = m$. By the induction hypothesis, for every node $x$ on $P$'s stack in $C_m$ (and hence on $P$'s stack in $C$), there is a path in $G_\alpha$ from $x$ to $location(P,C_m) = location(P,C)$.
	
	\item Suppose $s$ is a step that pops a node $m$ off of $P$'s stack on line~\ref{ln:cleanup:p_pop} or \ref{ln:cleanup:gp_pop} (on line~\ref{ln:search:pop_p} for \textsc{BacktrackingSearch}). Since $location(P,C^-) = location(P,C)$ and the nodes on $P$'s stack in $C$ are a subset of those on $P$'s stack in $C^-$, the lemma follows from the induction hypothesis.
	
	\item Suppose $s$ is a step that pushes a node $m$ onto $P$'s stack on line~\ref{ln:cleanup:push} (on line~\ref{ln:search:push} for \textsc{BacktrackingSearch}). Let $C^-$ be the configuration immediately before $C$. Since $location(P,C^-) = location(P,C)$, the induction hypothesis implies the lemma is true for all nodes on $P$'s stack in $C$ except for $m$. Since $location(P,C) = m$, there is an empty path from $m$ to $location(P,C)$. Therefore, the lemma holds for all nodes on $P$'s stack in $C$.
	
	\item Suppose $s$ is a step that updates $P$'s local variable $l$ on line~\ref{ln:cleanup:leaf_update} (on line~\ref{ln:search:leaf_update} for \textsc{BacktrackingSearch}). Let $C^-$ be the configuration immediately before $C$. From the code, $l$ is updated to a child of $l$. Therefore, there is an edge in $G_\alpha$ from $location(P,C^-)$ to $location(P,C)$. The nodes on $P$'s stack are the same in $C^-$ and $C$. By the induction hypothesis, for every node $x$ on $P$'s stack in $C$, there is a path from $x$ to $location(P,C^-)$ in $G_\alpha$. Following the edge $(location(P,C^-), location(P,C))$, there is a path from $x$ to $location(P,C)$ in $G_\alpha$.
\end{itemize}
\end{proof}

\noindent We use $G_\alpha$ to prove some basic facts about the ordering of nodes on a process's local stack used during backtracking.
\begin{lemma}\normalfont\label{no_cycle}
	If a node $x$ is on a process $P$'s stack in configuration $C$, then $x$ is not reachable from $location(P,C)$.
\end{lemma}

\begin{proof}

By the Lemma~\ref{no_cycle_claim}, there is a path from $x$ to $location(P,C)$ in $G_\alpha$. If $x$ is reachable from $location(P,C)$, this implies there is a path from $location(P,C)$ to $x$ in $G_\alpha$. This contradicts Lemma~\ref{g_alpha}, which states that $G_\alpha$ contains no cycles. Therefore, $x$ is not reachable from $location(P,C)$.
\end{proof}

Notice that backtracking in \textsc{BacktrackingCleanup}$(k)$ is performed in the same manner as in \textsc{BacktrackingSearch}$(k,stack)$. Therefore, the proof of Lemma~\ref{reach_SP_backtracking} also shows that a node $x$ visited by \textsc{BacktrackingCleanup}$(k)$ was on the search path for $k$ in some earlier configuration. 
\begin{corollary}\label{reach_SP_backtracking_cleanup}
(of Lemma~\ref{reach_SP_backtracking}) Consider an instance $I$ of \textsc{BacktrackingCleanup}$(k)$ in an attempt $A$, and let $C$ be a configuration during $A$ sometime after backtracking. Let $C'$ be the configuration immediately after the last node popped during the backtracking in $A$, or the configuration after the start of $A$ if no such pop occurred. If $x$ is the node pointed to by $I$'s local variable $l$ in $C$, then there is a configuration between $C'$ and $C$ in which $x$ is on the search path for $k$. 
\end{corollary}

\begin{lemma}\normalfont\label{on_stack_SP_2}
In any configuration $C$, any node on the stack of a process performing \textsc{BacktrackingCleanup}$(k)$ (or \textsc{BacktrackingSearch}$(k,stack)$) that is in the chromatic tree is on the search path for $k$ in $C$.
\end{lemma}

\begin{proof}
Since $x$ is on the stack in $C$, $x$ was pushed onto the stack in some configuration before $C$. Only nodes pointed to by a process's local variable $l$ are pushed onto its stack on line~\ref{ln:cleanup:push} of \textsc{BacktrackingCleanup} or on line~\ref{ln:search:push} of  \textsc{BacktrackingSearch}. Hence, by Corollary~\ref{reach_SP_backtracking_cleanup} (or Lemma~\ref{reach_SP_backtracking} for \textsc{BacktrackingSearch}), $x$ was on the search path for $k$ in some configuration before $C$. Therefore, since $x$ is in the chromatic tree in $C$, by Lemma~\ref{hindsight}, $x$ is on the search path for $k$ in $C$.
\end{proof}

\begin{lemma}\normalfont\label{on_stack_SP}
	In any configuration $C$, any unmarked node on the stack of a process performing \textsc{BacktrackingCleanup}$(k)$ (or \textsc{BacktrackingSearch}$(k,stack)$) is on the search path for $k$ in $C$.
\end{lemma}

\begin{proof}
By Observation~\ref{stack_obs}.\ref{stack_obs:config}, $x$ was once in the chromatic tree. Since $x$ is unmarked in $C$, by Lemma~\ref{unmarked}, $x$ is still in the chromatic tree in $C$. By Lemma~\ref{on_stack_SP_2}, $x$ is on the search path for $k$ in $C$.
\end{proof}

\begin{lemma}\normalfont\label{stack_ancestors}
	Consider two nodes $x$ and $y$ on a process's stack in configuration $C$ during an instance of \textsc{BacktrackingCleanup}$(k)$ (or \textsc{BacktrackingSearcg}$(k, stack)$) that are in the chromatic tree. If $x$ appears above $y$ on the stack in $C$, then $x$ is on the search path for $k$ from $y$ in $C$.
\end{lemma}

\begin{proof}
By Observation~\ref{stack_obs}.\ref{stack_obs:connected}, there is a path from $y$ to $x$ in $G_\alpha$. By Lemma~\ref{on_stack_SP_2}, both $x$ and $y$ are on the search path for $k$ from $entry$. If $x$ is not on the search path for $k$ from $y$, then $y$ is on the search path for $k$ from $x$. This implies there is a path from $x$ to $y$ in $G_\alpha$, which forms a cycle. By Lemma~\ref{g_alpha}, $G_\alpha$ contains no cycles, and so $x$ is on the search path for $k$ from $y$ in $C$.
\end{proof}

\begin{lemma}\normalfont\label{viol_loc}
	Consider a process $P$ in its cleanup phase in configuration $C$. Let $m$ be any unmarked node on $P$'s stack in $C$. If a violation $v$ is on $P$'s search path from $location(P,C)$, then $v$ is on $P$'s search path from $m$.
\end{lemma}

\begin{proof}
By Lemma~\ref{on_stack_SP}, $m$ is on $P$'s search path in $C$. Suppose $m$ appears after $v$ on $P$'s search path. Since $v$ is on $P$'s search path from $location(P,C)$, so is $m$. This contradicts Lemma~\ref{no_cycle}, which states that $m$ is not reachable from $location(P,C)$. Therefore, $v$ occurs after $m$ on $P$'s search path, and so $v$ is on $P$'s search path from $m$.
\end{proof}

Next we prove some basic facts about how violations are moved after a rebalancing transformation $T$. Inspection of how violations are created, removed, or moved by $T$ according to Figure~\ref{violations_move} gives the following facts.

\begin{observation}\normalfont\label{move_obs} 
Consider an chromatic tree transformation $T$ that updates the pointer of a single node $u$, which removes a connected subtree of nodes $R$ and replaces it with a new connected subtree of nodes $A$.
\begin{enumerate}
	\item \label{move_obs:A} Every violation moved by $T$ is moved to a node in $A$.
	\item \label{move_obs:move} Every violation located at a node in $R$ is either moved to a node in $A$ or removed.
	\item \label{move_obs:blk} A BLK transformation centered at a node $x$ that elevates the red-red violation $v$ located at $x$ moves $v$ from $x$ to a node in $A$, where $x \notin R$ is a child of a node in $R$. Every other transformation only moves violations located at nodes in $R$.
\end{enumerate}
\end{observation}

\begin{lemma}\normalfont\label{violation_no_move}
Let $x$ be any node, not necessarily in the chromatic tree. Consider a chromatic tree transformation $T$ that does not move or remove violation $v$. If $v$ is on the search path for a key $k$ from $x$ before $T$, then $v$ is on the search path for $k$ from $x$ after $T$.
\end{lemma}

\begin{proof}
Consider the node $u$ whose pointer is updated by $T$. If $u$ is not on the search path for $k$ from $x$, then no pointer along the path from $x$ to the node containing $v$ is changed. Since $v$ is not moved, it is located at the same node before and after $T$. Therefore, $v$ is on the search path for $k$ from $x$ after $T$.

If $u$ is on the search path for $k$ from $x$, then by Observation~\ref{trans_obs}.\ref{trans_obs:reach_before}, the nodes reachable from $u$ before $T$ are reachable from $u$ after $T$, except for the nodes $R$ removed by $T$. So only nodes in $R$ are removed from the search path for $k$ from $x$. By Observation~\ref{move_obs}.\ref{move_obs:move}, every violation in $R$ is either moved or removed. Since $v$ is not moved or removed, it is not located at a node in $R$, and so $v$ is still on the search path for $k$ from $x$ after $T$.
\end{proof}

\noindent The following lemma is a generalized version of Lemma~\ref{viol} that applies to search paths starting from nodes in the chromatic tree other than $entry$.
\begin{lemma}\normalfont\label{violation_stay}
	Consider any node $x$ with no violations that is in the chromatic tree before and after a transformation $T$. If violation $v$ is on the search path for key $k$ from $x$ before $T$, then $v$ is still on the search path for $k$ from $x$ after $T$, or has been removed. 
\end{lemma}

\begin{proof}
Let $C^-$ be the configuration immediately before $T$, and $C$ be the configuration immediately after $T$. Let $u$ be the node whose pointer is updated by $T$, $R$ be the nodes removed by $T$, and $A$ be the nodes added by $T$. By Lemma~\ref{violation_no_move}, if $v$ is not moved by $T$, then $v$ is on the search path for $k$ from $x$ in $C$. 

So suppose $v$ is moved by $T$. Then it moves from a node $r$ to a new node $r'$, where $r' \in A$ by Observation~\ref{move_obs}.\ref{move_obs:A}. By Observation~\ref{move_obs}.\ref{move_obs:blk}, either $r \in R$ or $r$ is a child of a node in $R$. By inspection of Figure~\ref{fig_transformations}, every node in $R$ is reachable from $u$ in $C^-$, so there is a path from $u$ to $r$ in $C^-$. Since both $x$ and $u$ are in the chromatic tree, they are both on the path from $entry$ to $r$ in $C^-$. 

We argue $x$ is an ancestor of $u$ in $C^-$. By assumption, $x$ contains no violations in $C^-$, so $x \neq r$. Since $r \in R$ or $r$ is a child of a node in $R$, and the nodes in $R$ are connected, it follows that all nodes on the path from $u$ to the parent of $r$ except for $u$ are in $R$. By Observation~\ref{trans_obs}.\ref{trans_obs:R}, the nodes on this path are removed from the chromatic tree by $T$. Since $x$ is not removed from the chromatic tree by $T$, $x$ is an ancestor of $u$ in $C^-$.

The path of nodes from $entry$ to $u$ are unmodified by $T$, so $x$ is an ancestor of $u$ in $C$. By Lemma~\ref{v_path_entry}, $r'$ is still on the search path for $k$ in $C$. Since $r' \in A$, $r'$ is reachable from $u$ in $C$, and so $r'$ is reachable from $x$ in $C$. Therefore, by Lemma~\ref{v_reachable}, $v$ is on the search path for $k$ from $x$ in $C$.
\end{proof}

The next lemma informally states that for every violation $v$ in the chromatic tree in a configuration $C$, either its creator $P$ contains $v$ on its search path starting from $location(P,C)$, or $P$ will backtrack to a node such that $v$ is on the search path from that node. 

\begin{lemma}\normalfont\label{clean_main_invariant}
In every configuration $C$, and for every violation $v$ in the chromatic tree in $C$, the creator $P$ of $v$ is in the cleanup phase, and one of the following two conditions are true:
\begin{itemize}
	\item[(A)] $P$ is looking for a violation, and either
		\begin{enumerate}
			\item $P$'s search path from $location(P,C)$ contains $v$, or 
			
			\item $P$'s search path from the first unmarked node on $P$'s stack contains $v$, \textbf{and} either $P$ will find a violation on line~\ref{ln:cleanup:viol_check} before pushing an unmarked node onto its stack in a solo run from $C$, or has already found a violation on line~\ref{ln:cleanup:viol_check} in its current attempt, or
	
		\end{enumerate}
	\item[(B)] $P$ is backtracking, and either
		\begin{enumerate}
			\item $P$'s search path from $location(P,C)$ contains $v$, \textbf{and} $location(P,C)$ is unmarked, or
			\item $P$'s search path from the first unmarked node on its stack contains $v$, \textbf{and} $location(P,C)$ is marked.
		\end{enumerate}
\end{itemize}
\end{lemma}

\begin{proof} We assume the lemma holds for all violations in the configuration $C$ before a step $s$ taken by a process $Q$, and show the lemma holds for all violations in the configuration $C'$ immediately after $s$. In the initial configuration, there are no violations in the chromatic tree, so the lemma holds trivially in the base case. 

Consider an arbitrary violation $v$ in $C$ that was created by a process $P$. If the step $s$ by $Q$ only changes its local state, then the invariant can only be falsified for $v$ if $Q = P$. The only changes to $P$'s local state that affect the truth of the invariant for $v$ are those that change $P$'s local variable $l$, modify $P$'s local stack, cause $P$ to stop backtracking, or cause $P$ to stop looking for a violation. Such changes occur on lines~\ref{ln:cleanup:entry}, \ref{ln:cleanup:first_pop}, \ref{ln:cleanup:p_pop}, \ref{ln:cleanup:push} and \ref{ln:cleanup:leaf_update}, when reading that $location(P,C)$ is unmarked on line~\ref{ln:cleanup:backtracking_start}, or when reading that $location(P,C)$ is a leaf on line~\ref{ln:cleanup:leaf}.

The other possibility is that $s$ is a successful SCX from an invocation of \textsc{TryInsert}, \textsc{TryDelete}, \textsc{TryRebalance}, or \textsc{Help}. In this case, $s$ may affect the truth of the invariant for all violations. 

We consider each type of step $s$ and argue (A) or (B) is satisfied for each violation in $C'$.
\begin{itemize}
	\item Case 1: Suppose $s$ is a step by $P$ that sets the local variable $l$ to $entry$ on line~\ref{ln:cleanup:entry}. By Observation~\ref{stack_obs}.\ref{stack_obs:empty}, backtracking never pops from an empty stack, so $P$ must  be entering the cleanup phase after creating a violation $v$ in its update phase. By Lemma~\ref{v_path_entry}, a violation created by $P$ is on the search path for $P$ starting from $entry$. Therefore, (A1) is true for $v$ in $C'$.
	
	\item Case 2: Suppose $s$ is a step by $P$ that updates the local variable $l$ on line~\ref{ln:cleanup:first_pop} by popping from the stack. By definition, $P$ is looking for a violation in $C$, and is backtracking in $C'$. Therefore, (A) is true for $v$ in $C$, and we show (B) is true for $v$ in $C'$. 
	
	Let $m$ be the first unmarked node on $P$'s stack in $C$. We argue $v$ is on $P$'s search path from $m$ in $C'$. If (A1) is true for $v$ in $C$, then by Lemma~\ref{viol_loc}, $v$ is on $P$'s search path from $m$ in $C$. If (A2) is true for $v$ in $C$, then by the induction hypothesis, $v$ is on $P$'s search path from $m$ in $C$. In either case, since $s$ does not modify the chromatic tree, $v$ is on $P$'s search path from $m$ in $C'$.  
	
	If the node popped by $s$ is $m$, then $m = location(P,C')$. By definition of $m$, $location(P,C')$ is unmarked and $v$ is on $P$'s search path from $location(P,C')$ in $C'$, so (B1) is true for $v$ in $C'$. If the node popped by $s$ is not $m$, then $location(P,C')$ is a marked node and $m$ is still the first unmarked node on $P$'s stack in $C'$. Therefore, so (B2) is true for $v$ in $C'$.

	\item Case 3: Suppose $s$ is a step by $P$ that reads $location(P,C)$ is unmarked on line~\ref{ln:cleanup:backtracking_start}, which causes $P$ to stop backtracking. Thus (B) is true for $v$ in $C$. Since $location(P,C)$ is unmarked, (B2) is false for $v$ in $C$, so (B1) is true in $C$. Therefore, $v$ is on $P$'s search path from $location(P, C)$. Since $location(P,C') = location(P,C)$ and $P$ is looking for a violation in $C'$, (A1) is true for $v$ in $C'$.
	
	\item Case 4: Suppose $s$ is a step by $P$ that updates the local variable $l$ by the pop on line~\ref{ln:cleanup:pop}. Then $P$ is backtracking in both $C$ and $C'$. The last check on line~\ref{ln:cleanup:backtracking_start} saw that $location(P,C)$ was marked, and so (B2) is true for $v$ in $C$. Let $m$ be the first unmarked node on $P$'s stack in $C$. By the induction hypothesis, $v$ is on $P$'s search path from $m$ in $C$. Since $s$ does not modify the chromatic tree, $v$ is on $P$'s search path from $m$ in $C'$.
	
	If the node popped by $s$ is $m$, then $m = location(P,C')$. By definition of $m$, $location(P,C')$ is unmarked and $v$ is on $P$'s search path from $location(P,C')$ in $C'$, so (B1) is true for $v$ in $C'$. If the node popped by $s$ is not $m$, then $location(P,C')$ is a marked node and $m$ is still the first unmarked node on $P$'s stack in $C'$. Therefore, so (B2) is true for $v$ in $C'$.

	\item Case 5: Suppose $s$ is a step by $P$ that pops from $P$'s stack on line~\ref{ln:cleanup:p_pop} or \ref{ln:cleanup:gp_pop}. Then $P$ is looking for a violation in $C$ and $C'$. Note that $s$ does not modify $v$ or $P$'s local variable $l$, so if (A1) is true for $v$ in $C$, then (A1) is true for $v$ in $C'$.
	
	So suppose (A2) is true for $v$ in $C$. Let $m$ be the first unmarked node on $P$'s stack in $C$, and $m'$ be the first unmarked node on $P$'s stack in $C'$. (Note that either $m = m'$, or $m'$ is the second unmarked node on $P$'s stack in $C$.) By Lemma~\ref{on_stack_SP}, both $m$ and $m'$ are on $P$'s search path in $C$. By Lemma~\ref{stack_ancestors}, $m$ is on the $P$'s search path from $m'$ in $C$. By the induction hypothesis, $v$ is on $P$'s search path from $m$ in $C$. Therefore, $v$ is on $P$'s search path from $m'$ in $C$. Since $s$ does not modify the chromatic tree, $v$ is on $P$'s search path from $m'$ in $C'$. Therefore, (A2) is true for $v$ in $C'$.
	
	\item Case 6: Suppose $s$ is a step by $P$ where $P$ reads that $location(P,C)$ is a leaf, causing $P$ to exit its cleanup phase on line~\ref{ln:cleanup:leaf}. By the check on line~\ref{ln:cleanup:viol_check} and Lemma~\ref{no_viol_added}, there are no violations at $location(P,C)$. Since $P$ is looking for a violation, and there are no violations on $P$'s search path from $location(P,C)$, $P$ does not satisfy (A) for any violation the chromatic tree in $C$. By the induction hypothesis, the lemma is true in $C$, so $P$ cannot be the creator of any remaining violation in the chromatic tree in $C$. Since step $s$ does not create a violation, $s$ cannot falsify the lemma.
	
	\item Case 7: Suppose $s$ is a step by $P$ that pushes a pointer to $location(P,C)$ onto $P$'s stack on line~\ref{ln:cleanup:push}. Since $P$ is looking for a violation in both $C$ and $C'$, (A) is true for $v$ in  $C$, and we show (A) is true for $v$ in $C'$. If (A1) is true for $P$ in $C$, then (A1) is true for $v$ in $C'$ since the local variable $l$ is not changed. So suppose (A2) is true for $v$ in $C$. By the induction hypothesis, $P$ does not push an unmarked node onto its stack until after finding a violation on line~\ref{ln:cleanup:viol_check}. Since $s$ is the first step in its solo execution, $location(P,C)$ is a marked node and is pushed onto the stack by $s$. Therefore, neither condition of (A2) can be falsified by $s$, and so (A2) is true for $v$ in $C'$.
		
	\item Case 8: Suppose $s$ is a step by $P$ that updates the local variable $l$ on line~\ref{ln:cleanup:forward_end}. Since $P$ is looking for a violation in $C$, (A) is true for $v$ in $C$. 
	
	Suppose (A1) is true for $v$ in $C$, so $v$ is on $P$'s search path from $location(P,C)$. In the previous step by $P$, $location(P,C)$ was pushed onto the stack, and so by Lemma~\ref{no_viol_stack}, there is no violation at $location(P,C)$. From the code, $s$ changes $l$ to point to the next node on its search path. Since $v$ is not at $location(P,C)$, $v$ is on $P$'s search path from $location(P,C')$. Therefore (A1) is true for $v$ in $C'$.
	
	If (A2) is true for $v$ in $C$, then $P$ takes the first step in its solo execution from $C$. Since $s$ does not push a pointer to a node onto the stack, neither condition of (A2) can be falsified by $s$, so (A2) is true for $v$ in $C'$.
	
	\item Case 9: Suppose $s$ is a successful SCX of any chromatic tree transformation performed by some process $Q$. Let $u$ be the node whose pointer is updated by $s$, $R$ be the set of nodes removed from the chromatic tree by $s$, and $A$ be the set of nodes added to the chromatic tree by $s$.
	
	The invariant cannot be falsified by any violations removed by $s$. If a new violation is created by $s$ while performing an \textsc{Insert} or \textsc{Delete} transformation, then $Q$ enters the cleanup phase starting from $entry$, and so (A1) is true for the new violation in $C'$.
	
	So consider a violation that is not created or removed by $s$ (so it is either moved or unmodified by $s$). Let $v$ be this violation, whose creator is $P$. Since $s$ does not update $P$'s local variable $l$, $location(P,C) = location(P,C')$. We consider multiple cases, depending on whether (B1), (B2), (A2) or (A1) is true for $v$ in $C$.
	\begin{itemize}
		\item Suppose (B1) is true for $v$ in $C$. Then $location(P,C)$ is unmarked in $C$. Suppose $location(P,C)$ is still unmarked in $C'$. By Lemma~\ref{unmarked}, $location(P,C)$ is in the chromatic tree in both $C$ and $C'$. Since $P$ is backtracking in $C$, $location(P,C)$ was previously popped from the stack. By Lemma~\ref{no_viol_added} and Lemma~\ref{no_viol_stack}, there are no violations at $location(P,C)$ in $C$. By the induction hypothesis, $v$ is on $P$'s search path from $location(P,C)$ in $C$. Therefore, by Lemma~\ref{violation_stay}, $v$ is on $P$'s search path from $location(P,C')$ in $C'$, so (B1) is true for $v$ in $C'$. 
		
		Otherwise $location(P,C)$ is marked in $C'$ and so (B1) is false for $v$ in $C'$. Let $m'$ be the first unmarked node on $P$'s stack in $C'$. By the induction hypothesis, $v$ is on $P$'s search path from $location(P,C)$ in $C$, so by Lemma~\ref{viol_loc}, $v$ is on $P$'s search path from $m'$ in $C$. By Lemma~\ref{no_viol_stack}, $m'$ does not contain any violations, and by Lemma~\ref{unmarked}, $m'$ is in the chromatic tree in $C'$. Therefore, by Lemma~\ref{violation_stay}, $v$ is on $P$'s search path from $m'$ in $C'$. Therefore, (B2) is true for $v$ in $C'$.
		
		\item Suppose (B2) is true for $v$ in $C$. Let $m$ be the first unmarked node on $P$'s stack in $C$, and let $m'$ be the first unmarked node on $P$'s stack in $C'$. Note that $m = m'$ if $s$ does not remove $m$ from the chromatic tree. By Lemma~\ref{stack_ancestors}, $m$ is on $P$'s search path from $m'$ in $C$. By the induction hypothesis, $v$ is on $P$'s search path from $m$ in $C$. Therefore, $v$ is on $P$'s search path from $m'$ in $C$. By Lemma~\ref{no_viol_added} and Lemma~\ref{no_viol_stack}, there are no violations at $location(P,C)$. Therefore, by Lemma~\ref{violation_stay}, $v$ is on $P$'s search path from $m'$ in $C'$. So (B2) is true for $v$ in $C'$. 
		
		\item Suppose (A2) is true for $v$ in $C$. By the induction hypothesis, there is a path of marked nodes on $P$'s search path from $location(P,C)$ to a node where $P$ will fail the check on line~\ref{ln:cleanup:viol_check} if run solo from $C$. The pointers of marked nodes cannot be modified, so this path of marked nodes is not modified by $s$. So the second condition of (A2) remains true after $s$.
		
		Let $m$ be the first unmarked node on $P$'s stack in $C$, and $m'$ be the first unmarked node on $P$'s stack in $C'$. By Lemma~\ref{stack_ancestors}, $m$ is on $P$'s search path from $m'$ in $C$. By the induction hypothesis, $v$ is on $P$'s search path from $m$ in $C$. Therefore, $v$ is on $P$'s search path from $m'$ in $C$. By Lemma~\ref{unmarked}, $m'$ is in the chromatic tree in $C$ and $C'$, and by Lemma~\ref{no_viol_stack}, $m'$ contains no violations. Therefore, by Lemma~\ref{violation_stay}, $v$ is on $P$'s search path from $m'$ in $C'$. Therefore, (A2) is true for $v$ in $C'$. 
		
		\item Suppose (A1) is true for $v$ in $C$. If $s$ does not move $v$, then by the induction hypothesis and Lemma~\ref{violation_no_move}, $v$ is on the search path for $P$ from $location(P,C')$ in $C'$. Then (A1) is true for $v$ in $C'$. 
		
		So suppose $v$ is moved by $s$ from a node $r$ to a node $r'$, where $r' \in A$ by Observation~\ref{move_obs}.\ref{move_obs:A}. Let $nextLoc(P,C)$ be the first node on $P$'s search path from $location(P,C)$ that is in the chromatic tree in $C$. Note that if $location(P,C)$ is in the chromatic tree in $C$, then $nextLoc(P,C) = location(P,C)$. By definition, violations are located at nodes in the chromatic tree, so $r$ is in the chromatic tree in $C$. By the induction hypothesis, $r$ is on $P$'s search path from $location(P,C)$ in $C$. Hence, $r$ is on $P$'s search path from $nextLoc(P,C)$ in $C$. 
		
		Consider the unique path of nodes in the chromatic tree on $P$'s search path from $entry$ to $r$ in $C$. Since $nextLoc(P,C)$ is in the chromatic tree, $nextLoc(P,C)$ is a node on this path. By Observation~\ref{move_obs}.\ref{move_obs:blk}, $r \in R$ or $r$ is a child of a node in $R$, and since each node in $R$ is reachable from $u$, $u$ is also a node on this path.
		
		Suppose $nextLoc(P,C)$ is an ancestor of $u$ in $C$, then $s$ does not modify any nodes on the path from $nextLoc(P,C)$ to $u$. So $nextLoc(P,C)$ is an ancestor of $u$ in $C'$. Since $r'$ is reachable $u$, and $u$ is reachable from $nextLoc(P,C)$, it follows from Lemma~\ref{v_reachable} that $r'$ is on $P$'s search path from $nextLoc(P,C)$ in $C'$. The nodes on the path from $location(P,C)$ to $nextLoc(P,C)$ are unmodified by $s$ because fields of nodes not in the chromatic tree do not change. So $v$ is on $P$'s search path from $location(P,C)$ in $C'$, and (A1) is true for $v$ in $C$. 
		
		So suppose $nextLoc(P,C)$ is a proper descendant of $u$. Then $nextLoc(P,C)$ is a node on the path from $u$ to $r$ (excluding $u$) in $C$. We first argue the second condition of (A2) is true for $v$ in $C'$. By Observation~\ref{move_obs}.\ref{move_obs:blk}, either $r \in R$ or $r$ is a child of a node in $R$. By Observation~\ref{trans_obs}.\ref{trans_obs:R} and the fact that the nodes in $R$ are connected, all nodes on the path from $u$ to the parent of $r$ (excluding $u$) are removed from the chromatic tree. Therefore, all nodes from $nextLoc(P,C)$ to $r$'s parent are marked in $C'$. If $P$ has not already found a violation on line~\ref{ln:cleanup:viol_check} in its current attempt, $P$ will only push marked nodes onto its stack in a solo run in $C'$ until it visits $r$, or it finds a violation before reaching $r$. Note that if $P$ visits $r$ in its solo run in $C'$, the check on line~\ref{ln:cleanup:viol_check} will evaluate true and $P$ will not push $r$ onto its stack. If $r$ was an overweight violation in $C$, then $P$ will find $r.w > 1$ on line~\ref{ln:cleanup:viol_check}. If $r$ was a red-red violation in $C$, then by Lemma~\ref{red_red_child}, all nodes with child pointers pointing to $r$ have weight 0. Therefore, $P$ will find a red-red violation on line~\ref{ln:cleanup:viol_check} since $r.w = 0$ and the last node pushed onto $P$'s stack before reaching $r$ has weight 0. Since $P$ will only push marked nodes onto its stack before finding a violation on line~\ref{ln:cleanup:viol_check}, the second condition of (A2) is satisfied in $C'$.
		
		To show the first condition of (A2) is true for $v$ in $C'$, let $m'$ be the first unmarked node on $P$'s stack in $C'$. By the induction hypothesis, $v$ is on $P$'s search path from $location(P,C)$ in $C$, so by Lemma~\ref{viol_loc}, $v$ is on $P$'s search path from $m'$ in $C$. By Lemma~\ref{no_viol_stack}, $m'$ contains no violations in $C$. Therefore, by Lemma~\ref{violation_stay}, $v$ is on $P$'s search path from $m'$ in $C'$. So (A2) is true for $v$ in $C'$. 

	\end{itemize}

\end{itemize}
Therefore, for any violation $v$ in the chromatic tree in $C'$, either condition (A) or (B) is satisfied. The lemma follows by induction.
\end{proof}

The next result follows from Lemma~\ref{clean_main_invariant} together with the fact that each operation creates at most 1 violation in its update phase, and none in its cleanup phase.

\begin{corollary}\label{num_viol}
The number of violations in the chromatic tree is bounded by the number of incomplete \textsc{Insert} and \textsc{Delete} operations.
\end{corollary}

\noindent Using Corollary~\ref{num_viol}, it is shown that the following theorem holds \cite{DBLP:conf/ppopp/BrownER14}.

\begin{theorem}\label{height_balanced}
If there are $c$ incomplete \textsc{Insert} and \textsc{Delete} operations and the data structure contains $n$ keys, then its height is $O(c + \log n)$.
\end{theorem}

\section{Chromatic Tree Amortized Analysis}\label{section_amortized}
In this section, we prove the amortized step complexity of the chromatic tree is proportional to the height of the tree plus the point contention of the execution. The structure of our proof is similar to the amortized step complexity of a binary search tree by Ellen et al~\cite{DBLP:conf/podc/EllenFHR13}. We first divide the steps taken by an operation $op$ by those taken in its update phase $up$ and cleanup phase $cp$. The total steps taken by a process for an operation $op$ is then\index{$steps(op)$} $steps(op) = steps(up) + steps(cp)$.

We divide the steps further by expressing $steps(up)$ and $steps(cp)$ as a function of the number of attempts and the number of pushes performed for $up$ and $cp$ respectively. 
\begin{definition}\normalfont For an update operation $op$ with update phase $up$ and cleanup phase $cp$, let:
	\begin{itemize}
		\item $attempts(up)$\index{$attempts(up)$} be the number of times $up$ performs \textsc{TryInsert} or \textsc{TryDelete} during its execution interval,
		\item $attempts(cp)$\index{$attempts(cp)$} be the number of times $cp$ performs \textsc{TryRebalance} during its execution interval,
		\item $pushes(up)$\index{$pushes(up)$} be the number of times $up$ performs \textsc{Push} across all of its calls to \textsc{BacktrackingSearch}, and
		\item $pushes(cp)$\index{$pushes(cp)$} be the number of times $cp$ performs \textsc{Push} during \textsc{BacktrackingCleanup}.
	\end{itemize}
\end{definition}

\begin{lemma}\normalfont\label{step_breakdown}
	The number of steps taken by an operation $op$ with update phase $up$ and cleanup phase $cp$ is $steps(op) = O(attempts(up) + attempts(cp) + pushes(up) + pushes(cp))$. If $op$ has no cleanup phase, then $steps(op) = O(attempts(up) + pushes(up))$.
\end{lemma}

\begin{proof}
First we consider the number of steps taken in the subroutines \textsc{TryInsert}, \textsc{TryDelete}, and \textsc{TryRebalance}. These subroutines each perform a constant number of LLXs, followed by a single SCX. From the code, an invocation of SCX$(V, R, fld, new)$ takes $O(1)$ steps, since $|V| \leq 6$ for all chromatic tree transformations. It also follows from the code that an LLX takes $O(1)$ steps. Therefore, a single invocation of \textsc{TryInsert}, \textsc{TryDelete}, or \textsc{TryRebalance} takes $O(1)$ steps. By definition, the number of times $up$ performs \textsc{TryInsert} or \textsc{TryDelete} is $attempts(up)$, and so the total number of steps taken is $O(attempts(up))$. Similarly, the total number of steps taken during all instances of \textsc{TryRebalance} invoked by $cp$ is $O(attempts(cp))$.

All that remains is to count the number of steps taken by searches during the update and cleanup phases. Let $pops(up)$ and $pops(cp)$ be the number of times $op$ performs \textsc{Pop} when backtracking across all its calls to \textsc{BacktrackingSearch} and \textsc{BacktrackingCleanup} respectively. Every call of \textsc{Pop} during backtracking is performed on a non-empty stack by Observation~\ref{stack_obs}.\ref{stack_obs:empty}. Therefore, $pops(up) \leq pushes(up)$ and $pops(cp) \leq pushes(cp)$. It is therefore sufficient to count the number of pushes made during $up$ and $cp$. By definition, this number is $pushes(up)$ for $up$, so the total number of steps taken during all instances of \textsc{BacktrackingSearch} invoked by $up$ is $O(pushes(up))$. Likewise, the total number of steps taken during the instance of \textsc{BacktrackingCleanup} invoked by $cp$ is $O(pushes(cp))$. The lemma follows by adding the steps taken by the update phase and cleanup phase of an operation.
\end{proof}

In order to express the step complexity of operations for the chromatic tree, we use the following notions of height.
\begin{definition}\normalfont
	For each configuration $C$ and operation $op$ with update phase $up$ and cleanup phase $cp$ of an execution $\alpha$, let:
	\begin{itemize}
		\item $h(C)$\index{$h(C)$} be the height of the chromatic tree in configuration $C$,
		\item $h(up)$\index{$h(up)$} and $h(cp)$\index{$h(cp)$} be the height of the chromatic tree in the starting configuration of $up$ and $cp$ respectively, and
		\item $h(op)$\index{$h(op)$} be the maximum height of the chromatic tree over all configurations during the execution interval of $op$.
	\end{itemize}
\end{definition}

\noindent For each operation $op$ with update phase $cp$ and cleanup phase $cp$ in an execution $\alpha$, let $\dot{c}(up)$\index{$\dot{c}(up)$} and $\dot{c}(cp)$\index{$\dot{c}(cp)$} be the maximum number of active operations in a single configuration during the execution interval of $up$ and $cp$ respectively. Let $viol(cp)$\index{$viol(cp)$} be the violation created by $up$, or \textsc{Nil} if $up$ did not create a violation. Several rebalancing transformations may be applied to a violation $v$ before it is removed. For a particular violation $v$ in an execution $\alpha$, let $rebal(v)$\index{$rebal(v)$} be the number of times $v$ is located at the \textit{center} of a successful rebalancing transformation in $\alpha$. Let $rebal(\textsc{Nil}) = 0$.

In Section~\ref{section_pushes}, we will prove that for any finite execution $\alpha$,
\begin{multline*}
\sum_{up \in \alpha} pushes(up) + \sum_{cp \in \alpha} pushes(cp) \leq \sum_{up \in \alpha} [2 \cdot attempts(up) + h(up) + 4\dot{c}(up)] \\ +  \sum_{cp \in \alpha} [3 \cdot attempts(cp) + h(cp) + 10\dot{c}(cp) \cdot rebal(viol(cp))]
\end{multline*}

\noindent Section~\ref{section_attempts} is devoted to showing that for any finite execution $\alpha$,
\begin{equation*}
\sum_{up \in \alpha} attempts(up) \leq \sum_{up \in \alpha} [30h(up) + 5594\dot{c}(up) + 159]
\end{equation*}
and
\begin{equation*}
\sum_{cp \in \alpha} attempts(cp) \leq \sum_{cp \in \alpha} [78h(cp) + 312\dot{c}(cp) + 5008\dot{c}(cp) \cdot rebal(viol(cp)) + 357].
\end{equation*}

\noindent Combining these main results from Section~\ref{section_pushes} Section~\ref{section_attempts} with Lemma~\ref{step_breakdown} gives the following lemma.

\begin{lemma}\normalfont\label{main_theorem} For any finite execution $\alpha$ of \textsc{Insert}, \textsc{Delete}, and \textsc{Find} operations on the chromatic tree,
\begin{equation*}
\begin{aligned}
\sum_{op \in \alpha} steps(op) = O\bigg( \sum_{up \in \alpha} [h(up) + \dot{c}(up)] + \sum_{cp \in \alpha} [h(cp) + rebal(viol(cp)) \cdot \dot{c}(cp)]\bigg).
\end{aligned}
\end{equation*}
\end{lemma}

The height terms $h(up)$ and $h(cp)$ in Lemma~\ref{main_theorem} can be interpreted as the steps required to traverse the chromatic tree during searches in the update phase and cleanup phase respectively. The $\dot{c}(up)$ terms account for extra steps taken by operations concurrent with $up$ due to the successful insertion or deletion made by $up$. Similarly, the $rebal(viol(cp)) \cdot \dot{c}(cp)$ terms account for extra steps taken due to successful rebalancing transformations centered at $viol(cp)$. 

Lemma~\ref{main_theorem} and Theorem~\ref{thm_rebal} together prove the main theorem of our paper.
\begin{theorem}\label{main_theorem_amortized} The amortized number of steps made by a single operation $op$ in any finite execution $\alpha$ is 
\begin{equation*}
O(\dot{c}(\alpha) + \log n(op)).
\end{equation*}
\end{theorem}

\begin{proof}
For any operation $op$ with update phase $up$ and cleanup phase $cp$, $h(up) \leq h(op)$ and $h(cp) \leq h(op)$ since the starting configurations of $up$ and $cp$ are contained in $op$,
\begin{equation*}
\sum_{up \in \alpha} h(up) + \sum_{cp \in \alpha} h(cp)  \leq \sum_{op \in \alpha} 2h(op).
\end{equation*}
Similarly, $\dot{c}(up) \leq \dot{c}(\alpha)$ and $\dot{c}(cp) \leq \dot{c}(\alpha)$ since $up$ and $cp$ are intervals contained in $\alpha$. Therefore, Lemma~\ref{main_theorem} can be expressed as
\begin{equation*}
\sum_{op \in \alpha} steps(op) = O\bigg( \sum_{op \in \alpha} [h(op) +  \dot{c}(\alpha)] + \dot{c}(\alpha) \cdot \sum_{cp \in \alpha} rebal(viol(cp)) \bigg).
\end{equation*}
Suppose the execution $\alpha$ contains $i$ \textsc{Insert} operations and $d$ \textsc{Delete} operations. Then, by Theorem~\ref{thm_rebal},
\begin{equation*}
\begin{aligned}
\dot{c}(\alpha) \cdot \sum_{cp \in \alpha} rebal(viol(cp)) &\leq \dot{c}(\alpha)(3i + d -2) \\
&\leq \dot{c}(\alpha)(3i + 3d) \\
&\leq \sum_{op \in \alpha} 3\dot{c}(\alpha).
\end{aligned}
\end{equation*}
So, the total number of steps taken by all operations in an execution $\alpha$ is
\begin{equation*}
\sum_{op \in \alpha} steps(op) = O\bigg( \sum_{op \in \alpha} [h(op) +  \dot{c}(\alpha)] \bigg).
\end{equation*}
It follows that the amortized step complexity of a single operation $op$ is $O(h(op) + \dot{c}(\alpha))$. By Theorem~\ref{height_balanced}, this can be expressed as $O(\dot{c}(\alpha) + \log n(op))$.
\end{proof}

Notice that the amortized cost of each operation has an additive term of $O(\dot{c}(\alpha))$, as opposed to $O(\dot{c}(op))$. This is because the operation with the largest contention in the execution may perform the majority of rebalancing. Theorem~\ref{thm_rebal} only gives an upper bound on the total number of rebalancing transformations that may occur in an execution, but does not state anything about which operations will perform these rebalancing transformations.

\subsection{Bounding the Number of Pushes}\label{section_pushes}
In this section we provide an upper bound on $pushes(up)$ and $pushes(cp)$. To do so, we classify all pushes into one of three categories, depending on properties of the node that was pushed onto the stack. We use an aggregate analysis to give separate bounds on the number of pushes that can occur for each category.

\begin{definition}\normalfont We classify instances of \textsc{Push} during \textsc{BacktrackingSearch} (or \textsc{BacktrackingCleanup}) for the update phase $up$ (or cleanup phase $cp$) of an operation into three categories:
\begin{itemize}
	\item Category 1: 
	\begin{itemize}
		\item The node pushed is later popped on line~\ref{ln:search:pop_p} in some instance of \textsc{BacktrackingSearch} invoked by $up$, or the last node popped on line~\ref{ln:search:pop_l} or \ref{ln:search:pop} in a later instance of \textsc{BacktrackingSearch} invoked by $up$, or
		
		\item the node pushed is later popped on line~\ref{ln:cleanup:p_pop} or \ref{ln:cleanup:gp_pop} in $cp$'s instance of \textsc{BacktrackingCleanup}, or the last node popped on line~\ref{ln:cleanup:first_pop} or \ref{ln:cleanup:pop} in a later attempt of \textsc{BacktrackingCleanup}.
	\end{itemize}
	\item Category 2: A push of a node $x$ that is not in Category 1, and $x$ is in the chromatic tree for all configurations during the execution interval of $up$ (or $cp$).
	\item Category 3: A push of a node $x$ that is not in Category 1, and $x$ is not in the chromatic tree in some configuration during the execution interval of $up$ (or $cp$).
\end{itemize}
\end{definition}

Every push is in exactly one of these three categories. Category 1 pushes are of nodes that are later popped to be passed as arguments into \textsc{TryInsert}, \textsc{TryDelete}, or \textsc{TryRebalance}, or are of nodes that are the last nodes popped during backtracking in attempts. We will show that a Category 1 push is the only type of push that may push the same node onto the same stack twice in a phase. Category 2 pushes are of nodes that are not popped off the stack. Category 3 pushes are those that are performed on nodes added or removed from the chromatic tree by concurrent operations.
 
\begin{lemma}\normalfont\label{repeatpush}
The number of Category 1 pushes in an execution $\alpha$ is at most 
\begin{equation*}
\sum_{up \in \alpha} 2 \cdot attempts(up) + \sum_{cp \in \alpha} 3 \cdot attempts(cp).
\end{equation*}
\end{lemma}

\begin{proof}
Consider an operation $op$ with update phase $up$ and cleanup phase $cp$. In each invocation of \textsc{BacktrackingSearch} for $up$, there is at most 1 node that is the last node popped when backtracking, and at most 1 node popped on line~\ref{ln:search:pop_p}. So there are at most $2 \cdot attempts(up)$ Category 1 pushes by $up$ during its invocations of \textsc{BacktrackingSearch}. 
 
Similarly for \textsc{BacktrackingCleanup}, there is at most 1 node that is the last node popped during backtracking per cleanup attempt. Additionally, each cleanup attempt will pop a total of at most 2 nodes on lines~\ref{ln:cleanup:p_pop} and \ref{ln:cleanup:gp_pop}. So there are at most $3 \cdot attempts(cp)$ Category 1 pushes by $cp$ during its invocation of \textsc{BacktrackingCleanup}. The lemma follows from taking the sum of Category 1 pushes over all update and cleanup phases in an execution $\alpha$.
\end{proof}

The following properties of the implementation of LLX and SCX used in the chromatic tree were proven in \cite{DBLP:conf/podc/BrownER13}.
\begin{observation}\normalfont\label{llx_obs_2}
	\hfill
	\begin{enumerate}
\item \label{llx_obs_2:marked} Suppose a successful mark step $mstep$ belonging to an SCX-record $U$ on $r$ occurs. Then $r$ is frozen for $U$ when $mstep$ occurs, and forever after. 

\item \label{llx_obs_2:f_step} There cannot be both a frozen step and an abort step belonging to the same SCX-record. 

\item \label{llx_obs_2:help} Suppose a frozen step is performed for $U$. Then a commit step is performed for $U$ before any invocation of \textsc{Help}$(ptr)$ terminates, where $ptr$ is a pointer to $U$.
	\end{enumerate}
\end{observation}

\noindent The above properties are used in the following lemma.

\begin{lemma}\label{backtracking_marked}\normalfont
	Consider an attempt $A$ of an update or cleanup phase $xp$. Suppose $xp$ reads that a node $x$ is marked on line~\ref{ln:search:backtracking_start} of \textsc{BacktrackingSearch}, (or on line~\ref{ln:cleanup:backtracking_start} of \textsc{BacktrackingCleanup}). Then $x$ is not a node in the chromatic tree when the next pop by $xp$ on line~\ref{ln:search:pop} (or on line~\ref{ln:cleanup:pop}) is executed.
\end{lemma}

\begin{proof}
We prove the case for when $xp$ is an update phase. The case for when $xp$ is a cleanup phase follows similarly. 

Let $U$ be the SCX-recorded pointed to by $x.\mathit{info}$ when $x$ was marked. By Observation~\ref{llx_obs_2}.\ref{llx_obs_2:marked}, $x.\mathit{info}$ points to $U$ in all configurations after this point, so $x.\mathit{info}$ points to $U$ when $xp$ reads $x.\mathit{info.state}$ on line~\ref{ln:search:help}. From the code of LLX, if a mark step belonging to $U$ is performed, then a frozen step belonging to $U$ has been performed. By Observation~\ref{llx_obs_2}.\ref{llx_obs_2:f_step}, an abort step does not belong to $U$, and so $U$ is not in the Aborted state when $xp$ reads $x.\mathit{info.state}$ on line~\ref{ln:search:help}. So $xp$ either reads $x.\textit{info.state}$ is either InProgress or Committed. In either case, we argue there exists a configuration $C$ during $A$ in which $x$ is marked and $x.\mathit{info}$ points to a Committed SCX-record. 

If $xp$ reads that $U$ is in the Committed state, let $C$ be the configuration immediately after the read. If $xp$ reads that $U$ is in the InProgress state, then $xp$ executes \textsc{Help}$(x.\mathit{info})$ on line~\ref{ln:search:help}. By Observation~\ref{llx_obs_2}.\ref{llx_obs_2:help}, a commit step belonging to $U$ is performed before this invocation of \textsc{Help} terminates. Let $C$ be the configuration immediately after this commit step. In configuration $C$, $x$ is marked and $x.\mathit{info}$ points to a Committed SCX-record, and so by definition is not a node in the chromatic tree.
\end{proof}

\begin{lemma}\normalfont\label{unreachable_pop}
Suppose an update or cleanup phase $xp$ pushes a node $x$ onto its stack, and this push is not Category 1.  If $x$ is later popped off the stack, then before the next pop by $xp$, $x$ is not in the chromatic tree.
\end{lemma}

\begin{proof}
Since $x$ is not pushed onto the stack by a Category 1 push, it is popped during backtracking. By definition, $x$ is not the last node popped during backtracking. Therefore, when $xp$ next executes line~\ref{ln:search:backtracking_start} of \textsc{BacktrackingSearch} (or \ref{ln:cleanup:backtracking_start} of \textsc{BacktrackingCleanup}), $xp$ will read that $x$ is marked. By Lemma~\ref{backtracking_marked}, $x$ is not a node in the chromatic tree in some configuration before the next pop by $xp$.
\end{proof}

\begin{lemma}\normalfont\label{repeatpush2}
In each of an operation $op$'s update and cleanup phases, $op$ pushes a node onto its stack at most once by a \textsc{Push} that is not a Category 1 push.
\end{lemma}

\begin{proof}
Suppose, for contradiction, that a particular node $x$ is pushed onto the stack twice during an update phase or cleanup phase $xp$ by two pushes $push_1$ and $push_2$ that are not in Category 1. Without loss of generality, assume $push_1$ occurs before $push_2$. Only the node pointed to by $xp$'s local variable $l$ is pushed onto the stack. Hence by Lemma~\ref{no_cycle}, $x$ is not on $xp$'s stack in the configuration immediately before $push_2$. Therefore, $x$ is popped off the stack between $push_1$ and $push_2$. By Lemma~\ref{unreachable_pop}, $x$ is not in the chromatic tree before the next pop by $xp$.

Let $C'$ be the configuration immediately after the last pop during backtracking in the attempt containing $push_2$. By definition, $x$ is not the last node popped during backtracking of some attempt, so at least 1 other node has been popped by $xp$ since the time $x$ was popped. Therefore, $x$ is not in the chromatic tree in $C'$. 

In the configuration $C$ immediately before $push_2$, $xp$'s local variable $l$ points to $x$. Lemma~\ref{reach_SP_backtracking} (or Corollary~\ref{reach_SP_backtracking_cleanup} when $xp$ is a cleanup phase) implies $x$ is on $xp$'s search path in a configuration between $C'$ and $C$. Thus, $x$ is reachable from $entry$, which contradicts the fact that $x$ is not in the chromatic tree in $C'$.
\end{proof}

The next lemma informally states that a rebalancing transformation cannot cause a node $x$ to gain an ancestor that was already a node in the chromatic tree before the transformation.
\begin{lemma}\normalfont\label{remain_reachable}
Let $C$ be the first configuration in which two nodes $x$ and $y$ are both in the chromatic tree. If $x$ is the descendant of $y$ in a configuration $C'$ after $C$, then $x$ is the descendant of $y$ in all configurations between $C$ and $C'$. 
\end{lemma}

\begin{proof}
We prove by backwards induction on the sequence of configurations between $C$ and $C'$. We assume $x$ is a descendant of $y$ in $C'$ after a step $s$, and show $x$ is a descendant of $y$ in the configuration $C^-$ immediately before $C'$. 

The only type of step $s$ that changes the structure of the chromatic tree is the update CAS of some SCX. Let $u$ be node whose child pointer is updated by $s$, which removes a connected set of nodes $R$ from the chromatic tree, and adds a connected set of nodes $A$ to the chromatic tree. Note that since nodes removed from the chromatic tree are not added back to the chromatic tree, $x$ and $y$ are in the chromatic tree in all configurations between $C$ and $C'$. Note that $x$ and $y$ are not in $A$, otherwise Observation~\ref{trans_obs}.\ref{trans_obs:A} implies they are added to the chromatic tree by $s$ and so $x$ and $y$ were not in the chromatic tree in $C^-$. Additionally, $x$ and $y$ are not in $R$, otherwise Observation~\ref{trans_obs}.\ref{trans_obs:R} implies they are removed from the chromatic tree by $s$, and so $x$ and $y$ are not in the chromatic tree in $C'$. 

By Observation~\ref{trans_obs}.\ref{trans_obs:u}, $u$ is in the chromatic tree in both $C^-$ and $C'$, so we consider multiple cases depending on the location $u$ in the chromatic tree relative to $x$ and $y$. In particular, we consider the cases when $u$ is a proper ancestor of both $x$ and $y$, a proper ancestor of $x$ but not a proper ancestor of $y$, or a proper ancestor of neither $x$ nor $y$.
\begin{itemize}
\item Suppose $u$ is not a proper ancestor of $x$ or $y$ in $C'$. Since $s$ only modifies nodes in the subtree rooted at $u$, the path from $x$ to $y$ in $C'$ is not altered by the step $s$. Thus $x$ is the descendant of $y$ in $C^-$.

\item Suppose $u$ is a proper ancestor of $x$, but not a proper ancestor of $y$ in $C'$. So $u$ is on the path from $y$ to $x$ in $C'$.  Since $s$ does not modify the pointers of any nodes that are proper ancestors of $u$, there is a path from $y$ to $u$ in $C^-$. By Observation~\ref{trans_obs}.\ref{trans_obs:reach_after}, since $x$ is reachable from $u$ in $C'$ and $x \notin A$, $x$ is reachable from $u$ in $C^-$. Therefore, $x$ is a descendant of $y$ in $C^-$.

\item Suppose $u$ is a proper ancestor of both $x$ and $y$. By Observation~\ref{trans_obs}.\ref{trans_obs:reach_after}, since $x$ and $y$ are reachable from $u$ in $C'$ and $x,y \notin A$, $x$ and $y$ are reachable from $u$ in $C^-$. Consider the path of nodes from $y$ to $x$ in $C'$. Since $x$ and $y$ are not nodes in the sets $R$ and $A$, and the fact that $R$ and $A$ are connected sets of nodes, all nodes on this path are in the chromatic tree in $C^-$ and $C'$. Since $s$ does not modify the child pointers of any of these nodes, $x$ is a descendant of $y$ in $C^-$.
\end{itemize}
\end{proof}

\begin{lemma}\normalfont\label{necessary}
The number of Category 2 pushes in an execution $\alpha$ is at most 
\begin{equation*}
\sum_{up \in \alpha} h(up) + \sum_{cp \in \alpha} h(cp).
\end{equation*}
\end{lemma}

\begin{proof}
Consider an update phase or cleanup phase $xp$ of an operation. Let $C'$ be the last configuration of $xp$, and $C$ be the first configuration of $xp$. Any node pushed onto $xp$'s stack by a Category 2 push is in the chromatic tree from $C$ to $C'$, and hence by Lemma~\ref{unreachable_pop}, is not popped off the stack. Consider the nodes pushed onto $xp$'s stack by a Category 2 push in $C'$. By Lemma~\ref{on_stack_SP_2}, each of these nodes are on $xp$'s search path in $C'$. Let $x_1, \dots, x_m$ be the order in which these nodes appear on this path, where $x_{i+1}$ is a descendant of $x_{i}$, for $1 \leq i < m$. Since $x_{i}$ and $x_{i+1}$ are both in the chromatic tree in $C$, by Lemma~\ref{remain_reachable}, $x_{i+1}$ is a descendant of $x_{i}$ in $C$. Since there are at most $h(xp)$ internal nodes on a single path in the chromatic tree in $C$, Lemma~\ref{repeatpush2} implies there are at most $h(xp)$ Category 2 pushes during $xp$.
\end{proof}

\begin{lemma}\normalfont\label{extra} 
The number of Category 3 pushes in an execution $\alpha$ is at most
\begin{equation*}
\sum_{up \in \alpha} 4\dot{c}(up) + \sum_{cp \in \alpha} 10\dot{c}(cp) \cdot rebal(viol(cp)).
\end{equation*}
\end{lemma}

\begin{proof}
Consider the update phase or cleanup phase $xp$ of some operation. Nodes on the stack from Category 3 pushes are not in the chromatic tree at some point during $xp$, so they are either removed or added to the chromatic tree after $xp$ begins. Nodes can only become reachable or unreachable after an update CAS of a successful SCX. We charge each Category 3 push made by $xp$ to the update or cleanup phase that performed this update CAS. 

First consider an update CAS for an \textsc{Insert} or \textsc{Delete} transformation by an update phase $up$. An \textsc{Insert} causes at most 3 nodes to become reachable and 1 unreachable. A \textsc{Delete} causes 1 node to become reachable and 3 nodes to become unreachable. This causes every process concurrent with $up$ to perform at most 4 Category 3 pushes. By Lemma~\ref{repeatpush2}, each of these Category 3 pushes occur at most once, so at most $4\dot{c}(up)$ Category 3 pushes are charged to the update CAS by $up$. 

Next we consider Category 3 pushes that result from the update CASs of rebalancing transformations. Consider an update CAS of a rebalancing transformation centered at a node with violation $v$ in configuration $C$. It makes at most 5 nodes reachable and 5 nodes unreachable (for example, when the rebalancing transformation is W3). Then at most $10\dot{c}(C)$ Category 3 pushes are made among the $\dot{c}(C)$ active operations. We charge these $10\dot{c}(C)$ pushes to the cleanup phase $cp$ of the operation whose update phase created the violation (i.e.~$viol(cp) = v$). Since $C$ is in the execution interval of $cp$, it follows by definition that $10\dot{c}(C) \leq 10\dot{c}(cp)$. Recall that the number of times a rebalancing transformation is centered at $v$ in an execution $\alpha$ is $rebal(v)$. So the total amount charged to $cp$ throughout $\alpha$ is at most $10\dot{c}(cp) \cdot rebal(viol(cp))$. The total number of Category 3 pushes made in $\alpha$ is therefore at most $\sum_{up \in \alpha} 4\dot{c}(up) + \sum_{cp \in \alpha} 10\dot{c}(cp) \cdot rebal(viol(cp))$.
\end{proof}

\begin{lemma}\normalfont\label{thm_pushes}
For any finite execution $\alpha$,
\begin{multline*}
\sum_{up \in \alpha} pushes(up) + \sum_{cp \in \alpha} pushes(cp) \leq \sum_{up \in \alpha} [2 \cdot attempts(up) + h(up) + 4\dot{c}(up)] \\ +  \sum_{cp \in \alpha} [3 \cdot attempts(cp) + h(cp) + 10\dot{c}(cp) \cdot rebal(viol(cp))]
\end{multline*}
\end{lemma}

\begin{proof}
Every push performed by an operation $op$ in $\alpha$ is either Category 1, Category 2, or Category 3. Adding the upper bounds established in Lemmas~\ref{repeatpush}, \ref{necessary}, and \ref{extra} gives an upper bound on the number of pushes during update and cleanup phases in any finite execution $\alpha$.
\end{proof}

\subsection{Bounding the Number of Attempts}\label{section_attempts}
All that remains in the amortized analysis is to give an upper bound on $attempts(up)$ and $attempts(cp)$ for each operation in any finite execution. Recall that an attempt of an update phase begins with backtracking, and ends after a successful or unsuccessful \textsc{TryInsert} or \textsc{TryDelete}. An attempt of a cleanup phase begins with backtracking, and ends after a \textsc{TryRebalance} or if a leaf node with no violations is reached. In this section, we show the following main lemma.
\begin{lemma}\normalfont\label{thm_attempts}
For any finite execution $\alpha$,
\begin{equation*}
\sum_{up \in \alpha} attempts(up) \leq \sum_{up \in \alpha} [30h(up) + 5594\dot{c}(up) + 159]
\end{equation*}
and
\begin{equation*}
\sum_{cp \in \alpha} attempts(cp) \leq \sum_{cp \in \alpha} [78h(cp) + 312\dot{c}(cp) + 5008\dot{c}(cp)\cdot rebal(viol(cp)) + 357].
\end{equation*}
\end{lemma}

As described in Section~\ref{section_related}, the amortized analysis for counting attempts is difficult when failed attempts can cause failures. This happens for updates done using LLX and SCX. For example, consider an invocation $S$ of SCX for an operation $op$, where two nodes are required to be frozen for $S$ before the update CAS for $S$ is performed. If the first freezing CAS for $S$ on a node $x$ succeeds, $S$ can cause all concurrent SCX operations that perform freezing CASs on $x$ to fail. If the second freezing CAS for $S$ fails, $S$ performs an abort step that unfreezes $x$ and returns \textsc{False}. Operation $op$ has now caused a number of concurrent operations to fail without making any progress on completing its own operation. This argument can be extended to show that an operation can cause an unbounded number of failed attempts. Therefore, we cannot charge $op$ for all the failed attempts it causes.

Our analysis for counting the number of attempts made in the update phase closely follows the analysis done in \cite{DBLP:conf/podc/EllenFHR13}, but differs when counting the number of rebalancing attempts in the cleanup phase. As mentioned in the introduction, several factors make this more challenging. It was shown by Boyar and Larsen that removing a particular violation may require a number of rebalancing transformations proportional to the height of the chromatic tree \cite{DBLP:conf/wads/BoyarFL95}. Thus, the number of rebalancing transformations centered at the violation created by an operation may be large. Additionally, a process may take several failed attempts that try to remove a violation on its search path only to have a different process perform the successful rebalancing transformation to remove it. The rebalancing transformation needed to cleanup a violation may change as concurrent processes make nearby changes to the chromatic tree. Finally, violations can be moved onto a process's search path in the middle of its cleanup phase. We show how to extend the amortized analysis techniques used in \cite{DBLP:conf/podc/EllenFHR13} to circumvent these issues.

\subsubsection{Counting Successful Attempts}
Our analysis for counting the number of attempts performed in the update phase and cleanup phase of an operation is divided into counting successful and unsuccessful attempts, where
\begin{equation*}
attempts(up) = \mathit{successfulAttempts}(up) + \mathit{failedAttempts}(up)
\end{equation*}
and
\begin{equation*}
attempts(cp) = \mathit{successfulAttempts}(cp) + \mathit{failedAttempts}(cp).
\end{equation*}

Recall that an invocation of \textsc{TryInsert}, \textsc{TryDelete}, or \textsc{TryRebalance} is successful if a successful SCX is performed. The invocation is unsuccessful if it fails an LLX, SCX, or \textsc{Nil} check. We first count the number of successful attempts in an execution.

\begin{lemma}\normalfont\label{successful_attempts}
For a finite execution $\alpha$ with $i$ \textsc{Insert} and $d$ \textsc{Delete} operations, $\sum_{up \in \alpha} \mathit{successfulAttempts}(up) \leq i + d$ and $\sum_{cp \in \alpha} \mathit{successfulAttempts}(cp) \leq 3i + d - 2$.
\end{lemma}

\begin{proof} 
For every completed update phase of an operation, there is exactly 1 successful attempt, which is always the last attempt of the update phase. Therefore, there are at most $i + d$ successful attempts of \textsc{TryInsert} and \textsc{TryDelete}.

By Theorem~\ref{thm_rebal}, at most $3i + d - 2$ successful rebalancing transformations are sufficient to remove all violations from the chromatic tree. This corresponds to at most $3i + d - 2$ successful invocations of \textsc{TryRebalance} across the cleanup phases in an execution. 
\end{proof}

\subsubsection{Overview of the Accounting Method For Counting Failed Attempts}
In this section, we give upper bounds on $failedAttempts(up)$ and $failedAattempts(cp)$. Recall that a failed invocation of \textsc{TryInsert}, \textsc{TryDelete}, or \textsc{TryRebalance} is one that fails an LLX, SCX, or \textsc{Nil} check. Failures due to \textsc{Nil} checks can be handled separately and are explained in Section~\ref{section_nil}. The majority of the analysis will be for counting the number of failed attempts due to failed LLXs and SCXs.

Our analysis follows the accounting method. We use 2 types of bank accounts, one associated with update phases $up$ and one associated with cleanup phases $cp$. The rules that deposit and withdraw dollars from bank accounts are designed so that the \textit{bank} (i.e.~the collection of all bank accounts) satisfies the following properties for any execution $\alpha$.
\begin{itemize}
	\item (Property P1) The total number of dollars deposited into the bank by an operation is $O(h(up) + \dot{c}(up))$ during its update phase $up$, and $O(h(cp) + rebal(viol(cp)) \cdot \dot{c}(cp))$ during its cleanup phase $cp$.
	\item (Property P2) Every failed attempt by an operation in its update phase $up$ (or cleanup phase $cp$) withdraws one dollar from one of its own accounts. 
	\item (Property P3) All bank accounts have non-negative balance.
\end{itemize}
Assuming these properties of the bank hold, it follows that the total number of failed attempts made in an execution $\alpha$ is bounded by the total money deposited into the bank. Thus,
\begin{equation*}
\begin{aligned}
\mathit{failedAttempts}(\alpha) &= \sum_{up \in \alpha} \mathit{failedAttempts}(up) + \sum_{cp \in \alpha} \mathit{failedAttempts}(cp) \\
&= \sum_{op \in \alpha} O(h(up) + \dot{c}(up)) + \sum_{cp \in \alpha} O(h(cp) + rebal(viol(cp)) \cdot \dot{c}(cp)).
\end{aligned}
\end{equation*}

We note that a failed freezing CAS does not imply a failed SCX, as a helping process may have performed the freezing CAS on behalf of another process. The following definition is useful for characterizing when an SCX fails.
\begin{definition}\label{freezing_step}\normalfont
	Consider an invocation $S$ of SCX that creates an SCX-record $U$. A \textit{freezing iteration}\index{freezing iteration} of $S$ on a node $r$ is the iteration for $r$ of the for-each loop on lines 24-35 of \textsc{Help}$(U)$. A freezing iteration \textit{fails} if a successful abort step is performed during the freezing iteration, otherwise the freezing iteration is \textit{successful}.
\end{definition}
\noindent By definition, a SCX fails if and only if it performs a failed freezing iteration. It was shown that every unsuccessful invocation of LLX and SCX can blame some invocation of SCX for its failure \cite{DBLP:conf/podc/BrownER13}. We use the notion of blaming to help decide the bank account responsible to pay for a failed attempt caused by an unsuccessful LLX or SCX.
\begin{definition}\label{llx_blame}\normalfont
	Let $I$ be an invocation of LLX$(r)$ that returns \textsc{Fail} or \textsc{Finalized}. Let the \textit{failure step}\index{failure step} of $I$ be the step when $I$ reads $r.\mathit{info}$ on line 9 if it exists, otherwise it is the step $I$ reads $r.\mathit{info}$ on line 4. Let $U$ be the SCX-record pointed to by $r.\mathit{info}$ when $I$ performs its failure step. We say that $I$ \textit{blames}\index{blame} the invocation of SCX that created $U$.
\end{definition} 

\begin{definition}\label{scx_blame}\normalfont
	Consider an invocation $I$ of SCX$(V, R, fld, new)$ that returns \textsc{False}. Suppose its freezing iteration for a node $r \in V$ fails. Let the \textit{failure step}\index{failure step} of $I$ be the first successful freezing CAS on $r$ since the successful LLX$(r)$ linked to $I$. Let $U$ be the SCX-record pointed to by $r.\mathit{info}$ immediately after the failure step of $I$, and let $S$ be the invocation of SCX that created $U$. We say that $I$ \textit{blames}\index{blame} $S$ \textit{for} $r$.
\end{definition}

The failure step will be the step in which a dollar is withdrawn from a bank account to pay for a failed attempt. Note that each unsuccessful LLX and SCX has exactly one failure step. The failure step $s$ of an unsuccessful SCX $I$ may occur before the SCX is invoked, although it will occur in the same attempt after the LLX$(r)$ linked to $I$. If $s$ is the freezing CAS on a node $r$, then by definition $I$ fails its freezing iteration on $r$, and but was successful on all prior freezing iterations.

When referring to either an update or cleanup phase, we use the variable $xp$ to avoid repeating similar definitions. For an attempt $A$ of an update or cleanup phase $xp$, we define $\#llx(A)$ to be the number of LLX performed by $xp$ during $A$. Likewise, we define $\#\mathit{frz}(A)$ to be the number of freezing iterations performed by $xp$ during $A$. We define
$\#llx(xp) = \max_{A \in xp} \#llx(A)$ \index{$\#llx(xp)$} and $\#\mathit{frz}(xp) = \max_{A \in xp} \#\mathit{frz}(A)$ \index{$\#\mathit{frz}(xp)$}. For update phases $up$, it can be verified by inspection that $\#\mathit{frz}(up) = \#llx(up) \leq 2$ for \textsc{Insert} operations, and $\#\mathit{frz}(up) = \#llx(up) \leq 4$ for \textsc{Delete} operations. For cleanup phases $cp$, $\#\mathit{frz}(cp) \leq \#llx(cp) \leq 7$. 

We define $\mathit{LLXNode}_i(xp,A)$ \index{$\mathit{LLXNode}_i(xp,A)$} to be the $i$th node on which $xp$ performs an LLX in attempt $A$. Finally, consider an invocation $S$ of SCX$(V,R,fld,new)$ by $xp$ in attempt $A$. Let $\mathit{freezeNode}_i(xp,A)$ \index{$\mathit{freezeNode}_i(xp,A)$} be the $i$th node in the sequence $V$ enumerated in the order they are frozen. 

For each node $x$ in the chromatic tree at some point in the execution, and for each update or cleanup phase $xp$, we define the accounts $B(xp)$, $B_{llx}(xp,x)$ and $B_{scx}(xp,x)$, in addition to the auxiliary bank accounts $L_i(xp)$ (for $1 \leq i \leq \#llx(xp)$) and $F_i(xp)$ (for $1 \leq i \leq \#\mathit{frz}(xp)$). For each bank account $X(xp)$ (or $X(xp,x)$), where $X$ refers to one of the previously defined types of bank accounts, we define $X(xp, C)$  (or $X(xp,x,C)$) to be the number of dollars in $X(xp)$ in a configuration $C$. A bank account $X(xp)$ is \textit{active} if $xp$ is active. Each bank account is initially empty.

The auxiliary bank accounts will serve two purposes. First, they will pay for any failed attempts due to steps that occur prior to the start of a phase. Second, they will pay for failed attempts due to recent changes in the structure of the chromatic tree. Whenever the auxiliary accounts are empty, the $B_{llx}$ and $B_{scx}$ accounts will pay for failures due to failed LLXs and failed SCXs respectively. The $B(xp)$ account will be responsible for transferring dollars into the $B_{llx}(xp,x)$ and $B_{scx}(xp,x)$ accounts whenever some update or cleanup phase $xp'$ performs a successful freezing CAS on $x$ which causes a failed attempt by $xp$ that cannot be paid for by $xp$'s auxiliary accounts.

Each cleanup phase $cp$ will maintain additional $S(cp,x)$ accounts for each node $x$ in the chromatic tree.
\begin{definition}\normalfont
	An invocation $I$ of \textsc{TryRebalance}$(ggp,gp,p,v)$ starting from configuration $C$ is called \textit{stale}\index{stale} if one of $v$, $p$, $gp$, or $ggp$ is not in the chromatic tree in $C$. If $x$ is the first node in the sequence $\langle ggp,gp,p,v \rangle$ that is not in the chromatic tree in $C$, then the stale invocation $I$ \textit{blames}\index{blame} the node $x$.
\end{definition}
\noindent The $S(cp,x)$ accounts will pay for failed attempts resulting from stale invocations of \textsc{TryRebalance} that blame a node $x$. Finally, $cp$ will maintain a $B_{\mathit{nil}}(cp)$ account that will pay for all failed \textsc{TryRebalance} due to failed \textsc{Nil} checks by $cp$.

We next give a summary of all the bank account rules. Note that when referring to an update CAS or commit step $s$ \textit{for} an update or cleanup phase $xp$, $s$ may be performed by any process on behalf of the SCX invoked by $xp$. The deposit and withdraw rules for an update phase $up$ are summarized in Figures~\ref{fig_deposit_rules_update} and \ref{fig_withdraw_rules_update} respectively. The deposit and withdraw rules for a cleanup phase $cp$ are summarized in Figures~\ref{fig_deposit_rules_cleanup} and \ref{fig_withdraw_rules_cleanup} respectively. A final rule that transfers dollars between accounts by steps made by any operation is shown in Figure~\ref{fig_transfer_rules}.

\begin{figure}[htbp!]
	\begin{tabular}{ | p{3cm} | p{14cm} |}
	\hline
	D1-S & Consider a \textbf{successful update CAS} $ucas$ in configuration $C$ for an update phase $up$. For all nodes $x$ removed from the chromatic tree by $ucas$ and for all cleanup phases $cp'$ active during $C$, $up$ deposits 1 dollar into $S(cp', x)$. \\ \hline
	D1-LF & The \textbf{start of an update phase} $up$ deposits 4 dollars into each of its own $L_i(up)$ accounts (for $1 \leq i \leq \#llx(up)$) and $F_i(up)$ accounts (for $1 \leq i \leq \#\mathit{frz}(up)$). \\ \hline
	D2-LF & A \textbf{successful update CAS} for $up$ deposits 4 dollars into each active $L_i$ and $F_i$ account owned by update phases, and 3 dollars into each active $L_i$ and $F_i$ account owned by cleanup phases. \\ \hline
	D1-BUP & The \textbf{start of an update phase} $up$ deposits $30h(up) + 120\dot{c}(up) + 120$ dollars into its own $B(up)$ account, 468 dollars into each other active $B$ account. \\ \hline
	D2-BUP &  A \textbf{successful update CAS} for $up$ deposits 4934 dollars into each active $B$ account. \\ \hline
	D3-BUP & A \textbf{successful commit step} for $up$ deposits 26 dollars into each active $B$ account. \\ \hline
	D1-NIL & A \textbf{successful update CAS} for $up$ deposits 1 dollar into each active $B_{\mathit{nil}}$ account. \\ \hline
	\end{tabular}
	\caption{Steps for an update phase $up$ that deposit dollars into bank accounts.}
	\label{fig_deposit_rules_update}
\end{figure}

\begin{figure}[htbp!]
		\begin{tabular}{ | p{3cm} | p{14cm} |}
			\hline
			W-LLX & The \textbf{failure step of an LLX}, $I$, by $up$ in attempt $A$ on $x = \mathit{LLXNode}_i(up,A)$ withdraws 1 dollar from $L_i(up)$ if $L_i(up) > 0$. If $L_i(up) = 0$ and $x$ is a downwards node for the SCX blamed by $I$, then $up$ withdraws 1 dollar from $B_{llx}(up,x)$; otherwise it withdraws 1 dollar from $B_{llx}(up,p)$, where $p$ is the parent of $x$. \\ \hline
			W-SCX & The \textbf{failure step of an SCX}, $S$, by $up$ in attempt $A$ on $x = \mathit{freezeNode}_i(up,A)$ withdraws 1 dollar from $F_i(up)$ if $F_i(up) > 0$. If $F_i(up) = 0$ and $x$ is a downwards node for the SCX blamed by $S$, then $up$ withdraws 1 dollar from $B_{scx}(up,x)$; otherwise it withdraws 1 dollar from $B_{scx}(up,p)$, where $p$ is the parent of $x$. \\ \hline
		\end{tabular}
	\caption{Steps for an update phase $up$ that withdraw dollars from its own bank accounts.}
\label{fig_withdraw_rules_update}
\end{figure}
			
\begin{figure}[htbp!]
		\begin{tabular}{ | p{3cm} | p{14cm} |}
			\hline
			D1-S & Consider a \textbf{successful update CAS} $ucas$ in configuration $C$ for a rebalancing transformation centered at $viol(cp)$. For all nodes $x$ removed from the chromatic tree by $ucas$ and for all cleanup phases $cp'$ active during $C$, $cp$ deposits 1 dollar into $S(cp', x)$. \\ \hline
			D1-LF & The \textbf{start of an cleanup phase} $cp$ deposits 3 dollars into each of its own $L_i(cp)$ accounts (for $1 \leq i \leq \#llx(cp)$) and $F_i(cp)$ accounts (for $1 \leq i \leq \#\mathit{frz}(cp)$). \\ \hline
			D2-LF & A \textbf{successful update CAS} for a rebalancing transformation centered at $viol(cp)$ deposits 4 dollars into each active $L_i$ and $F_i$ account owned by update phases, and 3 dollars into each active $L_i$ and $F_i$ account owned by cleanup phases. \\ \hline
			D1-BCP & The \textbf{start of a cleanup phase} $cp$ deposits $78h(cp) + 312\dot{c}(cp) + 312$ dollars into its own $B(cp)$ account. \\ \hline
			D2-BCP & A \textbf{successful update CAS} of a rebalancing transformation centered at $viol(cp)$ deposits 4934 dollars into each active $B$ account. \\ \hline
			D3-BCP & A \textbf{successful commit step} of a rebalancing transformation centered at $viol(cp)$ deposits 26 dollars into each active $B$ account. \\ \hline
			D3-NIL & A \textbf{successful update CAS} of a rebalancing transformation centered at $viol(cp)$ deposits 1 dollar into each active $B_{\mathit{nil}}$ account. \\ \hline
		\end{tabular}
	\caption{Steps for a cleanup phase $cp$ that deposit dollars into bank accounts.}
\label{fig_deposit_rules_cleanup}
\end{figure}

\begin{figure}[htbp!]
		\begin{tabular}{ | p{3cm} | p{14cm} |}
			\hline
			W-STALE & A \textbf{stale invocation} of \textsc{TryRebalance} for a cleanup phase $cp$ that blames a node $x$ withdraws 1 dollar from $S(cp, x)$. \\ \hline
			W-LLX & The \textbf{failure step of an LLX}, $I$, by $cp$ in attempt $A$ during a non-stale invocation of \textsc{TryRebalance} on $x = \mathit{LLXNode}_i(cp,A)$ withdraws 1 dollar from $L_i(cp)$ if $L_i(cp) > 0$. If $L_i(cp) = 0$ and $x$ is a downwards node for the SCX blamed by $I$, then $cp$ withdraws 1 dollar from $B_{llx}(cp,x)$; otherwise it withdraws 1 dollar from $B_{llx}(cp,p)$, where $p$ is the parent of $x$.\\ \hline
			W-SCX & The \textbf{failure step of an SCX}, $S$, by $cp$ in attempt $A$ during a non-stale invocation of \textsc{TryRebalance} on $x = \mathit{freezeNode}_i(cp,A)$ withdraws 1 dollar from $F_i(cp) $ if $F_i(cp) > 0$. If $F_i(cp) = 0$ and $x$ is a downwards node for the SCX blamed by $S$, then $up$ withdraws 1 dollar from $B_{scx}(cp,x)$; otherwise it withdraws 1 dollar from $B_{scx}(cp,p)$, where $p$ is the parent of $x$. \\ \hline
			W-NIL & A \textbf{failed} \textsc{Nil} \textbf{check} during a non-stale invocation of \textsc{TryRebalance} by a cleanup phase $cp$ withdraws 1 dollar from $B_{\mathit{nil}}(cp)$. \\ \hline
		\end{tabular}
\caption{Steps for a cleanup phase $cp$ that withdraw dollars from its own bank accounts.}
\label{fig_withdraw_rules_cleanup}
\end{figure}

\begin{figure}[htbp!]
		\begin{tabular}{ | p{3cm} | p{14cm} |}
			\hline
			T1-B & Consider a \textbf{successful freezing CAS} performed by any process on a downward node $x$ for some SCX in configuration $C$. For every update or cleanup phase $xp$ where $x \in targets(xp, C)$, transfer 1 dollar from $B(xp)$ to $B_{llx}(xp,x)$, and 1 dollar from $B(xp)$ to $B_{scx}(xp,x)$. \\ \hline
		\end{tabular}
	\caption{A special rule that transfers dollars between accounts.}
\label{fig_transfer_rules}
\end{figure}

From these deposit rules, property P1 of the bank can be verified. We formalize this in the following two lemmas.

\begin{lemma}\label{P1_update}\normalfont
	An update phase $up$ deposits $O(h(up) + \dot{c}(up))$ dollars into the bank.
\end{lemma}

\begin{proof}
D1-LF and D1-BUP are only performed once at the start of each update phase $up$. Additionally, $up$ performs at most 1 successful SCX, and so has at most 1 successful update CAS and commit step. Therefore, the remaining deposit rules are performed at most once each. Since there are at most $\dot{c}(up)$ operations concurrent with $up$, there are at most $\dot{c}(up)$ active operations when each rule is applied. D1-S deposits at most $3\dot{c}(up)$ dollars since $up$ removes at most 3 nodes (in the case of \textsc{Delete}). D1-LF deposits at most 32 dollars since $\#\mathit{frz}(up) \leq \#llx(up) \leq 4$ and so $up$ owns at most 4 $L_i$ and 4 $F_i$ accounts. By D2-LF, $up$ deposits 4 dollars into each of at most 4 $L_i$ and 4 $F_i$ accounts of concurrent update phases, and 3 dollars into each of at most 7 $L_i$ and 7 $F_i$ accounts of concurrent cleanup phases. So D2-LF deposits at most $42\dot{c}(up)$ dollars. D1-BUP, D2-BUP, and D3-BUP deposits at most $30h(up) + 588\dot{c}(up) + 120$, $4934\dot{c}(up)$, and $26\dot{c}(up)$ dollars respectively. Finally, D1-NIL deposits at most $\dot{c}(up)$ dollars. The total number of dollars deposited by $up$ is $30h(up) + 5594\dot{c}(up) + 152 = O(h(up) + \dot{c}(up))$.
\end{proof}

\begin{lemma}\label{P1_cleanup}\normalfont
	A cleanup phase $cp$ in execution $\alpha$ deposits $O(h(cp) + rebal(viol(cp)) \cdot \dot{c}(cp))$ dollars into the bank.
\end{lemma}

\begin{proof}
An update phase enters a cleanup phase $cp$ at most once, so D1-LF and D1-BCP are each performed once. By definition, $viol(cp)$ is located at the center of $rebal(viol(cp))$ rebalancing transformations. Therefore, every other deposit rule is performed $rebal(viol(cp))$ times by $cp$. Since there are at most $\dot{c}(cp)$ operations concurrent with $cp$, there are at most $\dot{c}(cp)$ active operations when each rule is applied.

D1-S deposits at most $5\dot{c}(cp) \cdot rebal(viol(cp))$ dollars since a rebalancing transformation removes at most 5 nodes from the chromatic tree. D1-LF deposits at most 42 dollars since $\#\mathit{frz}(cp) \leq \#llx(cp) \leq 7$, and so $cp$ owns at most 7 $L_i$ and 7 $F_i$ accounts. By D2-LF, $cp$ deposits 4 dollars into each of at most 4 $L_i$ and 4 $F_i$ accounts of concurrent update phases, and 3 dollars into each of at most 7 $L_i$ and 7 $F_i$ accounts of concurrent cleanup phases. So D2-LF deposits at most $42\dot{c}(cp) \cdot rebal(viol(cp))$ dollars. D1-BCP, D2-BCP, and D3-BCP deposits at most $78h(cp) + 312\dot{c}(cp) + 312$, $4934\dot{c}(cp) \cdot rebal(viol(cp))$, and $26\dot{c}(cp) \cdot rebal(viol(cp))$ dollars respectively. Finally, D1-NIL deposits at most $\dot{c}(cp) \cdot rebal(viol(cp))$ dollars. The total amount deposited by $cp$ upon completion is at most $78h(cp) + 312\dot{c}(cp) + 5008\dot{c}(cp)\cdot rebal(viol(cp)) + 354 = O(h(cp) + rebal(viol(cp)) \cdot \dot{c}(cp))$. 
\end{proof}

Finally, we can verify property P2 of the bank by the withdraw rules.
\begin{lemma}\label{banks_P2}\normalfont
For an update phase or cleanup phase $xp$, $xp$ withdraws 1 dollar from one of its own bank accounts for every failed attempt.
\end{lemma}

\begin{proof}
A failed attempt of an update phase occurs after a failed LLX or SCX. The rules W-LLX and W-SCX pay for these respectively.

For a cleanup phase, all failed attempts due to stale invocations of \textsc{TryRebalance} are paid for by W-STALE. All non-stale invocations of \textsc{TryRebalance} fail due to a failed LLX, SCX, or \textsc{Nil} check. The rules W-LLX, W-SCX, and W-NIL pay for these respectively.
\end{proof}

We have shown that Property P1 and P2 of the bank are satisfied. The remainder of the amortized analysis shows that the balance in each bank account is always non-negative. In Section~\ref{section_blaming}, we prove properties of LLX and SCX required throughout the analysis. In Section~\ref{section_target} we define a set of nodes $targets(up)$ for each update phase $up$, which is used by rule T1-B. In Section~\ref{section_target_cleanup}, we define a similar set of nodes $targets(cp)$ for each cleanup phase $cp$, and prove properties about the $S$ accounts. In Section~\ref{section_aux_accounts}, we prove properties of the $L_i$ and $F_i$ accounts. In Sections~\ref{section_llx}, \ref{section_scx}, \ref{section_nil}, and  \ref{section_b_account}, we prove that the $B_{llx}$, $B_{scx}$, $B_{nil}$, and $B$ accounts, respectively, are non-negative.

\subsubsection{Blaming Invocations of SCX}\label{section_blaming}
Recall that, unlike in \cite{DBLP:conf/podc/BrownER13}, we consider an LLX that returns \textsc{Finalized} to be unsuccessful because it causes a failed attempt for an operation. In this section, we prove properties of blaming for an unsuccessful LLX that returns \textsc{Fail} or \textsc{Finalized}. These properties will be useful throughout the amortized analysis.

A step \textit{belongs} to an SCX-record $U$ if the step is performed in any invocation of \textsc{Help}$(ptr)$, where $ptr$ is a pointer to $U$. A step \textit{belongs} to an instance SCX of $S$ if the step belongs to the SCX-record created by $S$.  The following properties of the implementation of LLX and SCX used in the chromatic tree were proven in \cite{DBLP:conf/podc/BrownER13}, and are used throughout this section.

\begin{observation}\normalfont\label{llx_obs}
	Let $U$ be an SCX-record created by an instance $S$ of SCX$(V, R, fld, new)$. Let $r$ be a Data-record. Let $r$ be a Data-record.
	\begin{enumerate}
		\item \label{llx_obs:state} The $state$ of $U$ is initially \text{InProgress} when $U$ is first created. The $state$ of $U$ only changes from \text{InProgress} to \text{Aborted} after an abort step for $U$, or from \text{InProgress} to \text{Committed} after a commit step for $U$. 
		
		
		\item \label{llx_obs:aba} Every update to $\mathit{info}$ field of $r$ changes $\mathit{r.info}$ to a value that has never previously appeared there.
		
		\item \label{llx_obs:marked} Suppose a successful mark step $mstep$ belonging to an SCX-record $U$ on $r$ occurs. Then $r$ is frozen for $U$ when $mstep$ occurs, and forever after. 
		
		\item \label{llx_obs:f_step} There cannot be both a frozen step and an abort step belonging to the same SCX-record. 
		
		\item \label{llx_obs:help} Suppose a frozen step is performed for $U$. Then a commit step is performed for $U$ before any invocation of \textsc{Help}$(ptr)$ terminates, where $ptr$ is a pointer to $U$.
		
		
		\item \label{llx_obs:frozen_step} If a frozen step belongs to $U$ then, for each $x \in U.V$ , there is a successful freezing CAS belonging to $U$ on $x$ that occurs before the first frozen step
		belonging to $U$.
		
		
		
		\item \label{llx_obs:threat} Let the threatening section of $S$ begin with the first freezing CAS for $U$, and end with the first abort or commit step for $U$. Every successful freezing CAS for $U$ occurs during $S$'s threatening section.
		
		\item \label{llx_obs:fail_blame} Let $I$ be an invocation of LLX$(r)$ that returns \textsc{Fail}, and suppose $I$ blames $S$. Then a commit step or abort step is performed for $S$ before $I$ returns. 
		
		
		\item \label{llx_obs:scx_blame_limit} An invocation $S$ of $SCX(V, R, fld, new)$ cannot be blamed for any $r \in V$ by more than one invocation of SCX per process. 
	\end{enumerate}
\end{observation}

We extend Observation~\ref{llx_obs}.\ref{llx_obs:fail_blame} for invocations of LLX that return \textsc{Finalized} instead of \textsc{Fail}. 
\begin{lemma}\label{llx_commit}\normalfont
	Let $I$ be an invocation of LLX$(r)$ that returns \textsc{Finalized}, and $S$ be the invocation of SCX that is blamed by $I$. Then a commit step belonging to $S$ is performed before $I$ returns.
\end{lemma}

\begin{proof}
Since $I$ returns \textsc{Finalized}, the if-statement on line~\ref{ln:llx:final_if} evaluates to \textsc{True}. So $marked_1 = \textsc{True}$, and either $\mathit{rinfo.state} = \text{Committed}$, or $\mathit{rinfo.state} = \text{InProgress}$ and the invocation of \textsc{Help}$(\mathit{rinfo})$ returns \textsc{True}. Note that since $marked_1 = \textsc{True}$ and marked nodes do not become unmarked, this implies that $marked_2 = \textsc{True}$. Therefore, $I$ does not enter the if-block on line~\ref{ln:llx:main_if}. Since $I$ blames $S$, $\mathit{rinfo}$ points to the SCX-record created by $S$. By Observation~\ref{llx_obs}.\ref{llx_obs:state}, if $\mathit{rinfo.state} = \text{Committed}$, then a commit step has already been performed for $S$. If $\mathit{rinfo.state} = \text{InProgress}$ and the invocation of \textsc{Help}$(\mathit{rinfo})$ on line~\ref{ln:llx:final_if} returns \textsc{True}, then by Observation~\ref{llx_obs}.\ref{llx_obs:help}, a commit step has been performed for $S$. This is before $I$ returns \textsc{Finalized}. 
\end{proof}

It was shown that each invocation of SCX can be blamed by at most two invocations of LLX that return \textsc{Fail} per process \cite{DBLP:conf/podc/BrownER13}. This result does not hold for LLXs that return \textsc{Fail} or \textsc{Finalized}. The following lemmas show that, under certain conditions, two failed invocations of LLX$(r)$ (that either return \textsc{Fail} or \textsc{Finalized}) performed by the same process will blame different invocations of SCX.

\begin{lemma}\label{llx_abort_blame_once}\normalfont
	Let $S$ be an instance of SCX whose SCX-record $U$ is in the \text{Aborted} state. If $I$ is an instance of LLX$(r)$ that is invoked after the first abort step belonging to $S$, then $I$ does not blame $S$.
\end{lemma}

\begin{proof}
Suppose, for contradiction, that $I$ blames $S$. We consider two cases, depending on what $\mathit{rinfo}$ points to when $I$ executes line~\ref{ln:llx:rinfo} of LLX. First, suppose $\mathit{rinfo}$ points to some SCX-record $U' \neq U$. Since $I$ blames $S$, $r.\mathit{info}$ points to $U$ when $I$ performs line~\ref{ln:llx:rinfo_change}. Note that $\mathit{r.info}$ can only be updated by a successful freezing CAS on line~\ref{ln:help:freezing_cas} of \textsc{Help}. Additionally, $U$ is in the Aborted state before the invocation of $I$. So by Observation~\ref{llx_obs}.\ref{llx_obs:threat}, no successful freezing CAS for $U$ occurs during $I$. This implies that $\mathit{r.info}$ does not point to $U$ at any point during $I$. This is a contradiction.

So suppose $\mathit{rinfo}$ points to $U$ when $I$ executes line~\ref{ln:llx:rinfo} of LLX. By Observation~\ref{llx_obs}.\ref{llx_obs:state}, the state of $U$ remains \text{Aborted} after it is set to \text{Aborted}, and so $state = \text{Aborted}$ when $I$ executes line~\ref{ln:llx:state}. Therefore, the if-statement on line~\ref{ln:llx:main_if} evaluates to \textsc{True} for $I$. Since $I$ fails, the check on line~\ref{ln:llx:rinfo_change} must have failed, and so $\mathit{r.info}$ was updated to point to an SCX-record $U' \neq U$ sometime after $I$ executed line~\ref{ln:llx:rinfo}. By definition, $I$ blames the SCX that created $U'$. Since each instance of SCX creates only one SCX-record, $I$ does not blame $S$.
\end{proof}

\begin{lemma}\label{llx_commit_blame_once}\normalfont
	Let $S$ be an instance of SCX whose SCX-record $U$ is in the \text{Commited} state. Let $I$ be an instance of LLX$(r)$ invoked after the first commit step belonging to $S$. If $r$ is a node in the chromatic tree for all configurations during $I$, then $I$ does not blame $S$.
\end{lemma}

\begin{proof}
Note that since $U$ is committed before the start of $I$, by Observation~\ref{llx_obs}.\ref{llx_obs:threat}, there is no freezing CAS on $r$ belonging to $S$ during $I$. This implies that if $\mathit{r.info}$ does not already point to $U$ when $I$ executes line~\ref{ln:llx:rinfo}, then $\mathit{r.info}$ will not be updated to point to $U$ at any later point during $I$. Hence, $I$ does not blame $S$. So suppose when line~\ref{ln:llx:rinfo} is executed for $I$, $\mathit{r.info}$ points to $U$. Therefore, $state = \text{Committed}$ when line~\ref{ln:llx:state} is executed for $I$. 

Note that if $marked_1$ is \textsc{True} when $I$ executes line~\ref{ln:llx:mark1}, then $r$ is both marked and $\mathit{r.info}$ points to a committed SCX-record when $I$ executes line~\ref{ln:llx:rinfo}. Hence, by definition, $r$ is not a node in the chromatic tree. Therefore, $marked_1$ is \textsc{False} when $I$ executes line~\ref{ln:llx:mark1}.

We argue that $marked_2 = \textsc{False}$ when $I$ executes line~\ref{ln:llx:mark2}. Suppose, for contradiction, that $marked_2 = \textsc{True}$. This implies $r$ is marked by a mark step $mstep$ sometime after $I$ executes line~\ref{ln:llx:mark1}, but before $I$ executes line~\ref{ln:llx:mark2}. Note that $mstep$ cannot belong to $U$ since a commit step for $U$ has already been performed. Hence, by Observation~\ref{llx_obs}.\ref{llx_obs:marked}, when $mstep$ occurs, $r$ is frozen for some SCX-record $U' \neq U$, and $\mathit{r.info}$ points to $U'$ in all configurations after $mstep$. We argue that there exists a configuration in which $U'.state = \text{Committed}$ after $mstep$ but before the end of $I$. This implies that $r$ is not a node in the chromatic tree in this configuration because $r$ is marked and $\mathit{r.info}$ points to a committed SCX-record. Note that $I$ does not enter the if-block on line~\ref{ln:llx:main_if} since $marked_2 = \textsc{True}$. Additionally, $I$ does not enter the if-block on line~\ref{ln:llx:final_if} since $marked_1 = \textsc{False}$. Since $\mathit{r.info}$ points to $U'$ in all configurations after $mstep$, $\mathit{r.info}$ points to $U'$ immediately before $I$ executes line~\ref{ln:llx:help}. From the code, a mark step for $U'$ only occurs after a successful frozen step for $U'$. Therefore, by Observation~\ref{llx_obs}.\ref{llx_obs:f_step}, $U'$ is not in the Aborted state when $I$ reads the state of $\mathit{r.info}$ on line~\ref{ln:llx:help}. If $U'.state = \text{Committed}$ at this point, then $r$ is not a node in the chromatic tree. If $U'.state = \text{InProgress}$ at this point, then $I$ invokes \textsc{Help}$(\mathit{r.info})$, where $\mathit{r.info}$ points to $U'$. By Observation~\ref{llx_obs}.\ref{llx_obs:help}, a commit step for $U'$ is performed before the termination of $H$. Immediately after this commit step, $r$ is marked and $\mathit{r.info}$ points to a committed SCX-record, so $r$ is not in the chromatic tree. 

Therefore, consider the case when $marked_2 = \textsc{False}$. Recall that $state = \text{Committed}$ when line~\ref{ln:llx:state} is executed for $I$, and so $I$ enters the if-block on line~\ref{ln:llx:main_if}. Since $I$ is a failed invocation of LLX, it fails the check on line~\ref{ln:llx:rinfo_change}. Hence, $\mathit{r.info}$ points to an SCX-record $U'' \neq U$ when line~\ref{ln:llx:rinfo_change} is executed. Since an invocation of SCX only creates one SCX-record, $U'$ is not created by $S$. Thus, $I$ does not blame $S$. 
\end{proof}

\begin{lemma}\label{llx_blame_diff_node}\normalfont
	Consider two failed invocations of $I'$ of LLX$(r')$ and $I$ of LLX$(r)$ by an operation, where $I'$ precedes $I$. Suppose $r$ is in the chromatic tree in all configurations between the start of $I'$ and the end of $I$. Then $I'$ and $I$ blame different invocations of SCX. 
\end{lemma}

\begin{proof}
Suppose that $I'$ blames an invocation $S$ of SCX, and let $U$ be the SCX-record created by $S$. By Lemma~\ref{llx_commit} and Observation~\ref{llx_obs}.\ref{llx_obs:fail_blame}, either a commit step or abort step belonging to $S$ is performed before the completion of $I'$. If an abort step is performed, then by Lemma~\ref{llx_abort_blame_once}, $I$ does not blame $S$ since $I$ is invoked after the abort step. If a commit step is performed, then by Lemma~\ref{llx_commit_blame_once}, $I$ does not blame $S$ since $I$ is invoked after the commit step.
\end{proof}

\begin{lemma}\label{llx_no_ucas_abort}\normalfont
Let $I$ be a failed invocation of LLX$(r)$ that blames an SCX $S$. If no successful update CAS (belonging to any SCX) is performed between the last successful freezing CAS $fcas$ belonging to $S$ and the end of $I$, then an abort step belonging to $I$ is performed before the end of $I$.
\end{lemma}

\begin{proof}
By Observation~\ref{llx_obs}.\ref{llx_obs:fail_blame} and Lemma~\ref{llx_commit}, either a commit step or abort step belonging to $S'$ occurs before the end of $I$. Suppose, for contradiction, that a commit step belonging to $S$ is performed. Observation~\ref{llx_obs}.\ref{llx_obs:frozen_step} implies a frozen step belonging to $S$ does not occur before $fcas$. Thus, a successful freezing CAS does not occur before $fcas$. From the code, a successful update CAS belonging to $S$ is performed before its commit step. Therefore, a successful update CAS is performed sometime between $\mathit{fcas}$ and the end of $I$, a contradiction.
\end{proof}

\begin{lemma}\label{llx_blame_3}\normalfont
	Consider three failed invocations of $I''$, $I'$, and $I$ of LLX$(r)$ by an operation. Suppose that between the start of $I''$ and the failure step of $I$, $r$ is in the chromatic tree and no successful update CAS occurs. Then $I$ blames an SCX different from $I'$.
\end{lemma}

\begin{proof}
Let $S''$, $S'$, and $S$ be the invocations of SCX blamed by $I''$, $I'$, and $I$, respectively. Let $U'$ be the SCX-record created by $I'$. By Lemma~\ref{llx_blame_diff_node}, $I''$ and $I'$ blame different SCX, and so $S'' \neq S'$. Since $I''$ and $I'$ blame different SCX, a successful freezing CAS $\mathit{fcas}'$ changes $r.\mathit{info}$ to point to $U'$ sometime after the failure step of $I''$. Thus, no successful update CAS occurs between $\mathit{fcas}'$ and the end of $I'$. By Lemma~\ref{llx_no_ucas_abort}, an abort step belonging to $U'$ is performed before the end of $I'$. Since $I$ is invoked after the abort step belonging to $S'$, by Lemma~\ref{llx_abort_blame_once}, $I$ does not blame $S'$. Therefore, $S \neq S'$.
\end{proof}

In the case that two consecutive invocations $I'$ and $I$ of LLX$(r)$ blame different invocations of SCX, we can prove that the freezing CASs for the SCX blamed by $I$ occur sometime between the start of $I'$ and the failure step of $I$.

\begin{lemma}\label{llx_2_attempts}\normalfont	
	Consider two failed invocations of $I'$ and $I$ of LLX$(r)$ by an operation, where $I'$ precedes $I$. Let $S'$ and $S$ be the invocations of SCX blamed by $I'$ and $I$, respectively, where $S \neq S'$. Suppose $r$ is in the chromatic tree in all configurations between the start of $I'$ and the failure step of $I$. Then there is no successful freezing CAS belonging to $S$ occurs before the start of $I'$.
\end{lemma}

\begin{proof}
Let $U'$ and $U$ be the SCX-records created by $S'$ and $S$, respectively. Since $I'$ blames $S'$, the failure step of $I'$ reads that $r.\mathit{info}$ points to $U'$. Thus, there exists a successful freezing CAS $\mathit{fcas}'$ for $S'$ that points $r.\mathit{info}$ to $U'$. Likewise, there is a successful freezing CAS $\mathit{fcas}$ for $S$ that points $r.\mathit{info}$ to $U$. Suppose $\mathit{fcas'}$ occurs after $\mathit{fcas}$. Then $r.\mathit{info}$ points to $U$ sometime before the failure step of $I'$. Observation~\ref{llx_obs}.\ref{llx_obs:aba} implies that $r.\mathit{info}$ will not point back to $U$, and so $I$ does not blame $S$. Therefore, $\mathit{fcas'}$ occurs before $\mathit{fcas}$. 

It remains to show that no successful freezing CAS for $S$ occurs before the start of $I'$. For contradiction, suppose there is a successful freezing CAS for $S$ that occurs before the start of $I'$. Then all LLXs linked to $S$ have succeeded before the start of $I'$. Let $J$ be the successful invocation of LLX$(r)$ linked to $S$. First, suppose $J$ reads that $\mathit{rinfo}$ points to an SCX-record $U'' \neq U'$ on line~\ref{ln:llx:rinfo}. Note that $\mathit{fcas}'$ occurs sometime after this point since $I'$ reads $\mathit{rinfo}$ points to $U'$ sometime during its execution. Then $\mathit{r.info}$ no longer points to $U''$ when $\mathit{fcas}$ occurs. This implies that $\mathit{fcas}$ fails, which contradicts the fact that $\mathit{fcas}$ is a successful freezing CAS.

So suppose $J$ reads that $\mathit{rinfo}$ points to $U'$ on line~\ref{ln:llx:rinfo}. Since $J$ is successful, when $J$ executes line~\ref{ln:llx:main_if}, either $U'.state = \text{Aborted}$, or $U'.state = \text{Committed}$ and $marked_2 = \textsc{False}$. If $U'.state = \text{Aborted}$, then by Lemma~\ref{llx_abort_blame_once}, no LLX invoked after $J$ can blame $S'$. In particular, this contradicts the fact that $I'$ blames $S'$. If $U'.state = \text{Committed}$, then by Lemma~\ref{llx_commit_blame_once}, $I'$ does not blame $S'$. Therefore, no successful freezing CAS for $S$ occurs before the beginning of $I'$.
\end{proof}

A similar lemma can be proven for two invocations of SCX that fail a freezing iteration on the same node $r$, but blame different invocations of SCX.

\begin{lemma}\label{scx_2_attempts}\normalfont
	Consider two failed invocations of $I'$ and $I$ of SCX$(V,R,fld,new)$ by an operation, where $I'$ precedes $I$. Suppose both $I'$ and $I$ fail a freezing iteration on the same node $r$, and let $S$ be the invocation of SCX blamed by $I$. Let $L'$ be the first invocation of LLX linked to $I'$. Then no successful freezing CAS belonging to $S$ occurs before the beginning of $L'$.
\end{lemma}

\begin{proof}
By Observation~\ref{llx_obs}.\ref{llx_obs:scx_blame_limit}, $S$ can only be blamed for $r$ by at most one SCX per process. So $I'$ blames some invocation $S' \neq S$ of SCX for $r$.


Suppose, for contradiction, there is a successful freezing CAS belonging to $S$ that occurs before the start of $L'$. Then all LLXs linked to $S$ must have succeeded before the start of $L'$. Let $J$ be the LLX($r$) linked to $S$. Let $\mathit{fcas}'$ be the successful freezing CAS on $r$ belonging to $S'$. Since the LLX$(r)$ linked to $I'$ succeeded (otherwise $I'$ does not perform a freezing iteration on $r$), $\mathit{fcas}'$ occurs after the LLX$(r)$ linked to $I'$, but before the failed freezing iteration on $r$ belonging $I'$. This implies no freezing CASs on $r$ belonging to $S$ will succeed, since $r.\mathit{info}$ no longer points to the same SCX-record as seen by $J$. Since $r.\mathit{info}$ never points to the SCX-record created by $S$, $I$ does not blame $S$, which is a contradiction.
\end{proof}

\subsubsection{Target Nodes of an Update Phase}\label{section_target}
In this section, we define a set of nodes $targets(up,C)$ for every update phase $up$ and configuration $C$ in an execution. This set of nodes is used when proving properties of the bank accounts. We prove that the nodes in $targets(up, C)$, except for the node $targetRoot(up,C)$, are the nodes on which $up$ performs LLXs in a successful update attempt in a solo execution from $C$. To formally define $targets(up)$, we introduce the following definitions.

\begin{definition}\normalfont
	Let the \textit{extended tree}\index{extended tree} $T^*$\index{$T^*$} of a chromatic tree $T$ be a tree containing all nodes in $T$, in addition to new nodes so that:
	\begin{itemize}
		\item the subtree rooted at each leaf of $T$ is a complete tree with depth 2,
		\item the $entry$ node of $T$ has 4 ancestors,
		\item the parent of $entry$ has a new right child that is the root of a complete tree of depth 2, and
		\item the grandparent of $entry$ has a new right child which is a leaf.
	\end{itemize}
	The nodes in $T^*$ that are not in $T$ are called $phantom$ nodes. Let the $phantomRoot$ be the topmost ancestor of $entry$ in $T^*$. The extended tree is illustrated in Figure~\ref{extended_tree}.
\end{definition}

\begin{figure}[!h]
	\centering
	\begin{subfigure}[b]{0.45\textwidth}\centering
		\begin{tikzpicture}[-, >=stealth', 
		level distance=0.9cm,
		level 1/.style={sibling distance=1.4cm, level distance = 0.6cm},
		level 2/.style={sibling distance=0.7cm, level distance = 0.6cm},
		level 3/.style={sibling distance=0.7cm, level distance = 0.6cm}
		] 
		\node [arn_gx, label=right:{$l$}]{ } 
		child{ node [arn_w, dashed] { } 
			child{ node [arn_w, dashed] { } }
			child{ node [arn_w, dashed] { } }
		}
		child{ node [arn_w, dashed] { } 
			child{ node [arn_w, dashed] { } }
			child{ node [arn_w, dashed] { } }
		};
		\end{tikzpicture}
		\caption{leaf $l$ of $T$ extended by $phantom$ nodes}
	\end{subfigure}
	\qquad
	\begin{subfigure}[b]{0.45\textwidth}\centering
		\begin{tikzpicture}[-, >=stealth', 
		level distance=0.7cm,
		level 3/.style={sibling distance=1.4cm, level distance = 0.6cm},
		level 6/.style={sibling distance=0.7cm, level distance = 0.6cm}
		] 
		\node [arn_w, dashed, label=right:{$phantomRoot$}]{ } 
		child{node [arn_w, dashed, label=right:{}] { }
			child{node [arn_w, dashed, label=right:{}] { }
				child{node [arn_w, dashed, label=right:{}] { }   
					child{ node [arn_g, label=left:{$entry$}] { } }
					child{ node [arn_w, dashed, label=right:{}] { } 
						child{ node [arn_w, dashed] { } 
							child{ node [arn_w, dashed] { } }
							child{ node [arn_w, dashed] { } }
						}
						child{ node [arn_w, dashed] { } 
							child{ node [arn_w, dashed] { } }
							child{ node [arn_w, dashed] { } }
						}					
					}
				}
				child{node [arn_w, dashed, label=right:{}] { }}   
			}
		}; 
		\end{tikzpicture}
		\caption{$entry$ extended by $phantom$ nodes}
	\end{subfigure}
	\caption{Each leaf $l$ of $T$ and $entry$ extended by $phantom$ nodes (denoted with dashes) in its extended tree $T^*$.}
	\label{extended_tree} 
\end{figure}
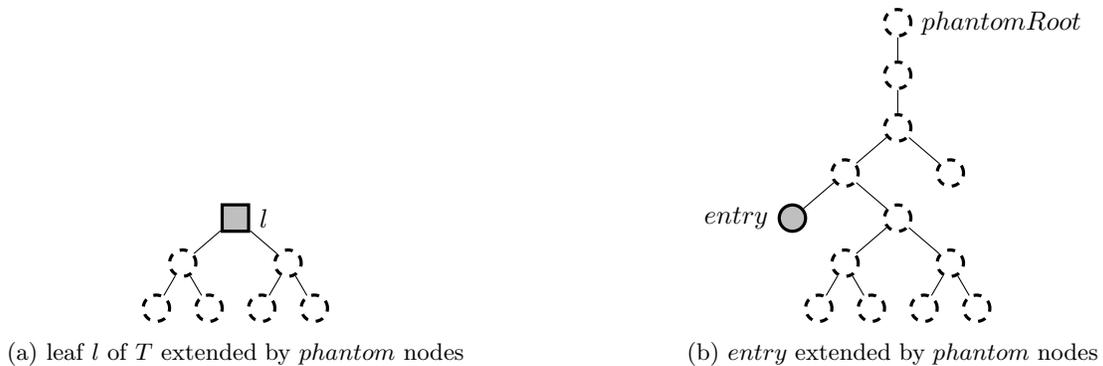

For an update phase or cleanup phase $xp$ of an \textsc{Insert}(k) or \textsc{Delete}(k) operation, let $xp$'s search path be the search path for the key $k$. For an update phase $up$, let $\mathit{focalNode}(up,C)$\index{$\mathit{focalNode}(up,C)$} be the unique leaf on $up$'s search path in $C$, which is the node labeled $l$ in Figure~\ref{op_targets}. For each update phase $up$ and configuration $C$, we define a set of nodes in the chromatic tree denoted by $targets(up, C)$. \index{$targets(up, C)$}

\begin{definition}\label{def_targets_up_insert}\normalfont
For the update phase $up$ of an \textsc{Insert}$(k)$ operation, let $targets(up, C)$ be the set of nodes in the extended tree $T^*$ including $\mathit{focalNode}(up,C)$, along with its parent $p$ and grandparent $ggp$. (See Figure~\ref{op_targets}a).
\end{definition}

\begin{definition}\label{def_targets_up_delete}\normalfont
For the update phase $up$ of a \textsc{Delete}$(k)$ operation, let $targets(up, C)$ be the set of nodes in the extended tree $T^*$ including $\mathit{focalNode}(up,C)$, along with its parent $p$, grandparent $gp$, great-grandparent $ggp$, and sibling $s$. (See Figure~\ref{op_targets}b).
\end{definition}

\begin{figure}[H]
\centering
\begin{subfigure}[b]{0.45\textwidth}\centering
\begin{tikzpicture}[-,>=stealth',level/.style={sibling distance = 4cm/#1, level distance = 0.8cm}] 
\node [arn_w, label=right:{$gp = targetRoot(up,C)$}] {}
child{ node [arn_w, label=right:{$p$}] { } 
	child{ node [arn_gx, label=right:{$l$}] { } }
}; 
\end{tikzpicture}
\caption{$targets(up)$ for \textsc{Insert} operations}
\end{subfigure}%
\begin{subfigure}[b]{0.45\textwidth}\centering
\begin{tikzpicture}[-,>=stealth',
level distance=0.9cm,
level 1/.style={sibling distance=1cm, level distance = 0.6cm},
] 
\node [arn_w, label=right:{$ggp = targetRoot(up,C)$}]{ } 
child{node [arn_w, label=right:{$gp$}] { }  
	child{node [arn_w, label=right:{$p$}] { }  
		child{ node [arn_gx, label=left:{$l$}] { } }
		child{ node [arn_w, label=right:{$s$}] { } }
	}
}; 
\end{tikzpicture}
\caption{$targets(up,C)$ for \textsc{Delete} operations}
\end{subfigure}
\caption{The set of nodes in $targets(up,C)$. The node highlighted gray is $\mathit{focalNode}(up,C)$.}\label{op_targets} 
\end{figure}
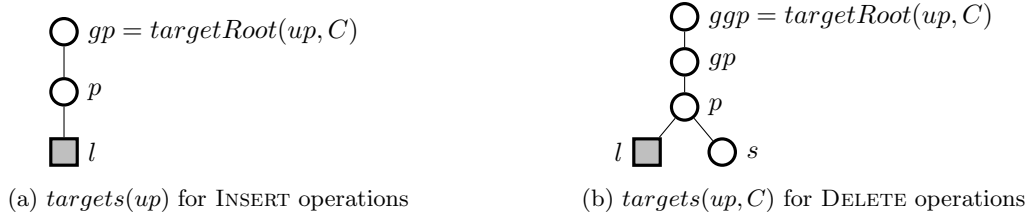

\noindent The phantom nodes that are ancestors of $entry$ are required so that the grandparent of $focalNode(up,C)$ is always defined, even when the chromatic tree is empty in $C$. The other phantom nodes are needed when $targets(cp)$ is defined in Section~\ref{section_target_cleanup}. For any execution $\alpha$ and function $f(x,C)$, we let $f(x)$ denote $f(x,C)$ as $C$ varies over the configurations of $\alpha$. Let $targetRoot(up,C)$ be the root of the tree formed by the nodes in $target(up,C)$. 

The remainder of this section is dedicated to proving properties about $target(up,C)$. The following lemmas relate to how the search routines of operations behave when no successful update CASs are performed by any process throughout the duration of an attempt. 

\begin{lemma}\label{backtrack_in_tree}\normalfont
Consider an attempt $A$ of an update or cleanup phase $xp$ where no successful update CAS (by any operation) occurs during $A$. If a node $m$ is in the chromatic tree when it is popped off $xp$'s stack during backtracking, then $m$ is unmarked when $xp$ next executes line~\ref{ln:search:backtracking_start} of \textsc{BacktrackingSearch} (or line~\ref{ln:cleanup:backtracking_start} of \textsc{BacktrackingCleanup}).
\end{lemma}

\begin{proof}
We prove the case for when $xp$ is an update phase. The case for when $xp$ is a cleanup phase follows similarly. 

Suppose, for contradiction, that $m$ is marked when $xp$ executes line~\ref{ln:search:backtracking_start} of \textsc{BacktrackingSearch}. By Lemma~\ref{backtracking_marked}, there exists a configuration $C$ before the next pop by $xp$ where $m$ is not a node in the chromatic tree. Only update CASs can remove a node from the chromatic tree. Thus, an update CAS has removed $m$ from the chromatic tree since the point when $m$ was popped by $xp$. This contradicts the assumption that no update CAS occurs during $A$.
\end{proof}

\begin{lemma}\label{last_pop_in_tree}\normalfont
Consider an attempt $A$ of an update or cleanup phase $xp$ for an \textsc{Insert}$(k)$ or \textsc{Delete}$(k)$ operation where no successful update CAS (by any operation) occurs during $A$. Then the last node popped by $xp$ during backtracking in $A$ (or $entry$ if $A$ is the first attempt of $xp$) is in the chromatic tree and on the search path for $k$ for all configurations during $A$.
\end{lemma}

\begin{proof}
We prove the case for when $xp$ is an update phase. The case for when $xp$ is a cleanup phase follows similarly.

Let $m$ be the last node popped by $xp$. If $m = entry$, then $m$ is always in the chromatic tree and is on the search path for every key in the tree. So suppose $m \neq entry$. From the code on lines~\ref{ln:search:backtracking_start}-\ref{ln:search:pop} of \textsc{BacktrackingSearch}, if $m$ is marked when $xp$ executes line~\ref{ln:search:backtracking_start} with $l = m$, then $xp$ pops another node on line~\ref{ln:search:pop}. This contradicts the fact that $m$ is the last node popped by $xp$ during backtracking in $A$. So $m$ is unmarked when $xp$ executes line~\ref{ln:search:backtracking_start}. Since marked nodes do not become unmarked, $m$ was unmarked in the configuration $C$ immediately before it is popped. By Lemma~\ref{on_stack_SP}, $m$ is on the search path for $k$ in $C$, and therefore also in the chromatic tree in $C$. Since no update CASs occur during $A$, the structure of the chromatic tree does not change during $A$. It follows that $m$ is in the chromatic tree and on the search path for $k$ for all configurations during $A$.
\end{proof}

\begin{lemma}\label{targets_same_order}\normalfont
Let $A'$ be an unsuccessful attempt of an update phase $up$, and let $A$ be an attempt of $up$ following $A'$. Suppose no successful update CAS occurs between the start of $A'$ and the end of $A$. Let $gp'$, $p'$, and $l'$ be the last three nodes on the $up$'s search path at the start of $A'$. If \textsc{TryInsert}$(p',l')$ or \textsc{TryDelete}$(gp',p',l')$ is invoked in $A'$, then $up$ invokes \textsc{TryInsert}$(p',l')$ or \textsc{TryDelete}$(gp',p',l')$ in $A$.
\end{lemma}

\begin{proof}
Note that since no update CAS occurs throughout $A'$ and $A$, $gp'$, $p'$, and $l'$ are the last three nodes on $up$'s search path for all configurations between the start of $A'$ and the end of $A$. Suppose $up$ invokes \textsc{TryDelete}$(gp',p',l')$ in $A'$. The top node on $up$'s stack at the end of $A'$ is $gp'$. Since $gp'$ is a node in the chromatic tree, by Lemma~\ref{backtrack_in_tree}, $up$ stops backtracking at $gp'$. Since $gp'$, $p'$, and $l'$ are the last three nodes on the search path for $k$, it follows that $up$ will invoke \textsc{TryDelete}$(gp', p', l')$ in $A$. The case for when $up$ invokes \textsc{TryInsert}$(p',l')$ follows similarly.
\end{proof}

\begin{lemma}\label{targets_fixed_invocations}\normalfont
	Let $A'$ be an unsuccessful attempt of an update phase $up$, and let $A$ be an attempt of $up$ following $A'$. Let $gp'$, $p'$, and $l'$ be the last three nodes on $up$'s search path at the start of $A'$. If no successful update CAS occurs between the start of $A'$ and the end of $A$, then $up$ invokes \textsc{TryInsert}$(p',l')$ or \textsc{TryDelete}$(gp',p',l')$ in $A$.
\end{lemma}

\begin{proof}
Suppose $up$ is the update phase of an \textsc{Insert}$(k)$ operation. In the instance of \textsc{BacktrackingSearch} by $up$ in $A$, $up$ begins by backtracking to an unmarked node $m$. By Lemma~\ref{last_pop_in_tree}, $m$ is in the chromatic tree and on the search path for $k$ for all configurations during $A$. By the check on line~\ref{ln:search:forward_start}, only internal nodes are pushed onto the stack, so $m$ is not a leaf and is a proper ancestor of $l'$ throughout $A$. Since no update CASs occur throughout $A$, any nodes that $up$ visits when looking for a leaf during $A$ by following pointers from $m$ are also in the chromatic tree. Therefore, by Lemma~\ref{v_reachable}, each node that $up$ visits is on the search path for $k$. Since $p'$ and $l'$ are on the last two nodes on the search path for $k$ starting from $m$ for any configuration during $A$, $up$'s invocation of \textsc{BacktrackingSearch} returns the nodes $p'$ and $l'$. It follows that $up$ will invoke \textsc{TryInsert}$(p',l')$ in $A$.

So suppose $up$ is the update phase of a \textsc{Delete}$(k)$ operation. At the start of $up$ in $A'$, $up$ begins by backtracking to an unmarked node. Let $m'$ be the last node popped by $up$ during backtracking (or $m' = entry$ if $A'$ is the first attempt of $up$). By Lemma~\ref{last_pop_in_tree}, $m'$ is in the chromatic tree and on the search path for $k$ for all configurations in $A'$. By the check on line~\ref{ln:search:forward_start}, only internal nodes are pushed onto the stack, and so $m'$ is not a leaf. It follows that $m'$ is a proper ancestor of $l$. We consider the following two cases, depending on the node $m'$.

Suppose $m' \neq p'$. So $m'$ is a proper ancestor of $p'$. Then by Lemma~\ref{v_reachable}, the nodes $gp'$, $p'$, and $l'$ are on the search path for $up$ starting from $m'$, and so $up$ invokes \textsc{TryDelete}$(gp',p',l')$ in $A'$. By Lemma~\ref{targets_same_order}, $up$ invokes \textsc{TryDelete}$(gp',p',l')$ in $A$.

Suppose $m' = p'$. Then $up$ only performs a single pointer traversal when looking for a leaf before visiting $l$. After visiting $l'$, $up$ performs \textsc{Pop} (on line~\ref{ln:search:pop_p} of \textsc{BacktrackingSearch}) on its stack that returns $p'$, and \textsc{Top} (on line~\ref{ln:search:pop_gp} of \textsc{BacktrackingSearch}), returning a node $gp''$. Then $up$ invokes \textsc{TryDelete}$(gp'',p',l')$ in $A'$.

If $gp'' = gp'$, then by Lemma~\ref{targets_same_order}, $up$ invokes \textsc{TryDelete}$(gp',p',l')$ in $A$. So suppose $gp'' \neq gp'$, and $up$ invokes \textsc{TryDelete}$(gp'',p',l')$ in $A'$. This invocation of \textsc{TryDelete} is unsuccessful, otherwise a successful update CAS will occur during its successful SCX. In the attempt $A$, $up$ backtracks starting from $gp''$ to a node $m$. Note that $m \neq p'$ because $p'$ was popped from the stack at the end of $A'$. By Lemma~\ref{last_pop_in_tree}, $m$ is in the chromatic tree and on the search path for $k$ for all configurations during $A$. Then $m$ is a proper ancestor of $p'$. Therefore, by Lemma~\ref{v_reachable}, the nodes $gp'$, $p'$, and $l'$ are on the search path for $up$ starting from $m$. Therefore, $up$ invokes \textsc{TryDelete}$(gp',p',l')$ in $A$.
\end{proof}

In the following lemma, we show that if no update CASs occur, then $up$ will eventually perform LLXs on nodes in $targets(up, C) - \{ targetRoot(up,C)\}$.

\begin{lemma}\label{targets_fixed}\normalfont
Let $A'$ be an unsuccessful attempt of an update phase $up$, and let $A$ be an attempt of $up$ following $A'$. If no successful update CAS occurs between the start of $A'$ and the end of $A$, then for all $1 \leq i \leq \#llx(A)$ and for any configuration $C$ between the beginning and end of $A$, $LLXNode_i(up, A) \in targets(up, C) - \{ targetRoot(up,C)\}$.
\end{lemma}

\begin{proof}
No update CAS occurs between the start of $A'$ and the end of $A$, so the structure of the chromatic tree does not change between start of $A'$ and the end of $A$. Thus $targets(up)$ does not change between the start of a $A'$ and the end of $A$. It is therefore sufficient to show that for any configuration $C$ during $A$, $up$ only performs LLXs on nodes in $targets(up, C)$ during $A$.

Suppose $up$ is for a \textsc{Delete}$(k)$ operation. By Lemma~\ref{targets_fixed_invocations}, $up$ invokes \textsc{TryDelete}$(gp,p,l)$ in $A$, where $gp$, $p$, and $l$ are the last three nodes on the search path for $k$ in $C$. From the code of \textsc{TryDelete}, if no update CAS is performed during this instance, $up$ will only perform LLXs on $gp$, $p$, $l$, and the sibling of $l$. By definition, these nodes are in $targets(up, C)$.

Suppose $up$ is for a \textsc{Insert}$(k)$ operation. By Lemma~\ref{targets_fixed_invocations}, $up$ invokes \textsc{TryInsert}$(p,l)$ in $A$, where $p$, and $l$ are the last two nodes on the search path for $k$ in $C$. From the code of \textsc{TryInsert}, $up$ will only perform LLXs on $p$ and $l$. By definition, these nodes are in $targets(up, C) - \{ targetRoot(up,C)\}$.
\end{proof}

\subsubsection{Target Nodes of a Cleanup Phase}\label{section_target_cleanup}
In this section, for each cleanup phase $cp$, we define a set of nodes $targets(cp)$ with properties similar to those of $targets(up)$. The goal is to prove a lemma similar to Lemma~\ref{targets_fixed}, except for $targets(cp)$. The definition of $targets(cp)$ is more complicated than $targets(up)$. This is because the violation $cp$ encounters may be different in different cleanup attempts, even in the case where the structure of the chromatic tree does not change during the attempts. Additionally, concurrent operations may move violations onto $cp$'s search path, causing it to perform cleanup attempts it would otherwise not perform.

To define $targets(cp)$, we require a number of other definitions.
\begin{definition}\normalfont
	For any node $v$ in the chromatic tree in configuration $C$, let $rebalSet(v, C)$\index{$rebalSet(v, C)$} be the set of nodes in the extended tree $T^*$ including the following 13 nodes:
	\begin{itemize}
		\item $v$ and its four closest ancestors (labeled $p$, $gp$, $ggp$, and $gggp$), 
		\item $p$'s sibling, and
		\item all nodes with depth 2 or less in the subtree rooted at $v$'s sibling. (See Figure~\ref{v_targets}.)
	\end{itemize}
Define $rebalSet(\textsc{Nil}, C) = \emptyset$.
\end{definition}

\begin{figure}[!h]
	\centering
	\begin{tikzpicture}[-, >=stealth', 
	level distance=0.7cm,
	level 3/.style={sibling distance=1.4cm, level distance = 0.6cm},
	level 6/.style={sibling distance=0.7cm, level distance = 0.6cm}
	] 
	\node [arn_w, label=right:{$gggp = targetRoot(cp)$}]{ } 
	child{node [arn_w, label=right:{$ggp$}] { }
		child{node [arn_w, label=right:{$gp$}] { }
			child{node [arn_w, label=right:{$p$}] { }   
				child{ node [arn_w, label=left:{sibling of $v$}] { } 
					child{ node [arn_w] { } 
						child{ node [arn_w] { } }
						child{ node [arn_w] { } }
					}
					child{ node [arn_w] { } 
						child{ node [arn_w] { } }
						child{ node [arn_w] { } }
					}
				}
				child{ node [arn_g, label=right:{$v$}] { } }
			}
			child{node [arn_w, label=right:{sibling of $p$}] { }}   
		}
	}; 
	\end{tikzpicture}
	\caption{The set $rebalSet(v)$ for a node $v$ (highlighted gray).}\label{v_targets} 
\end{figure}
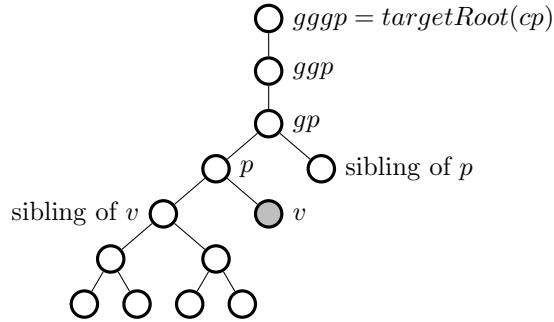

\noindent Notice that any rebalancing transformation centered at $v$ involves a subset of the nodes contained in $rebalSet(v)$. The tree $T^*$ is defined so that for any node $v$ in $T$ in configuration $C$, all 13 nodes in $rebalSet(v,C)$ exist in $T^*$.

Recall that $location(P,C)$ for a process $P$ in its cleanup phase is defined as the node pointed to by $P$'s local variable $l$ in \textsc{BacktrackingCleanup} in configuration $C$.

\begin{definition}\label{currentViol}\normalfont
For each cleanup phase $cp$ by a process $P$ and configuration $C$, we define an auxiliary variable $currentViol(cp)$ that is initially $\textsc{Nil}$ at the start of $cp$ and updated by the following rules: 
\begin{itemize}
\item Rule CV1: A step $s$ changes $currentViol(cp)$ from \textsc{Nil} to a node $v$ if 
	\begin{enumerate}
	\item $cp$ visits $v$ during $s$,
	\item $v$ is a node in the chromatic tree containing a violation, and
	\item the 3 topmost nodes on $P$'s stack are $v$'s three closest ancestors.
	\end{enumerate}
\item Rule CV2: An update CAS changes $currentViol(cp)$ from a node $v \neq \textsc{Nil}$ to $\textsc{Nil}$ if $v$ or one of its three closest ancestors are removed from the chromatic tree by the update CAS.
\end{itemize}
\end{definition}

Intuitively, when $cp$ is performing a rebalancing transformation centered at a node $v$ in the chromatic tree in a configuration $C$, $currentViol(cp,C) = v$. The definition of $currentViol(cp)$ gives conditions where $cp$ will encounter the same violation $v$ between consecutive cleanup attempts, provided no update CAS occurs during these attempts.


\begin{lemma}\label{currentViol_main}\normalfont
Consider a process $P$ in cleanup phase $cp$ in configuration $C$. If $currentViol(cp,C) = v \neq \textsc{Nil}$, then 
\begin{enumerate}
	\item $v$ is in the chromatic tree in $C$,
	\item $location(P,C) \in \{v, p, gp, ggp \}$, where $p$, $gp$, and $ggp$ are the three closest ancestors of $v$,
	\item there is a violation at $v$ on $P$'s search path,
	\item there is no violation at $p$, $gp$, or $ggp$, and
	\item if $C$ is immediately after a step that updates $P$'s local variable $l$, then the nodes in the sequence $\langle p,gp,ggp\rangle$ that are proper ancestors of $location(P,C)$ are the topmost nodes on $P$'s stack in $C$.
\end{enumerate}
\end{lemma}

\begin{proof}
We prove by induction on the sequence of steps in the execution. Consider a step $s$ in the execution, and let $C^-$ and $C$ be the configurations immediately before and after $s$. In the base case, suppose $currentViol(cp,C^-) = \textsc{Nil}$ and $currentViol(cp,C) = v$. Then by rule CV1 of Definition~\ref{currentViol}, $location(P,C) = v$, where $v$ is a node in the chromatic tree containing a violation. So conditions 1, 2, and 3 of the lemma are satisfied. Additionally, the topmost nodes on $P$'s stack are $p$, $gp$, and $ggp$. Thus, condition 5 is satisfied. Finally, by Lemma~\ref{no_viol_stack}, nodes on a process's stack during its cleanup phase do not contain violations, so condition 4 is satisfied. 

So suppose the lemma holds in $C^-$ where $currentViol(cp,C^-) = v$, and we show the lemma holds in $C$. By rule CV2 of Definition~\ref{currentViol}, $s$ is not an update CAS that removes one of $v$, $p$, $gp$, or $ggp$ from the chromatic tree, otherwise $currentViol(cp,C) = \textsc{Nil}$. By Lemma~\ref{no_viol_added}, violations are only added to a node when it is first added to the chromatic tree. Therefore, since conditions 1, 3 and 4 are satisfied in $C^-$, they are satisfied in $C$.  Note that if $s$ does not update the local variable $l$, then since conditions 2 and 5 are satisfied in $C^-$, they will be satisfied in $C$. In this case, $currentViol(P,C) = v$. So suppose $s$ updates $P$'s local variable $l$. We first argue that $s$ is not a step that updates $P$'s local variable $l$ on line~\ref{ln:cleanup:entry} or \ref{ln:cleanup:pop}. 
\begin{itemize}
	\item Suppose step $s$ updates the local variable $l$ to $entry$ on line~\ref{ln:cleanup:entry}. Then $cp$ is the first cleanup attempt of $P$'s current operation. So $l$ is undefined before $s$, and $s$ is the first step that updates $P$'s local variable $l$ in $cp$. Thus, $location(P,C^-) = \textsc{Nil}$.
	
	\item Suppose step $s$ updates the local variable $l$ by the \textsc{Pop} on line~\ref{ln:cleanup:pop}. The last time $P$ executed line~\ref{ln:cleanup:backtracking_start}, $P$ read that the node pointed to by $l$ was marked. Thus $location(P,C^-)$ is a marked node in $C^-$. By Lemma~\ref{backtracking_marked}, $location(P,C^-)$ is not a node in the chromatic tree in $C^-$. Since the nodes $v$, $p$, $gp$, and $ggp$ are in the chromatic tree, $location(P,C^-) \notin \{v,p,gp,ggp\}$. Thus, $currentViol(P,C^-) = \textsc{Nil}$.
\end{itemize}

The only remaining updates to the local variable $l$ are on lines~\ref{ln:cleanup:first_pop} and \ref{ln:cleanup:leaf_update} of \textsc{BacktrackingCleanup}. 
\begin{itemize}
	\item Suppose step $s$ updates the local variable $l$ by the \textsc{Pop} on line~\ref{ln:cleanup:first_pop}. Then this is the first update to $l$ in $P$'s current attempt. The previous attempt by $P$ ended with an invocation $I$ of \textsc{TryRebalance}. Since $I$ is called when a violation is found at the node pointed to by $l$, and the only node that contains a violation among the nodes $\{v,p,gp,ggp\}$ is $v$, $l$ points to $v$ when $I$ is invoked. Note that $P$'s local variable $l$ is not updated between the start of $I$ and $C^-$, and so $location(P,C^-) = v$. Then by condition 5, the last nodes visited by $P$ in $C^-$ were the nodes $\{v,p,gp,ggp\}$. Thus, in the last three iterations of the while-loop on line~\ref{ln:cleanup:forward_start} prior to $I$, $P$'s local variable $l$ pointed to each of $p$, $gp$, and $ggp$. These nodes are pushed onto $P$'s stack on line~\ref{ln:cleanup:push}. Then when $P$ finds a violation at $v$ on line~\ref{ln:cleanup:viol_check} in the iteration $I$ is invoked, the top nodes on $P$'s stack are $p$, $gp$, and $ggp$. Thus, $p$ is the node popped by $P$ on line~\ref{ln:cleanup:p_pop}, $gp$ is the node popped by $P$ on line~\ref{ln:cleanup:gp_pop}, $ggp$ is the top node on $P$'s stack read on line~\ref{ln:cleanup:ggp_top}. So $ggp$ is the topmost node on $P$'s local stack in $C$. Since $s$ updates the local variable $l$ to the top node on $P$'s stack, $location(P,C) = ggp$. There are no proper ancestors of $ggp$ in $\langle p,gp,ggp \rangle$, so condition 5 is vacuously true in $C$. So $currentViol(P,C) = v$.
	
	\item Suppose step $s$ updates the local variable $l$ on line~\ref{ln:cleanup:leaf_update}. By the check on line~\ref{ln:cleanup:viol_check}, $l$ does not point to a node containing a violation in $C$. Thus, $location(P,C^-) \in \{p,gp,ggp\}$. The step $s$ updates the local variable $l$ to the descendant of $location(P,C^-)$ on $P$'s search path. So $location(P,C) \in \{v, p, gp\}$. In configuration $C^-$, the topmost nodes on $P$'s stack are the proper ancestors of $location(P,C^-)$ in $\langle p, gp, ggp \rangle$.  Since $location(P,C^-)$ is pushed onto $P$'s stack before $s$ on line~\ref{ln:cleanup:push}, the topmost nodes on $P$'s stack in $C$ are the proper ancestors of $location(P,C)$. So $currentViol(P,C) = v$.

\end{itemize}
\end{proof}

\begin{definition}\label{nextNode}\normalfont
	Suppose that $P$ is in its cleanup phase in configuration $C$. Consider the solo execution by $P$ from $C$ until it finishes its cleanup phase. If its local variable $l$ does not point to any node in the chromatic tree in $C$ during this solo execution, then define $\mathit{nextNode}(P, C) = \textsc{Nil}$. Otherwise, define $\mathit{nextNode}(P, C)$ to be the node in the chromatic tree in $C$ that is first pointed to by $l$ during this execution.
\end{definition}

Note that by definition, if $nextNode(P,C) \neq \textsc{Nil}$, then $nextNode(P,C)$ is a node in the chromatic tree in $C$. We next prove properties of $\mathit{nextNode}(P, C)$.

\begin{lemma}\label{nextNode_backtracking}\normalfont
Suppose $\mathit{nextNode}(P, C) \neq \textsc{Nil}$ and $P$ performs backtracking at some point during its solo execution starting from $C$ prior to visiting $\mathit{nextNode}(P, C)$. Then in configuration $C$, $\mathit{nextNode}(P, C)$ is a node on $P$'s stack.
\end{lemma}

\begin{proof}

By the definition of $nextNode$, $P$ does not visit any nodes that are in the chromatic tree in $C$ before the point in which $P$ first begins backtracking. By definition, backtracking only ends when $P$ visits an unmarked node, which by Observation~\ref{stack_obs}.\ref{stack_obs:config} and Lemma~\ref{unmarked}, is in the chromatic tree. Thus, $P$ is backtracking when it first visits $nextNode(P,C)$ in its solo execution starting from $C$. During backtracking, $P$ only visits nodes on its stack. Hence, $P$ is a node on $P$'s stack.
\end{proof}

\begin{lemma}\label{nextNode_obs}\normalfont
	If $\mathit{nextNode}(P, C) \neq \textsc{Nil}$, then in configuration $C$, it is either a node on $P$'s stack, or it is the first node on $P$'s search path starting from $location(P,C)$ that is in the chromatic tree.
\end{lemma}

\begin{proof}
Suppose $P$ performs backtracking at some point during its solo execution starting from $C$ but before the step in which $P$ first visits $\mathit{nextNode}(P, C)$. By Lemma~\ref{nextNode_backtracking}, $\mathit{nextNode}(P, C)$ is a node on $P$'s stack in $C$.

So suppose $P$ does not perform backtracking at some point during its solo execution starting from $C$ but before the step in which $P$ first visits $\mathit{nextNode}(P, C)$. Note that $P$ is not in an invocation of \textsc{TryRebalance} throughout its solo execution, since $P$ does not visit nodes in \textsc{TryRebalance} and $P$ begins backtracking immediately after \textsc{TryRebalance}. So $P$ is looking for a violation throughout its solo execution. From the code, $P$ only updates its local variable $l$ when it executes line~\ref{ln:cleanup:leaf_update} of \textsc{BacktrackingCleanup}. Each node $P$ visits when looking for a violation during its solo execution starting from $C$ is on its search path from $location(P,C)$. Thus, $nextNode(P,C)$, the first node that $P$ visits in its solo execution starting from $C$ that is in the chromatic tree in $C$, is on $P$'s search path from $location(P,C)$
\end{proof}

\begin{lemma}\label{nextNode_visit}\normalfont
	Let $P$ be a process in a configuration $C$ such that $\mathit{nextNode}(P,C) \neq \textsc{Nil}$. Then in any execution starting from $C$, $P$ only visits nodes that are no longer in the chromatic tree in $C$ before visiting $nextNode(P,C)$.
\end{lemma}

\begin{proof}
Consider the sequence of nodes $\langle x_1,\dots,x_k \rangle$ visited by $P$ in its solo execution starting from $C$, where $x_1 = location(P,C)$ and $x_k = \mathit{nextNode}(P,C)$. By definition, $\langle x_1,\dots,x_{k-1} \rangle$ are not nodes in the chromatic tree, and so are marked. Since the fields of marked nodes are not changed, $x_1,\dots,x_k$ have the same child pointers in all configurations $C'$ reachable from $C$. Thus, $x_1,\dots,x_k$ are the first nodes visited by $P$ in any execution starting from $C$. 
\end{proof}

\begin{lemma}\label{nextNode_nil}\normalfont
If $\mathit{nextNode}(P, C) = \textsc{Nil}$, then for any configuration $C'$ after $C$ but before the end of $P$'s current operation, $\mathit{nextNode}(P, C') = \textsc{Nil}$.
\end{lemma}

\begin{proof}
Consider the solo execution by $P$ from $C$ until it finishes its cleanup phase. We first prove $P$ does not pop from its stack on line~\ref{ln:cleanup:first_pop} or \ref{ln:cleanup:pop} of \textsc{BacktrackingCleanup}. Suppose, for contradiction, that $P$ does perform such a pop. Since $\mathit{nextNode}(P,C) = \textsc{Nil}$, every node that $P$ visits in this solo execution is not in the chromatic tree and also marked. Then when $P$ first executes line~\ref{ln:cleanup:backtracking_start} following $C$, it reads that the node pointed to by $l$ is marked. This implies that $P$ does not exit the while-loop on line~\ref{ln:cleanup:backtracking_start}. Then $P$ will eventually pop $entry$ from the stack on line~\ref{ln:cleanup:pop}. When $P$ visits $entry$, this contradicts the fact that $entry$ is always a node in the chromatic tree.

Note that $P$ does not invoke \textsc{TryRebalance} in its solo execution from $C$. Otherwise, a pop on line~\ref{ln:cleanup:first_pop} will be executed after this instance of \textsc{TryRebalance} is completed. Therefore, the check on line~\ref{ln:cleanup:viol_check} never finds a violation in $P$'s solo run from $C$. Thus, $P$ will eventually visit a leaf node and exit its current operation. Let $x_1, ..., x_k$ be the nodes visited by $P$, where $x_k$ is the leaf node. The nodes $x_1, ..., x_k$ are not in the chromatic tree and hence they are marked. Since the fields of marked nodes are not changed, $x_1, ..., x_k$ have the same child pointers in all configurations $C'$ reachable from $C$. Therefore, $P$ will only visit the nodes in $\{x_1, ..., x_k\}$ in a solo run execution starting from $C'$. Thus, $\mathit{nextNode}(P, C') = \textsc{Nil}$.
\end{proof}

\begin{definition}\label{focalNode}\normalfont
	Consider a process $P$ in its cleanup phase $cp$ in configuration $C$.\index{$\mathit{focalNode}(cp, C)$} Let $\mathit{focalNode}(cp, C)$ be a node in the chromatic tree in $C$ or \textsc{Nil} according to the following cases.
	\begin{enumerate}	
		\item If $\mathit{currentViol}(cp, C) \neq \textsc{Nil}$, then $\mathit{focalNode}(cp, C) = \mathit{currentViol}(cp,C)$.
		
		\item If $\mathit{currentViol}(cp, C) = \textsc{Nil}$ and $location(P,C)$ is in the chromatic tree, then $\mathit{focalNode}(cp, C) = location(P,C)$.
		
		\item If $\mathit{currentViol}(cp, C) = \textsc{Nil}$ and $location(P,C)$ is not in the chromatic tree, then $\mathit{focalNode}(cp, C) = nextNode(P, C)$.
	\end{enumerate}
\end{definition}

\noindent It follows by definition that $\mathit{focalNode}(cp, C)$ is a node in the chromatic tree in $C$, provided $\mathit{focalNode}(cp, C) \neq \textsc{Nil}$. The following lemmas prove properties of $\mathit{focalNode}(cp, C)$.

\begin{lemma}\label{focalNode_loc_tree}\normalfont
	For any process $P$ in its cleanup phase $cp$ and configuration $C$, if $location(P,C)$ is in the chromatic tree, then $location(P,C)$ is an ancestor of $\mathit{focalNode}(cp,C)$ and the distance between $location(P,C)$ and $\mathit{focalNode}(cp,C)$ is at most 3.
\end{lemma}

\begin{proof}
If $\mathit{focalNode}(cp,C) = location(P,C)$, the lemma holds trivially. If $\mathit{focalNode}(cp,C) = currentViol(cp,C)$, then by Lemma~\ref{currentViol_main}, $location(P,C)$ is an ancestor of $\mathit{focalNode}(cp,C)$ and the distance between $location(P,C)$ and $\mathit{focalNode}(cp,C)$ is at most 3.
\end{proof}

\begin{lemma}\label{focalNode_nn}\normalfont
	For any process $P$ in its cleanup phase $cp$ and configuration $C$, if $location(P,C)$ is not in the chromatic tree, then $focalNode(cp,C) = nextNode(P,C)$.
\end{lemma}

\begin{proof}
By Lemma~\ref{currentViol_main}, if $location(P,C)$ is not in the chromatic tree, then $currentViol(cp,C) = \textsc{Nil}$. Therefore, by the definition of the focal node, $focalNode(cp,C) = nextNode(P,C)$.
\end{proof}

\begin{lemma}\label{focalNode_nn_neq_loc}\normalfont
For any process $P$ in its cleanup phase $cp$ and configuration $C$,  $nextNode(P,C) \neq location(P,C)$ if and only if $\mathit{focalNode}(cp,C) = nextNode(P,C)$.
\end{lemma}

\begin{proof}
Suppose $nextNode(P,C) \neq location(P,C)$. By Definition~\ref{nextNode}, $location(P,C)$ is not in the chromatic tree in $C$. If $currentViol(cp,C) \neq \textsc{Nil}$, then by Lemma~\ref{currentViol_main}, $location(P,C)$ is a node in the chromatic tree, a contradiction. Therefore, $currentViol(cp,C) = \textsc{Nil}$. Therefore, by Definition~\ref{focalNode}.3, $\mathit{focalNode}(cp,C) = nextNode(P,C)$.

Suppose $\mathit{focalNode}(cp,C) = nextNode(P,C)$. By Definition~\ref{focalNode}.3, $location(P,C)$ is not in the chromatic tree in $C$. Since $P$'s local variable $l$ does not point to a node in the chromatic tree in $C$, $nextNode(P,C) \neq location(P,C)$.
\end{proof}

\begin{lemma}\label{focalNode_on_sp}\normalfont
	For every cleanup phase $cp$ and configuration $C$, if $\mathit{focalNode}(cp, C) \neq \textsc{Nil}$, then $\mathit{focalNode}(cp, C)$ is on $cp$'s search path in $C$.
\end{lemma}

\begin{proof}
Let $P$ be the process in cleanup phase $cp$. We prove the lemma by considering each case in Definition~\ref{focalNode}.

Suppose $\mathit{focalNode}(cp,C) = currentViol(cp,C)$. Then $P$ has visited $v$ in some step before $C$. By Lemma~\ref{reach_SP_backtracking_cleanup}, $currentViol(cp,C)$ is on $P$'s search path in $C$. 

Suppose $\mathit{focalNode}(cp,C) = location(P,C)$. By Corollary~\ref{reach_SP_backtracking_cleanup}, there is a configuration before $C$ in which $location(P,C)$ is on $P$'s search path. Therefore, by Lemma~\ref{hindsight}, $location(P,C)$ is on $P$'s search path in $C$. 

Suppose $\mathit{focalNode}(cp,C) = nextNode(P,C)$. By definition, $nextNode(P,C)$ is a node in the chromatic tree. If $nextNode(P,C)$ is on $P$'s stack, then by Lemma~\ref{on_stack_SP_2}, $nextNode(P,C)$ is on $P$'s search path in $C$. If $nextNode(P,C)$ is not on $P$'s stack, then it is the first node in the chromatic tree that $P$ will visit in a search starting from $location(P,C)$. Let $C'$ be the first configuration in which $P$ visits $nextNode(P,C)$ in its solo execution starting from $C$. By Corollary~\ref{reach_SP_backtracking_cleanup} and Lemma~\ref{hindsight}, $nextNode(P,C)$ is on $P$'s search path in the configuration $C'$. Since the structure of the chromatic tree does not change from between $C$ and $C'$, $nextNode(P,C)$ is on $P$'s search path in $C$. 
\end{proof}

\begin{definition}\label{def_target_cp}\normalfont
For a cleanup phase $cp$ in configuration $C$, let \index{$targets(cp, C)$} $targets(cp, C) = rebalSet(\mathit{focalNode}(cp, C), C)$.
\end{definition}

Let $targetRoot(cp,C)$ be the root node of the tree formed by nodes in $targets(cp,C)$. The following lemma serves a similar role as Lemma~\ref{targets_fixed}. 

\begin{lemma}\label{currentViol_coro2}\normalfont
	Consider the invocation of \textsc{TryRebalance}$(ggp,gp,p,v)$ during an attempt $A$ of a cleanup phase $cp$. Let $C$ be the configuration immediately after this invocation of \textsc{TryRebalance}. If $currentViol(cp,C) = v \neq \textsc{Nil}$ and no update CAS occurs between $C$ and the end of $A$, then $cp$ only performs LLXs on nodes in $targets(cp,C) - \{targetRoot(cp,C)\}$.
\end{lemma}

\begin{proof}
Since $currentViol(cp,C) = v$, it follows by definition that $\mathit{focalNode}(cp, C) = v$ and $targets(cp,C) = rebalSet(v, C)$. Since no update CAS occurs between $C$ and the end of $A$, by Definition~\ref{currentViol}, $currentViol(cp)$ does not change between $C$ and the end of $A$. Thus, $targets(cp)$ does not change between $C$ and the end of $A$. By inspection of the code of \textsc{TryRebalance}, $cp$ will only perform LLXs on nodes in $targets(cp,C)$, no matter which rebalancing transformation is performed.
\end{proof}

We next describe the conditions in which $currentViol(cp)$ is equal to some node $v$ which contains a violation. Recall that an invocation $I$ of \textsc{TryRebalance}$(ggp,gp,p,v)$ is called \textit{stale} if one of $v$, $p$, $gp$, and $ggp$ is not in the chromatic tree. A \textit{stale} invocation of \textsc{TryRebalance} will fail since an LLX on a node not in the chromatic tree returns \textsc{Finalized}.  A \textit{stale} invocation of \textsc{TryRebalance}$(ggp,gp,p,v)$ blames a node $x$ if $x$ is the first node in the sequence $\langle ggp,gp,p,v \rangle$ that is not in the chromatic tree in $C$. Our analysis pays for all failed attempts caused by stale invocations of \textsc{TryRebalance} that blame a node $x$ using the auxiliary bank accounts $S(cp,x)$. Each $S(cp,x)$ account is initially empty. It is activated in the first configuration in which $cp$ is active and $x$ is in the chromatic tree. Once activated, it is updated as follows: 
\begin{itemize}
	\item D1-S: Consider a \textbf{successful update CAS} $ucas$ in configuration $C$ for an \textsc{Insert} or \textsc{Delete} transformation by $xp$, or for a rebalancing transformation centered at $viol(xp)$. For all nodes $x$ removed from the chromatic tree by $ucas$ and for all cleanup phases $cp'$ active during $C$, $xp$ deposits 1 dollar into $S(cp', x)$.
	\item W-STALE: A \textbf{stale invocation} of \textsc{TryRebalance} for a cleanup phase $cp$ that blames a node $x$ withdraws 1 dollar from $S(cp, x)$.
\end{itemize}

\begin{lemma}\label{currentViol_stale}\normalfont
	Consider a non-stale invocation $I$ of \textsc{TryRebalance}$(ggp, gp, p, v)$ by a process in a cleanup phase $cp$. Then $currentViol(cp, C) = v$, where $C$ is the configuration immediately after $I$ is invoked.
\end{lemma}

\begin{proof} 
From the code, $cp$ invokes \textsc{TryRebalance}$(ggp, gp, p, v)$ if $cp$'s local variable $l$ points to $v'$. Consider the step $s$ that last updated $l$ to point to $v$. Let $C'$ be the step immediately after this step. From the code, if $cp$ invokes \textsc{TryRebalance}$(ggp, gp, p, v)$, then the topmost nodes on $P$'s stack are $p$, $gp$, and $ggp$ in $C'$. Therefore, by rule CV1 of Definition~\ref{currentViol}, $currentViol(cp,C') = v$. Since $I$ is non-stale, $ggp$, $gp$, $p$, and $v$ are still in the chromatic tree immediately before the start of $I$. Thus, no update CAS has removed one of $ggp$, $gp$, $p$ or $v$. Since no step applies rule CV2 of Definition~\ref{currentViol} between $C'$ and $C$, $currentViol(cp,C) = v$.
\end{proof}

\begin{lemma}\label{stale_1_fail}\normalfont
A node $x$ is blamed by at most one invocation of \textsc{TryRebalance} for each cleanup phase $cp$.
\end{lemma}

\begin{proof} 
Let $I$ be a stale invocation of \textsc{TryRebalance}$(ggp,gp,p,v)$ that blames a node $x$, where $x \in \{ggp,gp,p,v\}$. From the code, $gp$, $p$, and $v$ are not on $cp$'s stack at the start of $I$. After $I$ fails, $cp$ begins a new cleanup attempt. Let $A$ be any attempt after $I$. Let $m$ the last node popped during backtracking during $A$. By the check on line~\ref{ln:cleanup:backtracking_start} of \textsc{BacktrackingCleanup}, $m$ is unmarked and therefore in the chromatic tree. So $m \neq x$. Since $x$ is not in the chromatic tree when $I$ was invoked and is never added back to the chromatic tree, $x$ is not reachable from $m$. Since $cp$'s search during $A$ does not visit the node $x$ and is not on $cp$'s stack, it follows from the code of \textsc{BacktrackingCleanup} that $cp$ does not invoke \textsc{TryRebalance} with $x$ as a parameter during $A$. Thus, $cp$ does not blame $x$ for a stale invocation of \textsc{TryRebalance} in any attempt that begins after the end of $I$.
\end{proof}

\begin{lemma}\label{S_positive}\normalfont
	For all cleanup phases $cp$, nodes $x$, and configurations $C$, $S(cp, x, C) \geq 0$.
\end{lemma}

\begin{proof} 
Suppose for contradiction, that there exists a cleanup phase $cp$, node $x$, and configuration $C$ such that $S(cp, x, C) < 0$. Let $C^-$ be the configuration immediately before $C$. Then in the step immediately prior to $C$, $cp$ performs a stale invocation $I$ of \textsc{TryRebalance} that blames $x$ and $S(cp, x, C^-) = 0$. By definition, $x$ is not in the chromatic tree in $C^-$. 

Consider the configuration $C'$ immediately before $x$ was removed from the chromatic tree. Node $x$ must have been reachable by $cp$ in some configuration before $C'$ for it to perform $I$ with $x$ as a parameter, and so $cp$ was active during $C'$. Therefore, by rule D1-S, a dollar is deposited into $S(cp, x)$ by the update or cleanup phase that removed $x$ from the chromatic tree. Therefore, $S(cp, x, C') > 0$. Since $S(cp, x, C^-) = 0$, there exists a stale invocation $I' \neq I$ of \textsc{TryRebalance} in some configuration between $C'$ and $C^-$ that blames $x$ and withdraws a dollar from $S(cp, x)$ by W-STALE. This contradicts Lemma~\ref{stale_1_fail}, which states that $cp$ can only blame $x$ for one stale invocation of \textsc{TryRebalance}. 
\end{proof}

Lemma~\ref{S_positive} proves that all unsuccessful cleanup attempts caused by stale invocations of \textsc{TryRebalance} can be paid for by the $S$ accounts. Failed attempts due to non-stale invocations of \textsc{TryRebalance} by $cp$ are accounted for by $cp$'s other accounts, as described in the following sections.

\subsubsection{The $L_i$ and $F_i$ Accounts}\label{section_aux_accounts}

In this section, we prove properties of the auxiliary bank accounts $L_i(xp)$ (for $1 \leq i \leq \#llx(xp)$) and $F_i(xp)$ (for $1 \leq i \leq \#\mathit{frz}(xp)$) for an update phase or cleanup phase $xp$. Intuitively, the $L_i(xp)$ accounts pay for the failed LLXs that occur within $xp$'s first few attempts, or when successful update CASs have occurred sometime within $xp$'s last few attempts. Likewise, the $F_i(xp)$ accounts pay for failed SCXs in these same cases. The remaining failed LLXs and SCXs are paid for by the $B_{llx}(xp)$ and $B_{scx}(xp)$ accounts, respectively.

By rule W-LLX, a dollar is withdrawn from an $L_i$ account as the result of a failure step of an LLX, provided the $L_i$ account is non-empty. Therefore, the $L_i$ accounts always have a non-negative balance. We next prove the properties of the $L_i$ accounts in the configurations that they are empty and a failure step of an LLX by occurs.

\begin{lemma}\label{L_property}\normalfont	
	Suppose $A$ is an attempt of a phase $xp$ that fails an LLX which blames an SCX $S$. Let $i$ be the number of LLXs performed by $xp$ during $A$ (so the $i$th LLX by $xp$ is unsuccessful). Let $C$ be the configuration immediately before the failure step of $xp$'s unsuccessful LLX. If $L_i(xp, C) = 0$, then there exists a previous attempt $A'$ of $xp$ such that
	\begin{enumerate}
		\item no successful update CAS occurs between the start of $A'$ and $C$,
		\item for every configuration $C'$ between the starting configuration of $A'$ and $C$, $targets(xp, C) = targets(xp, C')$,
		\item $\mathit{LLXNode}_i(xp,A) \in targets(xp,C) -  \{ targetRoot(xp,C)\}$,
		\item the $i$th LLX performed by $xp$ during $A'$ is on $\mathit{LLXNode}_i(xp,A)$, which returns \textsc{Fail} or \textsc{Finalized},
		\item the failed LLXs on $\mathit{LLXNode}_i(xp,A)$ in $A'$ and $A$ blame different invocations of SCX, and
		\item each successful freezing CAS belonging to $S$ up to and including a successful freezing CAS on $\mathit{LLXNode}_i(xp,A)$ occurs between the start of $A'$ and $C$.
	\end{enumerate}
\end{lemma}

\begin{proof}
First we suppose $xp$ is an update phase $up$. By D1-LF, $L_i(up)$ has 4 dollars in the configuration immediately after $up$ is invoked. So $A$ is not among the first four attempts of $up$. Since $L_i(up, C) = 0$ and by W-LLX, a failed attempt can deduct at most 1 dollar from $L_i(up)$, there exists 4 attempts $A^{(4)}$, $A'''$, $A''$ and $A'$ prior to $A$ (performed by $up$ in this order) such that each of $A^{(4)}$, $A'''$, $A''$, and $A'$ deduct 1 dollar from $L_i(up)$. Furthermore, $L_i(up)$ contains 4 dollars sometime during $A^{(4)}$, 3 dollars at the start of $A'''$, 2 dollars at the start of $A''$, 1 dollar at the start of $A'$, and is empty at the start of $A$. 

Since an update CAS by any concurrent operation adds 4 dollars to $L_i(up)$ by D2-LF, no update CAS occurs between the start of $A'''$ and $C$, proving part 1. Since $targets(up)$ only changes due to an update CAS, $targets(up, C) = targets(up, C')$ for every configuration $C'$ between the starting configuration of $A'''$ and $C$, proving part 2. By Lemma~\ref{targets_fixed}, the nodes on which $up$ performs its LLXs in attempts $A''$, $A'$ and $A$ are in $targets(up,C) - \{ targetRoot(up,C)\}$. In particular, $\mathit{LLXNode}_i(up,A) \in targets(up,C) - \{ targetRoot(up,C)\}$, proving part 3. By Lemma~\ref{targets_fixed_invocations}, $up$ performs an invocation of \textsc{TryInsert} or \textsc{TryDelete} with the same parameters in $A''$, $A'$ and $A$. Since the structure of the chromatic tree does not change from between the start of $A''$ and $C$, the order of nodes in which $up$ performs its LLXs on nodes are the same in $A''$, $A'$ and $A$. Since a dollar is deducted from $L_i(up)$ in $A'$ and $A$, $A'$ and $A$ fail an LLX on the node $\mathit{LLXNode}_i(up,A)$, proving part 4. Note that since $\mathit{LLXNode}_i(up,A) \in targets(up,C)$, $\mathit{LLXNode}_i(up,A)$ is in the chromatic tree between the start of $A''$ and $C$. Hence, by Lemma~\ref{llx_blame_3}, the failed LLXs on $\mathit{LLXNode}_i(up,A)$ in $A'$ and $A$ blame different invocations of SCX, proving part 5. By Lemma~\ref{llx_2_attempts}, each successful freezing CAS belonging to $S$ up to its successful freezing CAS on $\mathit{LLXNode}_i(up,A)$ occurs sometime between the start of $A'$ and $C$. This proves part 6.

Now suppose $xp$ is a cleanup phase $cp$. By D1-LF, $L_i(cp)$ has 3 dollars in the configuration immediately after $cp$ is invoked. Since a failed attempt can deduct at most 1 dollar from $L_i(up)$ via rule W-LLX, $A$ is not among the first two attempts of $cp$. Since $L_i(cp, C) = 0$, there exists 3 attempts $A'''$, $A''$ and $A'$ prior to $A$ (performed by $cp$ in this order) such that they each deduct 1 dollar from $L_i(cp)$. Furthermore, $L_i(cp)$ contains 2 dollars sometime during $A'''$, 1 dollar at the start of $A''$, 1 dollar at the start of $A'$, and is empty at the start of $A$. 

Since an update CAS by any concurrent operation adds 3 dollars to $L_i(cp)$ by D2-LF, no update CAS occurs between the start of $A'$ and $C$, proving part 1. Let $C''$ be the configuration immediately after $cp$ invokes \textsc{TryRebalance}$(ggp,gp,p,v)$ during $A'''$. By Lemma~\ref{currentViol_stale}, $currentViol(cp,C'') = v$. Additionally, by Definition~\ref{currentViol}, $currentViol(cp)$ remains set to $v$ in all configurations between $C''$ and $C$. Thus, by Definition~\ref{focalNode}.1, $\mathit{focalNode}(cp)$ does not change from the beginning of $A''$ and $C$. Then $targets(cp)$ does not change in this interval, proving part 2. By Lemma~\ref{currentViol_coro2}, the nodes in which $cp$ performs its LLXs in attempts $A''$, $A'$ and $A$ are in $targets(cp)$, proving part 3. Since a dollar is deducted from $L_i(cp)$ in $A''$, $A'$ and $A$, $A''$, $A'$ and $A$ fail an LLX on $\mathit{LLXNode}_i(cp,A)$, proving part 4. Note that since $\mathit{LLXNode}_i(cp,A) \in targets(cp,C)$, $\mathit{LLXNode}_i(cp,A)$ is in the chromatic tree between the start of $A'$ and $C$. Hence, by Lemma~\ref{llx_blame_3}, the failed LLXs on $\mathit{LLXNode}_i(cp,A)$ in $A'$ and $A$ blame different invocations of SCX, proving part 5. By Lemma~\ref{llx_2_attempts}, each successful freezing CAS belonging to $S$ up to its successful freezing CAS on $\mathit{LLXNode}_i(cp,A)$ occurs sometime between the start of $A'$ and $C$. This proves part 6.
\end{proof}

By rule W-SCX, a dollar is withdrawn from an $F_i$ account as the result of a failure step of an SCX, provided the $F_i$ account is non-empty. We can prove a similar result to \ref{L_property} for the bank account $F_i(up)$.
\begin{lemma}\label{F_property}\normalfont
	Suppose $A$ is an attempt of a phase $xp$ that fails a freezing iteration. Let $i$ be the number of freezing performed by $up$ during $A$ (so the $i$th freezing iteration by $up$ is unsuccessful). Let $C$ be the configuration immediately before the failure step of this SCX. Let $S$ be the SCX blamed by $xp$'s unsuccessful SCX in $A$. If $F_i(xp, C) = 0$, then there exists a previous attempt $A'$ of $xp$ such that
	\begin{enumerate}
		\item no successful update CAS occurs between the start of $A'$ and $C$,
		\item for every configuration $C'$ between the starting configuration of $A'$ and $C$, $targets(xp, C) = targets(xp, C')$,
		\item $\mathit{freezeNode}_i(up,A) \in targets(xp,C) - \{targetRoot(xp,C)\}$, 
		\item the $i$th freezing iteration performed by $xp$ during $A'$ is on $\mathit{freezeNode}_i(xp,A)$, which returns \textsc{False}, and
		\item each successful freezing CAS belonging to $S$ up to and including a successful freezing CAS on $\mathit{freezeNode}_i(xp,A)$ occurs between the start of $A'$ and $C$.
	\end{enumerate}
\end{lemma}

\begin{proof}
The proof of parts 1, 2 and 3 follow from the proof of Lemma~\ref{L_property} by symmetry of the rules D1-LF and D2-LF to the $L_i$ and $F_i$ accounts, and by symmetry of the rule W-LLX to W-SCX.

If $xp$ is an update phase $up$, then part 4 follows from Lemma~\ref{targets_fixed_invocations} because $up$ performs the same successful LLXs in $A'$ and $A$, and so $up$ performs an invocation of SCX with the same parameters in $A'$ and $A$. If $xp$ is a cleanup phase $cp$, then part 4 follows from Lemma~\ref{currentViol_coro2}, because $cp$ only performs freezing iterations on nodes it has previously performed LLXs on in the same attempt.

Since $xp$ fails a freezing iteration on $\mathit{freezeNode}_i(xp,A)$, there is a successful freezing CAS belonging to $S$. Since all LLXs linked to $xp$'s failed SCX in $A'$ occur during $A'$, by Lemma~\ref{scx_2_attempts}, each successful freezing CAS belonging to $S$ up to and including its successful freezing CAS on $\mathit{freezeNode}_i(xp,A)$ occurs between the start of $A'$ and $C$. This proves part 5.
\end{proof}

\subsubsection{The $B_{llx}$ Accounts}\label{section_llx}
In this section, we prove that the $B_{llx}$ accounts have a non-negative balance. The $B_{llx}$ accounts are responsible for paying for a failed attempt due to a failed LLX whenever the auxiliary accounts cannot pay for the failure. No rules deposit dollars directly into the $B_{llx}$ accounts. Instead, dollars are transferred into $B_{llx}$ accounts from $B$ accounts by T1-B. We argue that the $B_{llx}$ accounts can always pay for a failed attempt when required by W-LLX.

\begin{lemma}\label{LLX1}\normalfont
	For all update phases and cleanup phases $xp$, all nodes $x$ in the chromatic at some point in an execution, and all configurations $C$, $B_{llx}(xp, x, C) \geq 0$.
\end{lemma}

\begin{proof}
A dollar is only deducted from a $B_{llx}$ account owned by $xp$ by rule W-LLX. Let $s$ be the failure step of an LLX$(x)$, $I$, from a configuration $C^-$ immediately before $C$ that applies rule W-LLX. Suppose $I$ is the $i$th LLX performed by $xp$ in its current attempt $A$, so $x = \mathit{LLXNode}_i(xp,A)$. Let $S$ be the SCX blamed by $I$, and let $C'$ be the configuration immediately after $x$ is frozen for $S$. 

By rule W-LLX, $L_i(xp,C^-) = 0$ if a dollar is withdrawn from $B_{llx}$. By Lemma~\ref{L_property}.4, there exists an attempt $A'$ by $xp$ before $C$ in which $xp$ fails an LLX on $x$. Note that by Lemma~\ref{L_property}.1 and Lemma~\ref{L_property}.2, for all configurations between the start of $A'$ and $C$, no update CAS occurs and $x$ is a node in $targets(xp)$. Thus, $x$ is in the chromatic tree during this interval.

We argue that the bank account from which $xp$ withdraws 1 dollar by rule W-LLX during $s$ is non-empty.
\begin{itemize}
	\item Suppose $x$ is a downwards node for $S$, so $s$ withdraws 1 dollar from $B_{llx}(xp, x)$. By Lemma~\ref{L_property}.6, $C'$ occurs between the start of $A'$ and $C$. By Lemma~\ref{L_property}.3, $x \in targets(up,C)$, and thus by Lemma~\ref{L_property}.2, $x \in targets(xp, C')$. Therefore, by T1-B, 1 dollar was transferred from $B(xp)$ to $B_{llx}(xp, x)$ in configuration $C'$. So $B_{llx}(xp, x, C') > 0$. 
	
	We argue that no step $s'$ withdraws from $B_{llx}(xp, x)$ between $C'$ and $C^-$, and so $B_{llx}(xp, x, C) \geq 0$. Consider an unsuccessful LLX $I' \neq I$ that blames an SCX $S'$. The failure step $s'$ of $I'$ withdraws from $B_{llx}(xp, x)$ if $I'$ is performed on $x$ and $x$ is a downwards node for $S'$, or if $I'$ is performed on a node $y$ and $y$ is a child of $x$ and is a cross node for $S'$. Note that since $C'$ is a step contained in $A'$, and there is a step of $I'$ between $C'$ and $C^-$, it follows that $I'$ starts after the start of $A'$. Since $x$ is in the chromatic tree and no successful update CAS occurs between the start of $A'$ and $C^-$, the node on which $I'$ is performed is in the chromatic tree between the start of $I'$ and the end of $I'$.
	
	Suppose $s'$ is on the node $x$ and $x$ is a downwards node for $S'$. Since $x$ is frozen for $S$ in all configurations between $C'$ and $C$, $x.\mathit{info}$ only points to $S$ immediately before $s$, and so $I'$ blames $S$. Thus, by Lemma~\ref{llx_blame_3}, $I$ does not blame $S$, a contradiction.
	
	Suppose $s'$ is on the node $y$ and $y$ is a cross node for $S'$. Note that a cross node is only frozen for some SCX $S'$ if its parent node is already frozen $S'$. Since $x$ is frozen for $S$ in all configurations between $C'$ and $C$, $y$ is a cross node for $S$. By Lemma~\ref{llx_no_ucas_abort}, an abort step is performed for $S$ before the end of $I'$. By Lemma~\ref{llx_abort_blame_once}, $I$ does not blame $S$, a contradiction.
	
	\item Suppose $x$ is a cross node for $S$, so $s$ withdraws 1 dollar from $B_{llx}(xp, p)$, where $p$ is the parent of $s$ in $C$. Let $C''$ be the configuration in which $p$ is frozen for $S$. 
	
	Notice that if $x$ is a cross node for $S$, then by inspection of each chromatic tree transformation, the parent, $p$, of $x$ is frozen for $S$ before $x$. By Lemma~\ref{L_property}.6, $C''$ occurs between the start of $A'$ and $C$. By Lemma~\ref{L_property}.3, $x \in targets(xp,C) - \{targetRoot(xp,C)\}$, and thus by Lemma~\ref{L_property}.2, $x \in targets(xp, C'') - \{targetRoot(up,C)\}$. This implies $p \in targets(up, C'')$. Therefore, by T1-B, 1 dollar was transferred from $B(xp)$ to $B_{llx}(xp, p)$ in configuration $C''$. So $B_{llx}(xp, p, C'') > 0$. 
	
	We next argue that no step $s'$ withdraws from $B_{llx}(xp, p)$ between $C'$ and $C^-$, and so $B_{llx}(xp, p, C) \geq 0$. Consider an unsuccessful LLX $I' \neq I$ that blames an SCX $S'$. The failure step $s'$ of $I'$ withdraws from $B_{llx}(xp, p)$ if $I'$ is performed on $p$ and $p$ is a downwards node for $S'$, or if $I'$ is performed on a node $x$ and $x$ is a cross node for $S'$. Note that since $C''$ is a step contained in $A'$, and there is a step of $I'$ between $C''$ and $C^-$, it follows that $I'$ starts after the start of $A'$. Since $x$ is in the chromatic tree and no successful update CAS occurs between the start of $A'$ and $C^-$, the node on which $I'$ is performed is in the chromatic tree between the start of $I'$ and the end of $I$
	
	Suppose $s'$ is on the node $p$ and $p$ is a downwards node for $S'$. Then since $p$ is frozen for $S'$ in all configurations between $C''$ and $C$, $p.\mathit{info}$ only points to $S$ during $I'$, and so $I'$ blames $S$. Thus, by Lemma~\ref{llx_blame_3}, $I$ does not blame $S$, a contradiction.
	
	Suppose $s'$ is on the node $x$ and $x$ is a cross node for $S'$. Note that a cross node is only frozen for some SCX $S'$ if its parent node is already frozen for $S'$. Since $p$ is frozen for $S$ in all configurations between $C''$ and $C$, $x$ is a cross node for $S$. Thus, $I'$ blames $S' = S$. By Lemma~\ref{llx_no_ucas_abort}, an abort step is performed for $S$ before the end of $I'$. By Lemma~\ref{llx_abort_blame_once}, $I$ does not blame $S$, a contradiction.
\end{itemize}
\end{proof}

\subsubsection{The $B_{scx}$ Account}\label{section_scx}
In this section, we prove that the $B_{scx}$ accounts have a non-negative balance. The $B_{scx}$ accounts are responsible for paying for the failed attempts caused by failed SCXs whenever the auxiliary accounts do not pay for the failures. No rules deposit dollars directly into $B_{scx}$. Instead, dollars are transferred into $B_{scx}$ from $B$ by T1-B.

Recall that when $x$ is a cross node for an SCX $S$, then it is not the first node frozen for $S$. In this scenario, the following lemma applies.

\begin{lemma}\label{freeze_cross}\normalfont Suppose an unsuccessful invocation $I$ of SCX blames some SCX $S$ for $x$, where $x$ is a cross node for $S$. Then the only freezing iteration belonging to $I$ is on $x$.
\end{lemma}

\begin{proof} By inspection, except for the first node frozen for $I$, a node is only frozen for $I$ if its parent is already frozen for $I$. Additionally, the first node frozen for an SCX is not a cross node. Therefore, when $x$ is frozen for $S$, the parent $p$ of $x$ is already frozen for $S$. 

Suppose, for contradiction, that $I$ performs a freezing iteration on some node before its freezing iteration on $x$. Then $I$ will perform a freezing iteration on $p$. Note that the LLX$(p)$ linked to $I$ occurs before $p$ is frozen for $S$, since this LLX is successful. Therefore, if $I$ performs a freezing iteration on $p$, the freezing iteration will be unsuccessful, and so $I$ blames $S$ for $p$ instead of for $x$. Thus, the first freezing iteration belonging to $I$ is performed on $x$. Since this freezing iteration is unsuccessful, no other freezing iterations belong to $I$.
\end{proof}

We argue that the $B_{scx}$ accounts can always pay for a failed freezing iteration when required by W-SCX.

\begin{lemma}\label{SCX1}\normalfont
	For all update phases and cleanup phases $xp$, all nodes $x$ in the chromatic at some point in an execution, and all configurations $C$, $B_{scx}(xp, x, C) \geq 0$.
\end{lemma}

\begin{proof}
A dollar is only deducted from a $B_{scx}$ owned by $xp$ according to rule W-SCX after the failure step of an SCX. Let $s$ be the failure step of an unsuccessful SCX $I$ from the configuration $C^-$ immediately before $C$. Suppose that in $xp$'s attempt $A$ containing $I$, the $i$th freezing iteration fails. Let $x = \mathit{freezeNode}_i(xp,A)$ and let $S$ be the SCX blamed by $I$. By definition, $s$ is the step in which $x$ was frozen for $S$.

We argue that the bank account in which $xp$ withdraws 1 dollar by rule W-SCX is non-empty. By rule W-SCX, $F_i(xp,C^-) = 0$ whenever a dollar is withdrawn from a $B_{scx}$ account as a result of $s$.

Suppose $x$ is a downwards node for $S$, so $s$ withdraws 1 dollar from $B_{scx}(xp, x)$. By Lemma~\ref{F_property}.3, $x \in targets(xp,C) - \{targetRoot(xp,C)\}$, and thus by Lemma~\ref{F_property}.2, $x \in targets(xp, C^-) - \{targetRoot(xp,C^-)\}$. Therefore, by T1-B, $s$ transfers 1 dollar from $B(xp)$ to $B_{scx}(xp, x)$. So $B_{scx}(xp, x, C) \geq 0$.

So suppose $x$ is a cross node for $S$, so $s$ withdraws 1 dollar from $B_{scx}(xp, p)$, where $p$ is the parent of $s$ in $C$. Let $C''$ be the configuration immediately after $p$ is frozen for $S$.  By Lemma~\ref{F_property}.4, there exists an attempt $A'$ by $xp$ before $A$ in which $xp$ fails a freezing iteration on $x$. By Lemma~\ref{F_property}.5, $C''$ occurs between the start of $A'$ and $C$.  By Lemma~\ref{F_property}.3, $x \in targets(xp,C) - \{targetRoot(xp,C)\}$, and thus by Lemma~\ref{F_property}.2, $x \in targets(xp, C'') - \{targetRoot(xp,C'')\}$. This implies $p \in targets(xp,C'')$. Therefore, by T1-B, 1 dollar was transferred from $B(xp)$ to $B_{scx}(xp, p)$ in configuration $C''$. So $B_{scx}(xp, p, C'') > 0$. 

We argue that no step $s'$ withdraws from $B_{scx}(xp, p)$ between $C''$ and $C^-$, and so $B_{scx}(xp, p, C) \geq 0$. Suppose $xp$ performs a failure step $s'$ for some SCX $I'$ between $C''$ and $C$ that blames some SCX $S'$. By W-SCX, $s'$ withdraws from $B_{scx}(xp, p)$ if $s'$ is a freezing CAS on $p$ and $p$ is a downwards node for $S'$, or if $s'$ is a freezing CAS on $x$ and $x$ is a cross node for $S'$.

Suppose $s'$ is a freezing CAS on $p$ and $p$ is a downwards node for $S'$. Thus, by Definition~\ref{scx_blame}, $xp$ fails its freezing iteration on $p$ belonging to $I'$. By Lemma~\ref{F_property}.1 and Lemma~\ref{F_property}.2, no successful update CAS occurs between the start of $A'$ and $C^-$, and the targets of $xp$ do not change. Since the structure of the chromatic tree does not change, $xp$ performs freezing iterations on the same nodes in $I'$ and $I$. This implies that $xp$ performs a successful freezing iteration on $p$ belonging to $I$ before its unsuccessful freezing iteration on $x$ belonging to $I$. This contradicts Lemma~\ref{freeze_cross}.

Suppose $s'$ is a freezing CAS on $x$ and $x$ is a cross node for $S'$. Thus, by Definition~\ref{scx_blame}, $xp$ fails its freezing iteration on $x$ belonging to $I'$. Note that a cross node $x$ is only frozen for $S'$ if its parent node is already frozen for $S'$. Since $p$ is frozen for $S$ in all configurations between $C''$ and $C^-$, $S' = S$. Thus, $I'$ blames $S$ for $x$. This contradicts Observation~\ref{llx_obs}.\ref{llx_obs:scx_blame_limit}. 
\end{proof}

\subsubsection{The $B_{\mathit{nil}}$ Account}\label{section_nil}
The $B_{\mathit{nil}}(cp)$ accounts are responsible for paying for the failed attempts of a cleanup phase $cp$ due to unsuccessful \textsc{Nil} checks in \textsc{TryRebalance}. The \textsc{Nil} checks are used by the \textsc{TryRebalance} routine to guarantee \textsc{Nil} pointers are not traversed. It was shown that unsuccessful \textsc{Nil} checks only occur when the structure of the chromatic tree changes sometime after \textsc{TryRebalance} is invoked \cite{BrownThesis17}.

\begin{lemma}\label{nil_check}\normalfont
Each invocation $I$ of \textsc{TryRebalance} that fails a \textsc{Nil} check is concurrent with a operation that performs a successful SCX during $I$.
\end{lemma}

Intuitively, we can charge unsuccessful \textsc{Nil} checks to the update CAS that caused this change. We argue a dollar can always be withdrawn from $B_{\mathit{nil}}$ to pay for a failed \textsc{Nil} check when required by W-CONFLICT.

\begin{lemma}\label{CONFLICT1}\normalfont
	For all cleanup phases $cp$ and all configurations $C$, $B_{\mathit{nil}}(cp, C) \geq 0$. 
\end{lemma}

\begin{proof} 
Suppose $s$ is a step that applies W-NIL for a cleanup phase $cp$. We argue that in the configuration $C^-$ before $s$, $B_{\mathit{nil}}(cp, C^-) \geq 1$. By Lemma~\ref{nil_check}, a failed \textsc{Nil} check can only occur due to a concurrent update CAS sometime during $cp$'s instance of \textsc{TryRebalance}. Therefore, there exists at least 1 update CAS in the execution interval of the current attempt of $cp$ that applied D1-NIL, depositing 1 dollar into $B_{\mathit{nil}}(cp)$. Since $cp$ only withdraws 1 dollar in each failed attempt, $B_{\mathit{nil}}(cp, C^-) \geq 1$. 
\end{proof}

\subsubsection{The $B$ Account}\label{section_b_account}
In this section, we show that, for all update and cleanup phases $xp$, $B(xp)$ has a non-negative balance. For an update phase $up$, dollars are deposited into the bank account $B(up)$ by steps made by $up$ according to the rules D1-BUP, D2-BUP, and D3-BUP. For a cleanup phase $cp$, dollars are deposited into $B(cp)$ by steps made by $cp$ according to the rules D1-BCP, D2-BCP, and D3-BCP. Only rule T1-B decreases the number of dollars stored in $B(up)$ or $B(cp)$. 
\begin{itemize}
	\item T1-B: Consider a \textbf{successful freezing CAS} performed by any process on a downward node $x$ for some SCX in configuration $C$. For every update or cleanup phase $xp$ where $x \in targets(xp, C)$, transfer 1 dollar from $B(xp)$ to $B_{llx}(xp, x)$, and 1 dollar from $B(xp)$ to $B_{scx}(xp, x)$.
\end{itemize}
Since an update or cleanup phase can perform an arbitrary number of successful freezing CASs, an arbitrary number of dollars may be transferred from the $B$ accounts by T1-B. 

We show $B(xp,C) \geq 0$ for all update and cleanup phases $xp$ and all configurations $C$. To do so, we instead give a function $J(xp,C)$ such that $B(xp,C) \geq J(xp,C) \geq 0$. Recall that for any execution $\alpha$ and function $f(x,C)$, we let $f(x)$ denote $f(x,C)$ as $C$ varies over the configurations of $\alpha$. The following definitions are used to define $J$. For any node $x$ in the chromatic tree in configuration $C$, let
\begin{itemize}
	\item $isFrozen(x, C)$\index{$isFrozen(x, C)$} be 1 if $x$ is frozen in $C$, and 0 otherwise,
	\item $sib(x, C)$\index{$sib(x, C)$} to be sibling node of $x$ in $C$,
	\item $par(x, C)$\index{$par(x, C)$} to be parent node of $x$ in $C$ if $x \neq entry$, and
	\item $npa(x, C)$\index{$npa(x, C)$} be the set of nodes that are not proper ancestors of $x$ in the chromatic tree in $C$.
\end{itemize}
\noindent For any node or phantom node $x$ in the extended tree $T^*$ in configuration $C$, let 
\begin{itemize}
	\item $depth(x, C)$\index{$depth(x, C)$}  be the number of edges along the path from $phantomRoot$ to $x$ in $C$.
\end{itemize}

For each node $x$ that is in the chromatic tree at some point in an execution, we define an auxiliary Boolean variable $abort(x)$\index{$abort(x,C)$} that is initially 0. The $abort$ variables are updated by the following four rules:
\begin{itemize}
	\item A1: A successful freezing CAS on a downwards node $x$ sets $abort(par(x)) = 1$. 
	\item A2: A successful freezing CAS on a cross node $x$ sets $abort(x) = 1$. 
	\item A3: A successful abort step that unfreezes node $x$ sets $abort(x) = 0$.
	\item A4: The completion of an update or cleanup attempt sets $abort(x) = 0$ for all nodes $x$ except for those on which LLX$(x)$ has been performed in the latest attempt of an active operation.
\end{itemize}

We begin with a technical lemma.

\begin{lemma}\label{abort_claim}\normalfont
Let $r$ be a node that is unfrozen by an abort step belonging to an instance of $S$ SCX. Let $C^-$ be the configuration immediately before the abort step and let $C_r$ the configuration immediately after the successful LLX on $r$ linked to $S$. If there exists a configuration $C'$ between $C_r$ and $C^-$ such that $abort(r, C') = 1$, then $abort(r, C^-) = 1$.
\end{lemma}

\begin{proof}
Suppose for contradiction, that $abort(r, C^-) = 0$. Then there exists some step $s$ between $C'$ and $C^-$ that either applies A3 or A4, changing $abort(r)$ to 0. We consider each of these cases in turn.
\begin{itemize}
	\item Case 1: Suppose step $s$ applies A3. Then $s$ is a successful abort step that unfreezes $r$. There is only 1 successful abort step belonging to $S$, which occurs in the step after $C^-$, so $s$ is not an abort step belonging to $S$. So $s$ unfreezes $r$ after the successful freezing CAS $\mathit{fcas}'$ on $r$ belonging to some instance $S' \neq S$ of SCX. Note that since $r$ is unfrozen by an abort step belonging to $S$, there exists a successful freezing CAS $\mathit{fcas}$ on $r$ belonging to $S$ in some configuration between $C_r$ and $C^-$. Since $r$ is frozen for $S$ in all configurations after $\mathit{fcas}$ but before $C^-$, $\mathit{fcas}'$ occurs before $\mathit{fcas}$. After $\mathit{fcas}'$, $r.\mathit{info}$ no longer points to the same SCX-record read by the LLX on $r$ linked to $S$. This implies that any freezing CASs on $r$ belonging to $S$ will fail. This contradicts the fact that $\mathit{fcas}$ is a successful freezing CAS. 

	\item Case 2: Suppose step $s$ applies A4. Let $A$ be the current attempt of the update or cleanup phase $xp$ that invoked $S$. Rule A4 requires that no LLX on $r$ has been performed for any active attempt of an update or cleanup phase. This contradicts the fact that a successful LLX on $r$ has been performed by $xp$ in $A$ before $s$ occurs.
\end{itemize}
Since no step sets $abort(r)$ to 0 between $C'$ and $C^-$, $abort(r, C^-) = 1$.
\end{proof}

\noindent We next show that unless $abort(r,C) = 1$, the step immediately following $C$ is not a successful abort step that unfreezes node $r$.

\begin{lemma}\label{abort_1}\normalfont
Consider the set $R$ of nodes that are unfrozen by an abort step belonging to an instance $S$ of SCX. In the configuration $C^-$ prior to the abort step, $abort(r, C^-) = 1$ for each $r \in R$.
\end{lemma}

\begin{proof} If $R = \emptyset$, then the lemma holds trivially. So suppose $R \neq \emptyset$. For each node in $r \in R$, there exists a successful freezing CAS on $r$ belonging to $S$ which occurs after the LLX on $r$ linked to $S$, but before the abort step belonging to $S$. Furthermore, no field of $r$ is modified between the LLX on $r$ linked to $S$ and the abort step belonging $S$ since $r$ is successfully frozen for $S$. Since each chromatic tree transformations involves a connected subtree of the chromatic tree, the nodes in $R$ are a connected subtree in all configurations between the first LLX linked to $S$ and $C^-$.

First we argue $abort(k, C^-)  = 1$ for any cross node $k$ in $R$.  By inspection of the chromatic tree transformations, there is at most 1 cross node. When $k$ is successfully frozen for $S$, $abort(k)$ is set to 1 by rule A2. By Lemma~\ref{abort_claim}, $abort(k, C^-) = 1$.

Next we argue that for each downward node $d \in R$, $abort(d, C^-)  = 1$. Let the downward nodes in $R$ enumerated in the order that they are frozen be $d_1, \dots, d_t$. By inspection of the chromatic tree transformations, the first node frozen for each transformation is a downward node, so $t \geq 1$ since $R \neq \emptyset$. Let the successful freezing CASs on $d_1, \dots, d_t$ be $c_1, \dots, c_t$ respectively.

First we argue $abort(d_i, C^-)  = 1$ for $1 \leq i < t$. To do so, we show that $c_{i+1}$ sets $abort(d_i) = 1$. Since $d_i$ is a downward node, the next node frozen for $S$ is a child of $d_i$. If the next node frozen for $S$ after $d_i$ is $d_{i+1}$, then $c_{i+1}$ sets $abort(d_i) = 1$ by rule A1. If the next node frozen for $S$ after $d_i$ is the cross node $k$, then the node frozen for $S$ after $k$ is the sibling of $k$. There is only one cross node for $S$, so the sibling of $k$ is the downward node following $d_i$, namely $d_{i+1}$. Thus $d_{i}$ is the parent of $d_{i+1}$, so $c_{i+1}$ sets $abort(d_i) = 1$ by rule A1. In either case, by Lemma~\ref{abort_claim}, $abort(d_i, C^-) = 1$.

Finally we argue $abort(d_t, C^-) = 1$. Suppose the failed freezing iteration belonging to $S$ be for a node $x$. If $d_t$ is the last node frozen for $S$, then $x$ is a child of $d_t$ by the definition of downward nodes since $x$ is the next node frozen after $d_t$. If $d_t$ is not the last node frozen for $S$, then the last node frozen for $S$ is the cross node $k$. Then $d_t$ is the parent of $k$, and $k$ is the sibling of $x$, so $d_t$ is also a parent of $x$.

The freezing CAS on $x$ belonging to $S$ fails, so there exists a freezing CAS $c_x'$ on $x$ for an SCX $S' \neq S$ that occurs sometime after the LLX on $x$ linked to $S$. This freezing CAS sets $abort(d_t) = 1$ by rule A1. By Lemma~\ref{abort_claim}, $abort(d_t, C^-) = 1$. Therefore, $abort(r, C^-) = 1$ for each $r \in R$.
\end{proof}

The definition of $J(cp,C)$ uses another function $H(x,C)$, which maps node $x$ in the chromatic tree in configuration $C$ to an integer. For any node $x$ in the chromatic tree in configuration $C$, let\index{$H(x,C)$}
\begin{equation*}
H(x,C) = 3h(C) + 12 - 3depth(x, C) + 6\dot{c}(C) + \sum_{u \in \textit{npa}(x, C)} [abort(u,C) - isFrozen(u, C)].
\end{equation*}
The terms $6\dot{c}(C)$ and $3h(C)$ terms are only included to ensure that $H$ remains positive. 

\begin{lemma}\label{H_max}\normalfont
	For any node $x$ in the chromatic tree in configuration $C$, $0 \leq H(x, C) \leq 3h(C) + 12\dot{c}(C) + 12$.
\end{lemma}

\begin{proof} 
In configuration $C$, there are at most $\dot{c}(C)$ active operations. Each of these may have up to 6 distinct nodes frozen (for example, if all operations are performing a W3 or W4 transformation), so the summation term in $H$ has a minimum value of $-6\dot{c}(C)$. Likewise, at most $6\dot{c}(C)$ nodes $u$ have $abort(u) = 1$, so the summation term in $H$ has a maximum value of $6\dot{c}(C)$. Finally, $0 \leq depth(u,C) \leq h(C) + 4$ for all nodes $u$, since the height of an extended tree $T^*$ is greater than $T$ by 4. Taking the minimum and maximum values of the summation term in $H$ and $depth(u, C)$ gives $0 \leq H(u, C) \leq 3h(C) + 12\dot{c}(C) + 12$.
\end{proof}

Next we prove additional properties of $H$ required for our analysis.  For any node $x$ in configuration $C$, let\index{$\gamma(x,C)$}
\begin{equation*}
\gamma(x,C) = abort(x,C) - isFrozen(x, C).
\end{equation*}

\begin{lemma}\label{gamma_bounds}\normalfont
For any node $x$ and configuration $C$, $-1 \leq \gamma(x, C) \leq 1$.
\end{lemma}

\begin{proof} 
For any node $x$, since $abort(x, C)  \in \{0, 1\}$ and $isFrozen(x, C) \in \{0, 1\}$. It follows that $-1 \leq \gamma(u, C) \leq 1$.
\end{proof}

\begin{lemma}\label{H_par}\normalfont
	For any node $x$ and its parent $p$ in configuration $C$,
	\begin{equation*}
	2 \leq H(p, C) -  H(x, C) \leq 4.
	\end{equation*}
\end{lemma}

\begin{proof} Since $npa(p, C) = npa(x, C) \cup \{p\}$ and $depth(x,C) - depth(p,C) = 1$, it follows from the definition of $H$ that
\begin{equation*}
\begin{aligned}
H(p,C) &= H(x,C) + \gamma(p, C) + 3(depth(x, C) - depth(p,C)) \\
&= H(x,C) + \gamma(p, C) + 3.
\end{aligned}
\end{equation*}
By Lemma~\ref{gamma_bounds}, $2 \leq H(p, C) -  H(x, C) \leq 4$.
\end{proof}

\begin{corollary}\label{H_par_coro}\normalfont
	For any two nodes $x$ and $y$ in the chromatic tree in a configuration $C$ where $y$ is an ancestor of $x$ and $depth(x,C) - depth(y,C) = k \geq 0$,
	\begin{equation*}
	2k \leq H(y, C) -  H(x, C) \leq 4k.
	\end{equation*}
\end{corollary}

\begin{lemma}\label{H_sib}\normalfont
	For any node $x$ and its sibling $s$ in the chromatic tree in configuration $C$, $H(x, C) = H(s, C)$.
\end{lemma}

\begin{proof}
Since $npa(x, C) = npa(s, C)$ and $depth(x,C) = depth(s,C)$, the lemma follows by definition of $H$.
\end{proof}

For any function $f(x)$ with value $f(x,C)$ in configuration $C$, we define $\Delta f(x,C) = f(x,C) - f(x,C^-)$, where $C^-$ is the configuration immediately before $C$.

\begin{lemma}\label{H_down_freeze}\normalfont
	A successful freezing CAS on a downward node $x$ for an SCX decreases $H(x)$ by at least 1, and does not increase $H(u)$ at any other node $u \neq x$ in the chromatic tree.
\end{lemma}

\begin{proof} Let $C$ be the configuration immediately after the freezing CAS. The freezing CAS on the downward node $x$ sets $isFrozen(x) = 1$. Additionally, by rule A1 of the abort variable, the freezing CAS sets $abort(p) = 1$, where $p$ is the parent of $x$. Consider an arbitrary node $u$ in the chromatic tree. We consider how the value of $H$ changes for $u$ after the freezing CAS, depending on the location $u$ in the chromatic tree:
	\begin{itemize}
		\item Case 1: Suppose $x \in npa(u,C)$ and $p \notin npa(u,C)$. So $u$ is a proper descendant of $p$ and not a proper descendant of $x$. Since $\Delta isFrozen(x,C) = 1$, it follows that $\Delta H(u,C) = -1$. When $u = x$, this proves $H(x)$ decreases by 1.
		\item Case 2: Suppose $x, p \notin npa(u,C)$. So $u$ is a proper descendant of $p$ and a proper descendant of $x$. Therefore $\Delta H(u,C) = 0$ since no terms in $H(u)$ change due to the freezing CAS.
		\item Case 3: Suppose $x, p \in npa(u,C)$. So $u$ is not a proper descendant of $p$. Since $\Delta isFrozen(x,C) = 1$, and $\Delta abort(p,C) \leq 1$, it follows that $\Delta H(u,C) \leq 0$.
	\end{itemize}
	The case where $x \notin npa(u,C)$ and $p \in npa(u,C)$ cannot occur. In all possible cases, $\Delta H(u,C) \leq 0$, and so $H(u)$ does not increases for any node $u$ in the chromatic tree.
\end{proof}

\begin{lemma}\label{H_cross_freeze}\normalfont
	A successful freezing CAS on a cross node $x$ for an SCX does not increase $H(u)$ at any node $u$ in the chromatic tree.
\end{lemma}

\begin{proof} Let $C$ be the configuration immediately after the freezing CAS. The freezing CAS on the cross node $x$ sets $isFrozen(x) = 1$. By rule A2 of the abort variable, the freezing CAS sets $abort(x) = 1$. Consider any node $u$ in the chromatic tree. If $x \in npa(u,C)$, then $\Delta isFrozen(x,C) = 1$ and $\Delta abort(s,C) \leq 1$, it follows that $\Delta H(u,C) \leq 0$. If $x \notin npa(u,C)$, then $\Delta H(u,C) = 0$ since no terms in $H(u)$ change due to the freezing CAS.
\end{proof}

\begin{lemma}\label{H_abort}\normalfont
An abort step for an invocation $S$ of SCX does not change $H(u)$ for any node $u$ in the chromatic tree. 
\end{lemma}

\begin{proof} The abort step unfreezes all nodes that were previously frozen for $S$. Let the set of nodes that are unfrozen by the abort step be $R$. The abort step then changes $isFrozen(r)$ from 1 to 0 for each $r \in R$ since $r$ is no longer frozen after the abort step. Additionally, by Rule A3 of $abort$, we set $abort(r) = 0$ when $r$ is unfrozen. By Lemma~\ref{abort_1}, $abort(r, C^-) = 1$ for each $r \in R$ in the configuration $C^-$ immediately before the abort step. Therefore, every increase in $H$ due to a change in $isFrozen(r)$ from -1 to 0 is balanced by a decrease in $abort(r)$ from 1 to 0. So $H(u)$ does not change for any node $u$ in the chromatic tree after the abort step.
\end{proof}


Next, we give the final definitions needed to define $J(xp,C)$. 

\begin{definition}\label{focalPath_def}\normalfont
For an update or cleanup phase $xp$ and configuration $C$, let $\mathit{focalPath}(xp, C)$\index{$\mathit{focalPath}(xp, C)$} be the set of all nodes along the path from $entry$ to $\mathit{focalNode}(xp, C)$, inclusive. 
\end{definition}

\noindent The nodes in this set have the following useful property.

\begin{lemma}\label{focalPath}\normalfont
Consider a configuration $C$ during a cleanup phase $cp$. For every node $x$ in $\mathit{focalPath}(cp,C)$, $x$ is on $cp$'s search path in $C$.
\end{lemma}

\begin{proof}
By Lemma~\ref{focalNode_on_sp}, $\mathit{focalNode}(cp,C)$ is on $cp$'s search path in $C$. The unique path of nodes from $entry$ to $\mathit{focalNode}(cp,C)$ are also on $cp$'s search path in $C$, which by definition are the nodes in $\mathit{focalPath}(cp,C)$.
\end{proof}

\begin{definition}\label{backup}\normalfont
For every cleanup phase $cp$ and every node $x$ in the chromatic tree at the beginning of $cp$, we define $backup(cp,x) = 0$. When an update CAS adds a new node $x$ into the chromatic tree during $cp$, $backup(cp,x)$ is initialized to 1 if $x$ is on $\mathit{focalPath}(cp)$, and 0 otherwise. Finally, a step that removes a node $x$ from $\mathit{focalPath}(cp)$ sets $backup(cp,x) = 0$.
\end{definition}
\noindent For convenience, we let $backup(up,x,C) = 0$ for all update phases $up$, nodes $x$, and configurations $C$.

\begin{lemma}\label{backup_zero}\normalfont
If there exists a configuration $C$ during a cleanup $cp$ where a node $x$ is a proper descendant of $\mathit{focalNode}(cp,C)$, then $backup(cp,x,C') = 0$ for all configurations $C'$ after $C$. 
\end{lemma}

\begin{proof} 
Let $C_x$ be the first configuration in which $x$ is in the chromatic tree. If $backup(cp,x,C_x) = 0$, then $backup(cp,x,C') = 0$ since $backup(cp,x)$ is only initialized to 1 when $x$ is added to the chromatic tree. So suppose $backup(cp,x,C_x) = 1$, and so $x \in \mathit{focalPath}(cp,C_x)$. Since $x$ is a proper descendant of $\mathit{focalNode}(cp,C)$, $x \notin \mathit{focalNode}(cp,C)$. Thus, there exists a step between $C_x$ and $C$ that removes $x$ from $\mathit{focalPath}(cp)$. This step sets  $backup(cp,x)$ to 0, and so $backup(cp,x,C') = 0$.
\end{proof}

Finally, for an update or cleanup phase $xp$ and configuration $C$, we define the function\index{$J(xp, C)$}
\begin{equation*}
J(xp, C) = \sum_{u \in targets(xp, C)} 2H(u, C) + \sum_{v \in \mathit{focalPath}(xp,C)} 576 \cdot backup(xp,v,C).
\end{equation*}
If $xp$ is not active in $C$, define $J(xp, C) = 0$. Note that the second summation in $J$ is only non-zero when $xp$ is a cleanup phase.

\begin{observation}\label{obs_J_pos}\normalfont
For every update or cleanup phase $xp$ and every configuration $C$ in an execution, $J(xp,C) \geq 0$.
\end{observation}

\begin{proof} 
By Lemma~\ref{H_max}, $H(x, C) \geq 0$ for all $x$ and $C$. Since it is the summation of $H$ functions and $backup$ Boolean variables, $J(xp,C) \geq 0$.
\end{proof}

We prove that $B(xp,C) \geq J(xp,C)$ for every configuration $C$. The purpose of the first summation in $J$ is to decrease $J$ whenever a successful freezing CAS is performed on a downwards node $x \in targets(xp,C)$. By Lemma~\ref{H_down_freeze}, $H(u)$ decreases by at least 1 and no other $H$ term increases. Thus, $J$ decreases by at least 2 as the result of the freezing CAS on $x$. Since T1-B transfers 1 dollar from $B(xp)$ to each of $B_{llx}(xp,x)$ and  $B_{scx}(xp,x)$ after the successful freezing CAS on $x$, the invariant cannot be violated by the freezing CAS.

The purpose of the second summation in $J$ is illustrated in the following example. Let $cp$ be a cleanup phase by a process $P$. Consider an update CAS from a configuration $C$ that removes $\mathit{focalNode}(cp,C) = location(P,C)$ from the chromatic tree. Let $C'$ be the configuration  immediately following $C$. Since $location(P,C)$ is no longer in the chromatic tree, by Definition~\ref{focalNode}.3, $\mathit{focalNode}(cp,C') = \mathit{nextNode}(P,C')$. Thus, the first summation in $J$ may change since $targets(cp,C')$ is not necessarily equal to $target(cp,C)$. If the depth of $\mathit{nextNode}(P,C')$ in $C'$ is less than the depth of $location(P,C)$ in $C$, then we can show that the first summation in $J$ may increase by an amount proportional to the difference in depth. This scenario occurs when $cp$ must pop many nodes off its stack during backtracking before it visits a node in the chromatic tree. These nodes on $cp$'s stack were removed from the chromatic tree after they were pushed onto the stack. When nodes were removed from the chromatic tree, a proportional number of new nodes were added into $\mathit{focalPath}(cp)$. A new node $v$ added into $\mathit{focalPath}(cp)$ will have $backup(cp,v) = 1$. When $\mathit{focalNode}(cp)$ changes from $location(P,C)$ to $\mathit{nextNode}(P,C')$, many nodes $v$ with $backup(cp,v,C) = 1$ will be removed from $\mathit{focalPath}(cp,C)$. This decrease in the second summation of $J$ multiplied by a sufficiently large constant offsets the increase in the first summation of $J$.

In the following lemmas, we show how various steps change $J$. In particular, we will show that only update CASs and commit steps can increase $J(cp)$, and only increase $J(cp)$ by a constant amount. 

To show how $J$ changes as a result of a cleanup phase $cp$ visiting new nodes, we first consider how a summation of $H$ values may change.

\begin{lemma}\label{J_rebalSet_move}\normalfont
For any node $x$ and its parent $p$ in the chromatic tree in configuration $C$,
\begin{equation*}
 6 \leq \bigg(\sum_{v \in rebalSet(p,C)} H(v,C)\bigg) - \bigg(\sum_{u \in rebalSet(x,C)} H(u,C)\bigg) \leq 72.
\end{equation*}
\end{lemma}

\begin{proof}
Without loss of generality, assume $x$ is the left child of $p$. Figure~\ref{fig_cp_target_down} shows an example of $rebalSet(p,C)$ and $rebalSet(x,C)$. We define a mapping $\psi$ from each node $u \in rebalSet(p,C)$ to a node in $rebalSet(x,C)$. If $u \in rebalSet(p,C) - rebalSet(x,C)$, $\psi$ maps $u$ to a node in $rebalSet(x,C) - rebalSet(p,C)$. If $u \in rebalSet(p,C) \cap rebalSet(x,C)$, $\psi(u) = u$. Therefore, 
\begin{equation*}
\begin{aligned}
\bigg(\sum_{v \in rebalSet(p,C)} H(u,C)\bigg) - \bigg(\sum_{u \in rebalSet(x,C)} H(v,C)\bigg) = \sum_{u \in rebalSet(p,C)} [H(u,C) - H(\psi(u),C)] 
\end{aligned}
\end{equation*}

\begin{figure}[!h]
	\centering
	\begin{subfigure}[b]{0.4\textwidth}\centering
		\begin{tikzpicture}[-, >=stealth', level distance=0.6cm] 
		\node [arn_g, label=right:{$n_1$}]{ } 
		child{node [arn_g, label=right:{$n_2$}] { }
			child{node [arn_g, label=right:{$n_3$}] { }
				child[sd=2.0cm]{node [arn_g, label=right:{$n_4$}] { }   
					child[sd=2.5cm]{ node [arn_g, label=right:{$n_6$}] { } 
						child[sd=1.4cm]{ node [arn_g, label=right:{$n_8$}] { } 
							child[sd=0.7cm]{ node [arn_g, label=below:{$n_{12}$}] { } }
							child[sd=0.7cm]{ node [arn_g, label=below:{$n_{13}$}] { } }
						}
						child[sd=1.4cm]{ node [arn_g, label=right:{$n_9$}] { } 
							child[sd=0.7cm]{ node [arn_g, label=below:{$n_{14}$}] { } }
							child[sd=0.7cm]{ node [arn_g, label=below:{$n_{15}$}] { } }
						}
					}
					child[sd=2.5cm]{ node [arn_g, label=right:{$n_7$}] { $p$  } 
						child[sd=1.2cm]{ node [arn_w, label=right:{$n_{10}$}] { $x$ } }
						child[sd=1.2cm]{ node [arn_w, label=right:{$n_{11}$}] { } 
							child[sd=1.4cm]{ node [arn_w, label=right:{$n_{16}$}] { } 
								child[sd=0.7cm]{ node [arn_w, label=below:{$n_{18}$}] { } }
								child[sd=0.7cm]{ node [arn_w, label=below:{$n_{18}$}] { } }
							}
							child[sd=1.4cm]{ node [arn_w, label=right:{$n_{17}$}] { } 
								child[sd=0.7cm]{ node [arn_w, label=below:{$n_{20}$}] { } }
								child[sd=0.7cm]{ node [arn_w, label=below:{$n_{21}$}] { } }
							}
						}
					}
				}
				child[sd=2.0cm]{node [arn_g, label=right:{$n_5$}] { }}   
			}
		};
		\end{tikzpicture}
		\caption{$rebalSet(p,C)$}
	\end{subfigure}
	\begin{subfigure}[b]{0.4\textwidth}\centering
		\begin{tikzpicture}[-, >=stealth', level distance=0.6cm] 
		\node [arn_w, label=right:{$n_1$}]{ } 
		child{node [arn_g, label=right:{$n_2$}] { }
			child{node [arn_g, label=right:{$n_3$}] { }
				child[sd=2.0cm]{node [arn_g, label=right:{$n_4$}] { }   
					child[sd=2.5cm]{ node [arn_g, label=right:{$n_6$}] { } 
						child[sd=1.4cm]{ node [arn_w, label=right:{$n_8$}] { } 
							child[sd=0.7cm]{ node [arn_w, label=below:{$n_{12}$}] { } }
							child[sd=0.7cm]{ node [arn_w, label=below:{$n_{13}$}] { } }
						}
						child[sd=1.4cm]{ node [arn_w, label=right:{$n_9$}] { } 
							child[sd=0.7cm]{ node [arn_w, label=below:{$n_{14}$}] { } }
							child[sd=0.7cm]{ node [arn_w, label=below:{$n_{15}$}] { } }
						}
					}
					child[sd=2.5cm]{ node [arn_g, label=right:{$n_7$}] { $p$  } 
						child[sd=1.2cm]{ node [arn_g, label=right:{$n_{10}$}] { $x$ } }
						child[sd=1.2cm]{ node [arn_g, label=right:{$n_{11}$}] { } 
							child[sd=1.4cm]{ node [arn_g, label=right:{$n_{16}$}] { } 
								child[sd=0.7cm]{ node [arn_g, label=below:{$n_{18}$}] { } }
								child[sd=0.7cm]{ node [arn_g, label=below:{$n_{18}$}] { } }
							}
							child[sd=1.4cm]{ node [arn_g, label=right:{$n_{17}$}] { } 
								child[sd=0.7cm]{ node [arn_g, label=below:{$n_{20}$}] { } }
								child[sd=0.7cm]{ node [arn_g, label=below:{$n_{21}$}] { } }
							}
						}
					}
				}
				child[sd=2.0cm]{node [arn_w, label=right:{$n_5$}] { } }   
			}
		};
		\end{tikzpicture}
		\caption{$rebalSet(x,C)$}
	\end{subfigure}
	\caption{Nodes in $targets(cp)$ are highlighted in gray.}\label{fig_cp_target_down} 
\end{figure}
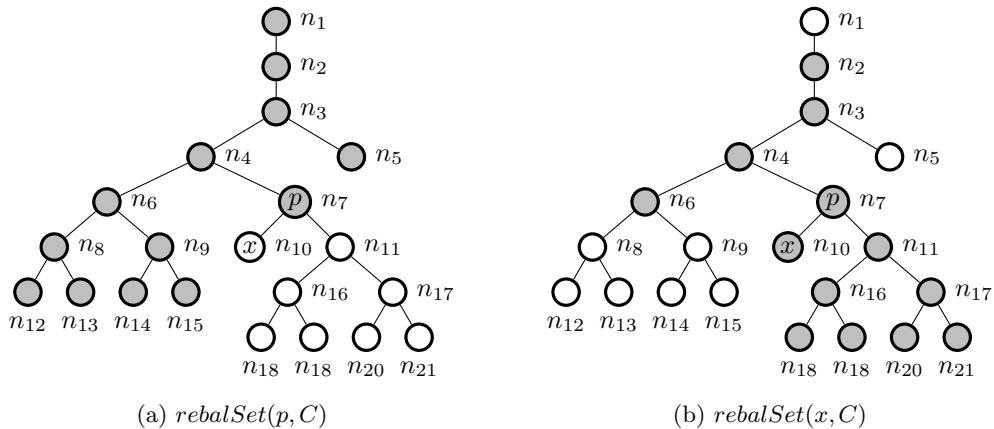

We next give an upper and lower bound for $H(u,C) - H(\psi(u),C)$ for $u \in rebalSet(p,C)$. For example, consider the mapping $\psi(n_{12}) = n_{18}$. Lemma~\ref{H_par} and  Lemma~\ref{H_sib} can be used to give an upper and lower bound for $H(n_{12}, C) - H(n_{18}, C)$ as follows:
\begin{equation*}
\begin{aligned}
H(n_{12}, C) &\leq H(n_6, C) - 4 &&\text{by Lemma~\ref{H_par} applied 2 times} \\
&=  H(n_7, C) - 4 &&\text{by Lemma~\ref{H_sib}} \\
&\leq  H(n_{18}, C) - 4 + 12 &&\text{by Lemma~\ref{H_par} applied 3 times} \\
&=  H(n_{18}, C) + 8
\end{aligned}
\end{equation*}
\noindent and
\begin{equation*}
\begin{aligned}
H(n_{12}, C) &\geq H(n_6, C) - 8 &&\text{by Lemma~\ref{H_par} applied 2 times} \\
&=  H(n_7, C) - 8 &&\text{by Lemma~\ref{H_sib}} \\
&\geq  H(n_{18}, C) - 8 + 6 &&\text{by Lemma~\ref{H_par} applied 3 times} \\
&=  H(n_{18}, C) + 2.
\end{aligned}
\end{equation*}

Figure~\ref{tab_cp_target_up} summarizes the upper and lower bounds on $H(u,C) - H(\psi(u),C)$ for each node $u \in rebalSet(p,C) - rebalSet(x,C)$. For all other $u$ not listed in Figure~\ref{tab_cp_target_up}, $u = \psi(u)$, and so $H(u,C) - H(\psi(u),C) = 0$.
\begin{figure}[h!]
\begin{center}
	\begin{tabular}{ | p{4cm} | p{4cm} | p{4cm} | p{4cm} |}
		\hline
		 \multicolumn{2}{|c|}{}  & \multicolumn{2}{|c|}{$H(u,C) - H(\psi(u),C)$} \\ \hline
		\multicolumn{1}{|c|}{$u$} & \multicolumn{1}{|c|}{$\psi(u)$} & \multicolumn{1}{|c|}{Upper bound} & \multicolumn{1}{|c|}{Lower bound} \\ \hline
		$n_{1}$ & $n_{10}$ & $\leq 20$ & $\geq 10$ \\ \hline
		$n_5$ & $n_{11}$ & $\leq 8$ & $\geq 4$ \\ \hline
		$n_8$ and $n_9$  & $n_{16}$ and $n_{17}$ resp. & $\leq 6$ & $\geq 0$ \\ \hline
		$n_{12}$,  $n_{13}$, $n_{14}$ and $n_{15}$ & $n_{18}$, $n_{19}$, $n_{20}$ and $n_{21}$ resp. & $\leq 8$ & $\geq -2$ \\ \hline
	\end{tabular}
\end{center}
\caption{Upper and lower bounds on $H(u,C) - H(\psi(u),C)$ for each node $u \in rebalSet(p,C) - rebalSet(x,C)$. }\label{tab_cp_target_up}
\end{figure}

Therefore,
\begin{equation*}
\sum_{u \in rebalSet(p,C)} [H(u,C) - H(\psi(u),C)] \leq 20 + 8 + 2 \cdot 6 + 4 \cdot 8 = 72
\end{equation*}
and.
\begin{equation*}
\sum_{u \in rebalSet(p,C)} [H(u,C) - H(\psi(u),C)] \geq 10 + 4 + 2 \cdot 0 + 4 \cdot (-2) = 6.
\end{equation*}
Thus,
\begin{equation*}
6 \leq \bigg(\sum_{v \in rebalSet(p,C)} H(v,C)\bigg) - \bigg(\sum_{u \in rebalSet(x,C)} H(u,C)\bigg) \leq 72.
\end{equation*}
\end{proof}

Steps that update a process $P$'s local variable $l$ in \textsc{BacktrackingCleanup} may change $J$ for a cleanup phase $cp$ due to a change in its focal node, and thus a change in $targets(cp)$. The following lemmas show that these steps do not increase $J$.

\begin{lemma}\label{J_descendant}\normalfont
	Consider a process $P$ in cleanup phase $cp$ in configuration $C$. Let $C^-$ be the configuration immediately before $C$. If the structure of the chromatic tree is the same in $C^-$ and $C$, and $\mathit{focalNode}(cp,C)$ is a (not necessarily proper) descendant of $\mathit{focalNode}(cp,C^-)$, then $\Delta J(cp, C) \leq 0$.
\end{lemma}

\begin{proof}
Let the configuration immediately before $C$ be $C^-$.  We first consider how the first summation in $J$ may change. Since the chromatic tree is the same in $C^-$ and $C$, it follows by definition that $H(u,C) = H(u,C^-)$ and $rebalSet(u,C) = rebalSet(u,C^-)$ for all nodes $u$ in the chromatic tree. Let $k = depth(\mathit{focalNode}(cp,C),C) - depth(\mathit{focalNode}(cp,C^-),C)$. Since $\mathit{focalNode}(cp,C)$ is a descendant of $\mathit{focalNode}(cp,C^-)$, $k \geq 0$. Let $n_1,\dots,n_{k+1}$ be the nodes on the path from $n_1 = \mathit{focalNode}(cp,C^-)$ to $n_{k+1} = \mathit{focalNode}(cp,C)$. By Lemma~\ref{J_rebalSet_move}, for $1 \leq i \leq k$,
\begin{equation*}
\begin{aligned}
\bigg(\sum_{u \in rebalSet(n_{i+1},C)} H(u,C)\bigg) - \bigg(\sum_{v \in  rebalSet(n_i,C)} H(v,C^-)\bigg) &\leq -6. \\
\end{aligned}
\end{equation*}

\noindent It follows that
\begin{equation*}
\begin{aligned}
\bigg(\sum_{u \in rebalSet(n_{k+1},C)} H(u,C)\bigg) - \bigg(\sum_{v \in  rebalSet(n_1,C)} H(v,C^-)\bigg) &\leq -6k. \\
\end{aligned}
\end{equation*}

\noindent By definition, $targets(cp,C) = rebalSet(n_{k+1},C)$ and $targets(cp,C^-) = rebalSet(n_1,C^-)$. Since $rebalSet(n_1,C^-) = rebalSet(n_1,C)$, it follows that 
\begin{equation*}
\begin{aligned}
\bigg(\sum_{u \in targets(cp,C)} H(u,C)\bigg) - \bigg(\sum_{v \in  targets(cp,C^-)} H(v,C^-)\bigg) &\leq -6k. \\
\end{aligned}
\end{equation*}

By definition, $\mathit{focalPath}(cp, C) = \mathit{focalPath}(cp, C^-) \cup \{ n_2,\dots,n_{k+1} \}$. For $2 \leq i \leq k+1$, $n_i$ is a descendant of $n_1$ in $C^-$, so by Lemma~\ref{backup_zero}, $backup(cp,n_i,C) = 0$. Therefore, the second summation of $J$ does not change as a result of the newly added nodes in $\mathit{focalPath}(cp, C)$. Therefore, $\Delta J(cp,C) \leq -12k \leq 0$, since $k \geq 0$.
\end{proof}

\begin{lemma}\label{J_forward_traversal}\normalfont
Consider a step $s$ of a process $P$ in cleanup phase $cp$ that updates $P$'s local variable $l$ to point from a node $p$ to its child $x$. If $C$ is the configuration after $s$, then $\Delta J(cp, C) \leq 0$.
\end{lemma}

\begin{proof}
Step $s$ does not change the structure of the chromatic tree. Let the configuration immediately before $C$ be $C^-$. By definition, $location(P,C^-) = p$ and $location(P,C) = x$. We argue that $\mathit{focalNode}(cp,C)$ is a descendant of $\mathit{focalNode}(cp,C^-)$. We consider several cases.
\begin{itemize}
	\item Case 1: Suppose $currentViol(P,C^-) = v \neq \textsc{Nil}$. By Definition~\ref{focalNode}.1, $\mathit{focalNode}(cp,C^-) = v$. Since $s$ is not an update CAS, by Definition~\ref{currentViol}, $currentViol(P,C) = v$. By Definition~\ref{focalNode}.1, $\mathit{focalNode}(cp,C) = v$. Since $\mathit{focalNode}(cp,C) = \mathit{focalNode}(cp,C^-)$, $\mathit{focalNode}(cp,C)$ is a descendant of $\mathit{focalNode}(cp,C^-)$.
	
	\item Case 2: Suppose $currentViol(P,C^-) = \textsc{Nil}$ and $location(P,C^-)$ is in the chromatic tree. By Definition~\ref{focalNode}.2, $\mathit{focalNode}(cp,C^-) = location(P,C^-) = p$. Since $p$ is in the chromatic tree and $x$ is a child of $p$, $x$ is also in the chromatic tree. If $currentViol(P,C) = \textsc{Nil}$, then by Definition~\ref{focalNode}.2, $\mathit{focalNode}(cp,C) = location(P,C) = x$. Otherwise $currentViol(P,C) = x$, then by Definition~\ref{focalNode}.1, $\mathit{focalNode}(cp,C) = x$. 	In either case, $\mathit{focalNode}(cp,C)$ is a descendant of $\mathit{focalNode}(cp,C^-)$.
	
	\item Case 3: Suppose $currentViol(P,C^-) = \textsc{Nil}$, and $location(P,C^-)$ is not in the chromatic tree. By Definition~\ref{focalNode}.3, $\mathit{focalNode}(cp,C^-) = \mathit{nextNode}(P,C^-)$. By definition, $\mathit{nextNode}(P,C^-)$ is the first node in the chromatic tree that $P$ will visit in a solo execution starting from $C^-$. Notice that $P$'s solo execution starting from $C^-$ begins with $s$. There are two cases depending on $location(P,C)$, the node visited by $P$ during $s$.
	
	If $location(P,C)$ is not in the chromatic tree, then $\mathit{nextNode}(P, C) = \mathit{nextNode}(P, C^-)$. So by Definition~\ref{focalNode}.3, $\mathit{focalNode}(cp, C) = \mathit{nextNode}(P, C^-)$. In this case, $\mathit{focalNode}(cp,C) = \mathit{focalNode}(cp,C^-)$.
	
	If $location(P,C)$ is in the chromatic tree, then $\mathit{nextNode}(P,C^-) = location(P,C)$. If $currentViol(P,C) = \textsc{Nil}$, then by Definition~\ref{focalNode}.2, $\mathit{focalNode}(cp,C) = location(P,C)$. Otherwise, $currentViol(P,C) \neq \textsc{Nil} = location(P,C)$. By Definition~\ref{focalNode}.1, $\mathit{focalNode}(cp,C) = v$. In either case, $\mathit{focalNode}(cp,C) = \mathit{focalNode}(cp,C^-)$.
\end{itemize}

Hence, in any case, $\mathit{focalNode}(cp,C)$ is a descendant of $\mathit{focalNode}(cp,C^-)$. Thus, by Lemma~\ref{J_descendant}, $\Delta J(cp, C) \leq 0$.
\end{proof}

We can prove a similar result to Lemma~\ref{J_forward_traversal}, except for the case when $l$ is updated to point to the topmost node on $P$'s stack. The following lemmas are used to prove this case.

\begin{lemma}\label{sp_no_move}
	Let $x$ be any node, not necessarily in the chromatic tree. Let $y$ be a node in the chromatic tree immediately before and after a chromatic tree transformation $T$. If $y$ is on the search path for a key $k$ from $x$ before $T$, then $y$ is on the search path for $k$ from $x$ after $T$.
\end{lemma}

\begin{proof}
Consider the node $u$ whose pointer is updated by $T$. If $u$ is not on the search path for $k$ from $x$, then no pointer along the path from $x$ to $y$ is changed. Therefore, $y$ is on the search path for $k$ from $x$ after $T$.

So suppose $u$ is on the search path for $k$ from $x$. Then by Observation~\ref{trans_obs}.\ref{trans_obs:reach_before}, the nodes reachable from $u$ before $T$ are reachable from $u$ after $T$, except for the nodes in the set $R$ removed by $T$. So only nodes in $R$ are removed from the search path for $k$ from $x$. Since $y$ is not removed by $T$, $y$ is still on the search path for $k$ from $x$ after $T$.
\end{proof}

\begin{corollary}\label{sp_no_move_coro}
Suppose a node $y$ is in the chromatic tree and on the search path for a key $k$ from a node $x$ in a configuration $C'$. If $y$ is still in the chromatic tree in a later configuration $C$, then it is still on the search path from $x$ in $C$.
\end{corollary}

\begin{lemma}\label{visit_sp}\normalfont
Let $C'$ be a configuration during a process $P$'s cleanup phase, and suppose node $x$ is on $P$'s search path from $location(P,C')$. If $x$ is still in the chromatic tree in a later configuration, $C$, and $x$ is not on $P$'s search path from $location(P,C)$, then there exists a step between $C'$ and $C$ in which $x$ is pushed onto $P$'s stack.
\end{lemma}

\begin{proof}
By the properties of SCX, nodes removed from the chromatic tree are never added back to the chromatic tree, so $x$ is not removed from the chromatic tree between $C'$ and $C$. So $x$ is in the chromatic tree in all configurations between $C'$ and $C$. By Lemma~\ref{hindsight}, $x$ is on $P$'s search path in all configurations between $C'$ and $C$. 

Let $C''$ be the last configuration between $C'$ and $C$ in which $x$ is on $P$'s search path from $location(P,C'')$. Let $s$ be the step following $C''$. An update CAS by any process does not change $location(P)$. By Lemma~\ref{sp_no_move}, $x$ remains on the $P$'s search path from $location(P,C'')$ after $s$. Thus, $s$ is not an update CAS and the structure of the chromatic tree is not changed by $s$. So $s$ must change $location(P)$ by an update to $P$'s local variable $l$. The lines of code where this occurs are on lines~\ref{ln:cleanup:first_pop}, \ref{ln:cleanup:pop}, and \ref{ln:cleanup:leaf_update}.

Suppose $s$ updates $l$ to point to a node $t$ popped by line~\ref{ln:cleanup:first_pop} or line~\ref{ln:cleanup:pop} of \textsc{BacktrackingCleanup}. Let $\alpha'$ be the prefix of the execution up to $C''$. By Lemma~\ref{no_cycle_claim}, there is a path in $G_{\alpha'}$ from $t$ to $location(P,C'')$. Since $x$ is on $P$'s search path from $location(P,C'')$, there is a path from $t$ to $x$ in $G_{\alpha'}$. For any node $y$ from any path from $t$ to $x$ in $G_{\alpha'}$, consider an update CAS in $\alpha'$ that modifies a child pointer of $y$. By Observation~\ref{trans_obs}\ref{trans_obs:reach_before}, since $x$ is not removed from the chromatic tree by this update CAS, $x$ is reachable from $y$ after the update CAS. Therefore, $x$ is reachable from $t$ in $C''$. This implies that $x$ is on $P$'s search path from $t$ after $s$, since $s$ does not modify the chromatic tree. Since $P$'s local variable $l$ points to $t$ in the configuration following $s$, this contradicts the fact that $C''$ is the last configuration in which $x$ is on $P$'s search path from $location(P,C'')$.

Therefore, $s$ updates $P$'s local variable $l$ to point to a child node on line~\ref{ln:cleanup:leaf_update}. Since $l$ is updated to the next node on $P$'s search path from $location(P,C'')$, $x$ can only be removed from $P$'s search path if $location(P,C'') = x$. This implies that $x$ was pushed onto $P$'s stack on line~\ref{ln:cleanup:push} immediately prior to when $P$ executes line~\ref{ln:cleanup:leaf_update} during $s$.
\end{proof}

\begin{lemma}\label{visit_on_stack}\normalfont
Let $C'$ be a configuration during a process $P$'s cleanup phase, and suppose node $x$ is on $P$'s search path from $location(P,C')$. If $x$ is still in the chromatic tree in a later configuration, $C$, and $x$ is not on $P$'s search path from $location(P,C)$, then one of the following statements are true:
\begin{enumerate}
\item $x$ is on $P$'s stack in $C$, or
\item during $P$'s cleanup attempt containing $C$, $x$ was popped off $P$'s stack on lines~\ref{ln:cleanup:p_pop} or \ref{ln:cleanup:gp_pop} of \textsc{BacktrackingCleanup}.
\end{enumerate}
\end{lemma}

\begin{proof}
Note that if Statement 2 is true, then $x$ is not on $P$'s stack in $C$ since no pushes are performed in a cleanup attempt after line~\ref{ln:cleanup:gp_pop} is executed. Thus, both statements cannot be true. Suppose, for contradiction, that both statements are false. This occurs when $x$ is not on $P$'s stack in $C$, and $x$ was not popped from the stack on lines~\ref{ln:cleanup:p_pop} and \ref{ln:cleanup:gp_pop} of $P$'s latest cleanup attempt containing $C$.

Let $C''$ be the last configuration between $C'$ and $C$ in which $x$ is on the $P$'s search path from $location(P, C'')$ in $C''$. By Lemma~\ref{visit_sp}, there exists a step between $C''$ and $C$ in which $x$ is pushed onto $P$'s stack. Since $x$ is not on $P$'s stack in $C$, $x$ is popped off $P$'s stack sometime after it is last pushed onto the stack. Let $C_{pop}$ be the configuration immediately before this pop. 

Suppose $x$ was popped after line~\ref{ln:cleanup:first_pop} or line~\ref{ln:cleanup:pop} is executed. Then $P$'s local variable $l$ points to $x$ in $C_{pop}$. This contradicts the fact that $C''$ is the last configuration in which $x$ is on $P$'s search path from its location.

So suppose $x$ is popped on lines~\ref{ln:cleanup:p_pop} and \ref{ln:cleanup:gp_pop}. By assumption, $P$ does not pop $x$ off its stack on lines~\ref{ln:cleanup:p_pop} and \ref{ln:cleanup:gp_pop} during its attempt containing $C$, so $C_{pop}$ belongs to a prior attempt $A$. Let $t$ be the last node popped during backtracking in the attempt $A'$ after $A$. By the definition of backtracking, $t$ is unmarked in the configuration immediately before it is popped in $A'$, and hence in the chromatic tree. By Lemma~\ref{stack_ancestors}, since $x$ is above $t$ on $P$'s stack in $C_{pop}$, $x$ is on $P$'s search path from $t$ in $C_{pop}$. In the configuration $C_{pop}'$ immediately after $t$ is popped in $A'$ (either on line~\ref{ln:cleanup:pop} or \ref{ln:cleanup:first_pop}), $t = location(P,C_{pop}')$. Since $x$ is on $P$'s search path from $t$ in $C_{pop}$, by Lemma~\ref{sp_no_move}, $x$ is on $P$'s search path from $t$ in $C_{pop}'$. This contradicts the fact that $C''$ is the last configuration in which $x$ is on $P$'s search path from $location(P,C'')$.
\end{proof}

\begin{lemma}\label{sp_descendant}\normalfont
	Consider an update CAS $ucas$, and let the configurations immediately before and after $ucas$ be $C^-$ and $C$ respectively. Suppose $x$ is a node in the chromatic tree in $C^-$ and $C$. Let $P$ be a process in a cleanup phase $cp$ and $\mathit{nextNode}(P,C) \neq \textsc{Nil}$. If $x$ is not on a process $P$'s search path in $C^-$ and on $P$'s search path in $C$, then in $C$,
	\begin{itemize}
		\item if $\mathit{location}(P,C)$ is the chromatic tree, $x$ is on $P$'s search path from $\mathit{location}(P,C)$,
		\item otherwise $x$ is on $P$'s search path from $\mathit{nextNode}(P,C)$.
	\end{itemize}  
\end{lemma}

\begin{proof}
Suppose $ucas$ updates the pointer to a node $u$, removing a set of nodes $R$ from the chromatic tree. If $x$ is not a descendant of $u$, then $ucas$ does not move $x$ onto $P$'s search path. So suppose $x$ is a descendant of $u$. Since $x$ is not removed from the chromatic tree, $x \notin R$. Let $k$ be the key of $P$'s current operation. Since $P$'s local variable $l$ points to the same node in $C^-$ and $C$, $location(P,C) = location(P,C^-)$. 

Let the \textit{key range} of a node $x$ be the set of keys in which $x$ is on the search path from $entry$. It can be verified by inspection of the chromatic tree transformations that only the \textsc{Delete} transformation can add a new key to the key range of $x$ after $x$ has been added to the chromatic tree. Without loss of generality, suppose $ucas$ deletes a left leaf (i.e~the transformation shown in Figure~\ref{fig_transformations}), which removes a node \un{x}, its left child \un{xl}, and right child \un{xr}, and adds a new node \nn{}. The key range of \un{xl} before $ucas$ is added to each node on the leftmost path of the subtree rooted at \nn{}. No other key ranges change as a result of $ucas$. So suppose $k$ is in key range of \un{xl}, and $x$ is on the leftmost path of the subtree rooted at \nn{}.

Suppose $location(P,C)$ is in the chromatic tree in $C$. So $location(P,C)$ is in the chromatic tree in $C^-$ since it is not a new node added by $ucas$. By Corollary~\ref{reach_SP_backtracking_cleanup} and Lemma~\ref{hindsight}, $location(P,C)$ is on the search path for $k$ in $C^-$ and $C$. The only nodes in the chromatic tree in $C^-$ with $k$ in their key range and are still in the chromatic tree in $C$ are those on the path from $entry$ to $u$. Therefore, $location(P,C)$ is a node on the path from $entry$ to $u$ in $C$. Therefore, $x$ is on $P$'s search path from $location(P,C)$ in $C$.

Suppose $location(P,C)$ is not in the chromatic tree in $C$, but was in the chromatic tree in $C^-$. So $location(P,C)$ was removed from the chromatic tree by $ucas$. By Corollary~\ref{reach_SP_backtracking_cleanup} and Lemma~\ref{hindsight}, $location(P,C)$ is on $P$'s search path in $C^-$, and so is a node on the path from $u$ to \un{xl} in $C^-$. Since there are no nodes in the chromatic from $location(P,C)$ to the leaf \un{xl} in $C$, by Lemma~\ref{nextNode_obs}, $nextNode(P,C)$ is a node $t$ on $P$'s stack that is in the chromatic tree in $C$. Note that since $P$'s stack is the same in $C^-$ and $C$, $t$ is also on $P$'s stack in $C^-$. By Lemma~\ref{on_stack_SP_2}, the nodes on $P$'s stack in $C^-$ are on $P$'s search path in $C^-$. Thus, $t$ is a node on the path from $entry$ to $u$ in $C$. Therefore, $x$ is on $P$'s search path from $nextNode(P,C)$ in $C$.

Suppose $location(P,C)$ is not in the chromatic tree in $C$ or $C^-$. By Lemma~\ref{focalNode_nn}, $\mathit{focalNode}(cp,C^-) = nextNode(P,C^-)$. By Lemma~\ref{focalNode_on_sp}, $nextNode(P,C^-)$ is on $P$'s search path in $C^-$. Thus, $nextNode(P,C^-)$ is a node on the path from $entry$ to \un{xl} in $C^-$. If $nextNode(P,C^-)$ is still in the chromatic tree in $C$, then $nextNode(P,C) = nextNode(P,C^-)$ and is a node on the path from $entry$ to $u$. Since $x$ is a proper descendant of $u$ and added to $P$'s search path, $x$ is on $P$'s search path from $nextNode(P,C)$ in $C$. If $nextNode(P,C^-)$ is not in the chromatic tree in $C$, then it is a node removed by $ucas$. There are no nodes in the chromatic tree from $nextNode(P,C^-)$ to the leaf $\un{xl}$ in $C$. By Lemma~\ref{nextNode_obs}, $nextNode(P,C)$ is a node $t$ on $P$'s stack that is in the chromatic tree in $C$. Note that since $P$'s stack is the same in $C^-$ and $C$, $t$ is also on $P$'s stack in $C^-$. By Lemma~\ref{on_stack_SP_2}, the nodes on $P$'s stack in $C^-$ are on $P$'s search path in $C^-$. Thus, $nextNode(P,C)$ is a node on the path from $entry$ to $u$ in $C$. Therefore, $x$ is on $P$'s search path from $nextNode(P,C)$ in $C$.
\end{proof}

\begin{lemma}\label{J_backtracking}\normalfont
Consider a step $s$ by a process $P$ in a cleanup phase $cp$ that pops a node from its stack and updates its local variable $l$ to point to this node. If $C$ is the configuration immediately after $s$, then $\Delta J(cp, C) \leq 0$.
\end{lemma}

\begin{proof}
Note that, from the code, the pop by $s$ is performed on lines~\ref{ln:cleanup:first_pop} and \ref{ln:cleanup:pop} of \textsc{BacktrackingCleanup}. These pops are performed during backtracking, and so $s$ does not belong to the first attempt of $cp$. Let $m$ be the node popped by $s$, so $location(P,C) = m$. Let $C^-$ be the configuration immediately before $C$.  We consider multiple cases, depending on if $s$ contains the first pop in the current attempt $A$ of $cp$. 
\begin{itemize}
	\item Case 1: Suppose the pop by $s$ is not the first pop in $A$. So this pop occurs on line~\ref{ln:cleanup:pop}. We first argue that $\mathit{focalNode}(cp, C^-) = \mathit{nextNode}(P,C^-)$. Let $s'$ be a step prior to $s$ in which a pop is last performed by $cp$. Let $C'$ be the configuration immediately after $s'$, and let $m'$ be the node popped by $s'$. From the code, $P$'s local variable $l$ is set to $m'$ by $s'$ and is not updated again until $s$, so $location(P,C^-) = m'$. The node $m'$ was marked in $C'$. Otherwise backtracking would have ended after $s'$ by the following check on line~\ref{ln:cleanup:backtracking_start} of \textsc{BacktrackingCleanup}, and hence $s$ will not be performed during $A$. By Lemma~\ref{backtracking_marked}, $m'$ is removed from the chromatic tree before the next \textsc{Pop} for $cp$ is performed. So $location(P,C^-)$ is not in the chromatic tree in $C^-$. By Lemma~\ref{focalNode_nn}, $\mathit{focalNode}(cp, C^-) = \mathit{nextNode}(P,C^-)$.
	
	Suppose $m$ is not in the chromatic tree. Then $\mathit{nextNode}(P,C) = \mathit{nextNode}(P,C^-)$. By Lemma~\ref{focalNode_nn}, $\mathit{focalNode}(cp, C) = \mathit{nextNode}(P,C)$. Hence $\mathit{focalNode}(cp,C^-) = \mathit{focalNode}(cp,C)$. Therefore, by Lemma~\ref{J_descendant}, $\Delta J(cp,C) \leq 0$. 
	
	So suppose $m$ is in the chromatic tree. Hence, $\mathit{nextNode}(P,C^-) = m$. If $currentViol(P,C) = \textsc{Nil}$, then by Definition~\ref{focalNode}.2, $\mathit{focalNode}(cp,C) = location(P,C) = m$. Hence, $\mathit{focalNode}(cp,C^-) = \mathit{focalNode}(cp,C)$. Therefore, by Lemma~\ref{J_descendant}, $\Delta J(cp,C) \leq 0$. If $currentViol(P,C) = v \neq \textsc{Nil}$, then by Definition~\ref{focalNode}.1, $\mathit{focalNode}(cp,C) = v$. By Lemma~\ref{currentViol_main}, $v$ is a descendant of $location(P,C) = m$. Therefore, by Lemma~\ref{J_descendant}, $\Delta J(cp,C) \leq 0$.
	
	\item Case 2:  Suppose the pop by $s$ is the first pop in $A$. Note that this pop occurs on line~\ref{ln:cleanup:first_pop}. Since $s$ does not belong to the first attempt of $cp$, there exists an attempt $A'$ prior to $A$. Note that $C$ is the first configuration of $A$, and $C^-$ is the last configuration of $A'$. Note that $A'$ ended after an invocation $I$ of \textsc{TryRebalance}$(ggp,gp,p,l)$. Let the nodes pointed to by the local variables $ggp$, $gp$, $p$, and $l$ at the start of $I$ be $ggp'$, $gp'$, $p'$ and $v$, respectively. From the code, no step modifies $P$'s stack between the start of $I$ and $s$. Therefore, since $m$ is the node popped by $s$, it was the topmost node on $P$'s stack at the start of $I$. From the code, $ggp'$ is the top most node on $P$'s stack at the start of $I$, so $ggp' = m$. Additionally, no step modifies $P$'s local variable $l$ between the start of $I$ and $s$. From the code, $l$ points to $v$ at the start of $I$. Therefore, $l$ points to $v$ in $C^-$, and so $location(P,C^-) = v$.
	
	Suppose $currentViol(P,C^-) = v$. Since $s$ is not an update CAS, $s$ does not apply rule CV2 of Definition~\ref{currentViol}, so $currentViol(P,C) = v$. By Definition~\ref{focalNode}.1, $\mathit{focalNode}(cp, C) = \mathit{focalNode}(cp, C^-) = v$. By Lemma~\ref{J_descendant}, $\Delta J(cp, C) \leq 0$. 
	
	So suppose $currentViol(P,C^-) = \textsc{Nil}$. Since $m$ is on $P$'s stack in $C^-$, by Lemma~\ref{no_viol_stack}, $m$ does not contain a violation. Therefore, $s$ does not apply rule CV1 of Definition~\ref{currentViol}, and so $currentViol(P,C) = \textsc{Nil}$. By Lemma~\ref{currentViol_stale} and Definition~\ref{currentViol}, either $I$ is stale or an update CAS removed at least one of $m$, $gp'$, $p'$, or $v$ sometime between the invocation of $I$ and $C^-$. In either case, least one of $m$, $gp'$, $p'$, or $v$ is not in the chromatic tree in $C^-$.
	
	Let $t$ be the topmost node on $P$'s stack that is in the chromatic tree in $C^-$. We first argue that $\mathit{focalNode}(cp, C) = t$. Suppose $m$ is not in the chromatic tree in $C^-$, and so $m \neq t$. Then $t$ is still the topmost node on $P$'s stack that is in the chromatic tree in $C$, so by Lemma~\ref{nextNode_backtracking}, $\mathit{nextNode}(P,C) = t$. By Definition~\ref{focalNode}.3, $\mathit{focalNode}(cp, C) = t$. Now suppose $m$ is in the chromatic tree in $C^-$, and so $m = t$. Then $location(P,C) = t$ is in the chromatic tree. Thus, by Definition~\ref{focalNode}.2, $\mathit{focalNode}(cp, C) = location(P,C) = t$. In either case, $\mathit{focalNode}(cp, C) = t$.
	
	Suppose $v$ is not in the chromatic tree in $C^-$. By Lemma~\ref{nextNode_backtracking}, $nextNode(P,C^-) = t$. Since $location(P,C^-) = v$ is not in the chromatic tree, by Definition~\ref{focalNode}.3, $\mathit{focalNode}(cp, C^-) = \mathit{nextNode}(P,C^-) = t$. Since $\mathit{focalNode}(cp, C) = \mathit{focalNode}(cp, C^-)$, by Lemma~\ref{J_descendant}, $\Delta J(cp, C) = 0$.

	Now suppose $v$ is a node in the chromatic tree in $C^-$, and so $\mathit{focalNode}(cp,C^-) = v$ by Definition~\ref{focalNode}.2. Therefore, $cp$'s focal node changes from $v$ to $t$. By Lemma~\ref{no_cycle}, $t$ is not reachable from $v$ in $C^-$. By Lemma~\ref{focalPath}, $v$ is on $cp$'s search path in $C^-$, and by Lemma~\ref{on_stack_SP_2}, $t$ is on $cp$'s search path in $C^-$. So $t$ is a proper ancestor of $v$ in $C^-$. Let $k = depth(v, C^-) - depth(t, C^-) \geq 1$. Let the sequence of nodes on the path from $t$ to $v$ in $C^-$ be $\langle t, w_1, \dots, w_{k-1}, v \rangle$. We prove a series of claims that apply in this case.
	
	\paragraph{\textbf{Claim 1.}} The set $\{w_1, \dots, w_{k-1}\} - \{p', gp'\}$ is non-empty.
	
	\paragraph{\normalfont\textit{Proof of Claim 1.}}
	Consider the update CAS $ucas$ that updates the pointer of a node $u$ and removes at least one of $m$, $gp'$, or $p'$ from the chromatic tree. Since $m$, $gp'$, and $p'$ are internal nodes, $ucas$ is not an \textsc{Insert} transformation. Suppose $ucas$ removes a connected set of nodes rooted at a node $r$. By inspection of the remaining transformations of Figure~\ref{fig_transformations}, any path through the chromatic tree that passes through $r$ before the transformation will pass through at least one new internal node after the transformation. Thus, there is at least one node $w$ on the path between $t$ and $v$ in $C^-$. Since $m$, $gp'$, $p'$, and $v$ are the last nodes visited by $P$ before $C^-$, which includes at least one of the nodes removed by $ucas$, $P$ does not visit any of nodes added into the chromatic tree by $ucas$. Thus, $w \neq p', gp'$. So, $|\{w_1, \dots, w_{k-1}\} - \{p', gp'\}| \geq 1$.
	\qed
	
	\paragraph{\textbf{Claim 2.}} For each $w \in \{w_1, \dots, w_{k-1}\} - \{p', gp'\}$ and for all configurations $C'$ during $cp$ up to an including $C^-$, $w$ is not on $P$'s search path from $location(P,C')$ in $C'$. 
	
	\paragraph{\normalfont\textit{Proof of Claim 2.}}
	Suppose $w$ is on $P$'s stack in $C^-$. Since $t$ and $w$ are in the chromatic tree and $w$ is on the path from $t$ to $v$, then by Lemma~\ref{stack_ancestors}, $w$ will appear above $t$ on $P$'s stack in $C^-$. This contradicts the fact that $t$ is the topmost node on $P$'s stack that is in the chromatic tree in $C^-$. So suppose $w$ is not on $P$'s stack in $C^-$. By definition, $w \neq p', gp'$. Note that $p'$ and $gp'$ are the nodes popped off $P$'s stack on lines~\ref{ln:cleanup:p_pop} and \ref{ln:cleanup:gp_pop} of $A'$. Furthermore, $w$ is not on $P$'s search path from $location(P,C^-) = v$ in $C^-$. Hence, by Lemma~\ref{visit_on_stack}, $w$ is not on $cp$'s search path from $location(P,C')$ in $C'$ for all configurations $C'$ prior to $C^-$. 
	\qed
	
	\paragraph{\textbf{Claim 3.}} For each $w \in \{w_1, \dots, w_{k-1}\} - \{p', gp'\}$ and for all configurations $C'$ during $cp$ up to an including $C^-$, $w$ is not on $P$'s search path from $nextNode(P,C')$ in $C'$ 
	
	\paragraph{\normalfont\textit{Proof of Claim 3.}}
	Suppose, for contradiction, that there exists a configuration $C'$ prior to $C^-$ in which $w$ is on $cp$'s search path from $nextNode(P,C')$ in $C'$. In configuration $C^-$, $P$ visits $v = location(P,C^-)$, which is a node in the chromatic tree. Hence, by Lemma~\ref{nextNode_visit}, there exists a configuration $C''$ between $C'$ and $C^-$ in which $cp$ visits $nextNode(P,C')$. So $location(P,C'') = nextNode(P,C')$. By Corollary~\ref{sp_no_move_coro}, $w$ is on $P$'s search path from $location(P,C'')$ in $C''$. This contradicts the Claim 2. So there is no configuration $C'$ prior to $C^-$ in which $w$ is on $cp$'s search path from either $location(P,C')$ or $nextNode(P,C')$.
	\qed
	
	\paragraph{\textbf{Claim 4.}} For each $w \in \{w_1, \dots, w_{k-1}\} - \{p', gp'\}$, $\mathit{backup}(cp,w,C^-) = 1$.
	
	\paragraph{\normalfont\textit{Proof of Claim 4.}}
	Since $w \in \mathit{focalPath}(cp,C^-)$, then by Lemma~\ref{focalNode_on_sp}, $w$ is on $cp$'s search path in $C^-$. Let $C_s$ be the first configuration during $cp$ in which $w$ is on $cp$'s search path. Let $C_w$ be the first configuration during $cp$ in which $w$ is in the chromatic tree. Suppose $w$ is not on $cp$'s search path in $C_w$. Then there exists an update CAS $ucas$ after the start of $cp$ that adds $w$ onto $cp$'s search path. By Lemma~\ref{sp_descendant}, in $C_s$, $w$ is on $cp$'s search path from $location(P,C_s)$ in $C_s$ if $location(P,C_s)$ is in the chromatic tree, or from $nextNode(P,C_s)$ if $location(P,C_s)$ is not in the chromatic tree. This either contradicts Claim 2 or Claim 3. So $w$ is on $cp$'s search path in $C_w$. Since $w$ is in the chromatic tree in all configurations between $C_w$ and $C^-$, by Lemma~\ref{hindsight}, $w$ is on $cp$'s search path in all configurations between $C_w$ and $C^-$. If $w$ was added to the chromatic tree before the start of $cp$, then $C_w$ is the start of $cp$. So $w$ is on $cp$'s search path from $location(P,C_w) = entry$, a contradiction. Therefore, $w$ was added to the chromatic tree after the start of $cp$.
	
	Suppose there exists a configuration $\hat{C}$ between $C_w$ and $C^-$ in which $w$ is not in $\mathit{focalPath}(cp,\hat{C})$. Since $w$ is on $cp$'s search path in $\hat{C}$, $w$ is a proper descendant of $\mathit{focalNode}(cp,\hat{C})$. By Lemma~\ref{focalNode_nn} and Lemma~\ref{focalNode_loc_tree}, $\mathit{focalNode}(cp,\hat{C})$ is either $nextNode(P,\hat{C})$ or a descendant of $location(P,\hat{C})$. Thus, $w$ is on $cp$'s search path from either $nextNode(P,\hat{C})$ or $location(P,\hat{C})$ in $\hat{C}$. This either contradicts Claim 3 or Claim 2. 
	
	So for all configurations $\hat{C}$ between $C_w$ and $C^-$, $w$ is in $\mathit{focalPath}(cp,\hat{C})$. In particular, $w$ is in $\mathit{focalPath}(cp,C_w)$, and so $backup(cp,w,C_w) = 1$. Furthermore, there is no configuration between $C_w$ and $C^-$ in which $w$ is removed from $\mathit{focalPath}(cp)$. Therefore, no step sets $backup(cp,w) = 0$ between $C_w$ and $C^-$. It follows that $backup(cp,w,C^-) = 1$.
	\qed
	
	By Claim 1 and 4, it follows that
	\begin{equation*}
	\begin{aligned}
	\sum_{w \in \{w_1, \dots, w_{k-1}\}} \mathit{backup}(cp,w,C^-) &\geq \max(1,k-3). \\
	\end{aligned}
	\end{equation*}
	For all nodes $u$ in $\mathit{focalPath}(cp,C) \cap \mathit{focalPath}(cp,C^-)$, it follows by Definition~\ref{backup} that $backup(cp,u,C) = backup(cp,u,C^-)$. Since $\mathit{focalPath}(cp,C) = \mathit{focalPath}(cp,C^-) - \{ v, w_1, \dots, w_{k-1} \}$, 
	\begin{equation*}
	\begin{aligned}
	\bigg(\sum_{x \in \mathit{focalPath}(cp,C)} \mathit{backup}(cp,x,C)\bigg) - \bigg(\sum_{y \in \mathit{focalPath}(cp,C^-)} \mathit{backup}(cp,y,C^-)\bigg) &\leq -\max(1,k-3). \\
	\end{aligned}
	\end{equation*}
	
	The structure of the chromatic tree does not change as a result of $s$, so $H(u,C) = H(u,C^-)$ and $rebalSet(u,C) = rebalSet(u,C^-)$ for any node $u$ in the chromatic tree. Let $w_0 = t$ and $w_k = v$. By Lemma~\ref{J_rebalSet_move}, for $1 \leq i \leq k$,
	\begin{equation*}
	\begin{aligned}
	\bigg(\sum_{x \in rebalSet(w_{i-1},C)} H(x,C)\bigg) - \bigg(\sum_{y \in rebalSet(w_{i},C)} H(y,C)\bigg) &\leq 72. \\
	\end{aligned}
	\end{equation*}
	
	By definition, $targets(cp,C) = rebalSet(t,C)$ and $targets(cp,C^-) = rebalSet(v,C)$. It follows that
	\begin{equation*}
	\begin{aligned}
	\bigg(\sum_{x \in targets(cp,C)} H(x,C)\bigg) - \bigg(\sum_{y \in targets(cp,C^-)} H(y,C^-)\bigg) &\leq 72k. \\
	\end{aligned}
	\end{equation*}
	
	Therefore, by definition of $J$,
	\begin{equation*}
	\begin{aligned}
	\Delta J(cp, C) &\leq 2 \cdot 72k - \max(1,k-3)(576) \\
	&= 144k - 576\cdot\max(1,k-3)
	\end{aligned}
	\end{equation*}
This function has a maximum value of 0, which occurs when $k = 4$. Therefore, $\Delta J(cp, C) \leq 0$.
\end{itemize}
\end{proof}

We next give an upper bound on the change in $J(xp)$ as a result of an update CAS for any chromatic tree transformation. Unlike steps that update a process's local variable $l$ which may occur an unbounded number of times per operation, there is only 1 successful update CAS per update phase, and at most 1 successful update CAS per violation removed from the chromatic tree. We show that $\Delta J(xp)$ is bounded above by a constant.

The following lemmas state that for any $t \in targets(xp,C)$, $H(t,C) - H(\mathit{focalNode}(xp,C),C)$ is bounded by constants.

\begin{lemma}\label{update_H_to_f}\normalfont
Consider an update phase $up$ of an \textsc{Insert} or \textsc{Delete} operation, where $f = \mathit{focalNode}(up,C)$ in configuration $C$. For all $t \in targets(up,C)$, 
\begin{equation*}
0 \leq H(t,C) - H(f,C) \leq 12. 
\end{equation*}
\end{lemma}

\begin{proof}
Since the targets of an \textsc{Insert} operation are a subset of those of a \textsc{Delete} operation, it is sufficient to prove the lemma assuming $up$ is a \textsc{Delete} operation. Let $targets(up,C) = \{f, s, p, gp, ggp\}$, where $s$ is the sibling of $f$ and $p$, $gp$, and $ggp$ are the three closest proper ancestors of $f$. (See Figure~\ref{fig_v_to_all_update}). 

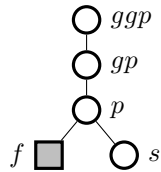
\begin{figure}[!h]
	\centering
	\begin{tikzpicture}[-,>=stealth',
	level distance=0.9cm,
	level 1/.style={sibling distance=1cm, level distance = 0.6cm},
	] 
	\node [arn_w, label=right:{$ggp$}]{ } 
	child{node [arn_w, label=right:{$gp$}] { }  
		child{node [arn_w, label=right:{$p$}] { }  
			child{ node [arn_gx, label=left:{$f$}] { } }
			child{ node [arn_w, label=right:{$s$}] { } }
		}
	}; 
	\end{tikzpicture}
	\caption{The nodes $targets(up,C)$, where $f = \mathit{focalNode}(up,C)$.}\label{fig_v_to_all_update} 
\end{figure}

For any $v \in \{p,gp,peak\}$, $1 \leq depth(f,C) - depth(v,C) \leq 3$. Thus, by Corollary~\ref{H_par_coro}, $6 \leq H(v,C) - H(f,C) \leq 12$.  By Lemma~\ref{H_sib},  $H(s,C) - H(f,C) = 0$. Therefore, for any $t \in targets(up,C)$, $0 \leq H(u,C) - H(f,C) \leq 12$.
\end{proof}

\begin{lemma}\label{cleanup_H_to_f}\normalfont
Consider an cleanup phase $cp$ where $f = \mathit{focalNode}(cp,C)$ in configuration $C$. For all $t \in targets(cp,C)$, 
\begin{equation*}
-8 \leq H(t,C) - H(f,C) \leq 16.
\end{equation*}
\end{lemma}

\begin{proof} Let $targets(cp,C) = \{ f, s, p, gp, ggp, gggp, ps, c_1, c_2, b_1, b_2, b_3, b_4 \}$, where $p$, $gp$, $ggp$, and $peak$ are the 4 closest proper ancestors of $f$, $ps$ is the sibling of $p$, $s$ is the sibling of $f$, $c_1$ and $c_2$ are the children of $s$, $b_1$ and $b_2$ are the children of $c_1$, and $b_3$ and $b_4$ are the children of $c_2$. (See Figure~\ref{fig_v_to_all}).

\begin{figure}[!h]
	\centering
	\begin{tikzpicture}[-, >=stealth', 
	level distance=0.7cm,
	level 4/.style={sibling distance=1.4cm, level distance = 0.6cm},
	level 5/.style={sibling distance=1.4cm, level distance = 0.6cm},
	level 6/.style={sibling distance=0.7cm, level distance = 0.6cm}
	] 
	\node [arn_w, label=left:{$gggp$}]{ } 
	child{node [arn_w, label=left:{$ggp$}] { }
		child{node [arn_w, label=left:{$gp$}] { }
			child{node [arn_w, label=left:{$p$}] { }   
				child{ node [arn_w, label=left:{$s$}] { } 
					child{ node [arn_w, label=left:{$c_1$}] { } 
						child{ node [arn_w, label=below:{$b_1$}] { } }
						child{ node [arn_w, label=below:{$b_2$}] { } }
					}
					child{ node [arn_w, label=left:{$c_2$}] { } 
						child{ node [arn_w, label=below:{$b_3$}] { } }
						child{ node [arn_w, label=below:{$b_4$}] { } }
					}
				}
				child{ node [arn_g, label=right:{$f = \mathit{focalNode}(cp, C)$}] { } }
			}
			child{node [arn_w, label=left:{$ps$}] { }}
		}
	}; 
	\end{tikzpicture}
	\caption{The set of nodes in $targets(cp,C) = rebalSet(f,C)$.}\label{fig_v_to_all} 
\end{figure}
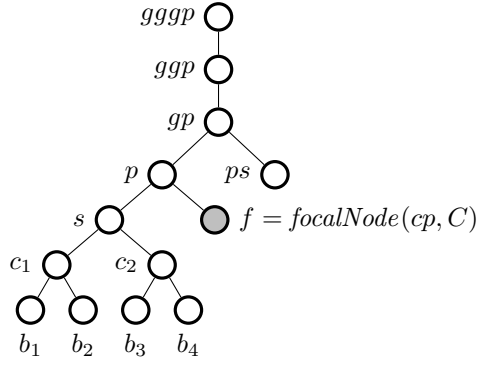

By Lemma~\ref{H_sib}, $H(s,C) = H(f,C)$. Therefore, it is sufficient to determine $H(t,C) - H(s,C)$ for any $t \in targets(cp,C)$. For $y \in \{p, gp, ggp, gggp\}$, $1 \leq depth(s,C) - depth(y,C) \leq 4$. Therefore, by Corollary~\ref{H_par_coro}, $8 \leq H(y,C) - H(s,C) \leq 16$. By Lemma~\ref{H_sib}, $H(ps,C) = H(p,C)$, and so $8 \leq H(ps,C) - H(s,C) \leq 16$. For  $x \in \{c_1, c_2, b_1, b_2, b_3, b_4 \}$, it follows from inspection that $1 \leq depth(x,C) - depth(s,C) \leq 2$. Therefore, by Corollary~\ref{H_par_coro}, $-8 \leq H(x,C) - H(s,C) \leq -4$. Therefore, for any $t \in targets(cp,C)$, $-8 \leq H(t,C) - H(f,C) \leq 16$.
\end{proof}

Since an update CAS changes the structure of the chromatic tree, the $H$ values of nodes that are still in the chromatic tree after an update CAS may change. The following lemma gives an upper bound on this change.

\begin{lemma}\label{H_ucas}\normalfont
Consider an update CAS $ucas$ for any transformation, where $C^-$ and $C$ are the configurations before and after $ucas$ respectively. If $x$ is a node in the chromatic tree in both $C^-$ and $C$, then $H(x, C) - H(x, C^-) \leq 19$.
\end{lemma}

\begin{proof}
Suppose $ucas$ changes the pointer of a node $u$, removing a set of nodes $R$ from $npa(x,C^-)$ and adding a set of nodes $A$ into $npa(x,C)$. By inspection of each transformation, the height of the chromatic tree can increase by at most 2, i.e.~$h(C) - h(C^-) \leq 2$. Additionally, $x$ can decrease in depth by at most 1, i.e.~$depth(x,C^-) - depth(x,C) \leq 1$. By the definition of $H$,
\begin{equation*}
\begin{aligned}
H(x,C) &= H(x,C^-) - \bigg(\sum_{r \in R} \gamma(r, C^-)\bigg) + \bigg(\sum_{a \in A} \gamma(a,C)\bigg) + 3( h(C) - h(C^-) + depth(x,C^-) - depth(x,C)) \\
&\leq H(x,C^-) + |R| + |A| + 12 \hspace{5cm} \text{by Lemma~\ref{gamma_bounds}.}
\end{aligned}
\end{equation*}
For any chromatic tree transformation, $|R|, |A| \leq 5$, so $H(x,C) - H(x,C^-) \leq 19$.
\end{proof}

We next give an upper bound on $\Delta J(up,C)$ for an update phase $up$ as a result of an update CAS.

\begin{lemma}\label{J_ucas_update}\normalfont
	Consider an update CAS $ucas$ for any transformation. If $C$ is the configuration immediately after $ucas$ is performed, then for any update phase $up$ active in $C$, $\Delta J(up,C) \leq 450$.
\end{lemma}

\begin{proof}
Let $C^-$ be the configuration before $C$. Suppose $ucas$ changes the pointer of a node $u$, removing nodes the $R$ from the chromatic tree and adding nodes $A$ to the chromatic tree. Let $t$ be any node in $targets(up, C)$, and $t^-$ be any node in $targets(up, C^-)$. We next argue that $H(t, C) - H(t^-, C^-) \leq 45$. Let $f^- = \mathit{focalNode}(up, C^-)$ and $f = \mathit{focalNode}(up, C)$. We consider two cases, depending on whether $f^-$ is removed by $ucas$. 
\begin{itemize}
	\item Suppose $f^-$ is not removed by $ucas$. By definition, $f^-$ is the leaf node on $up$'s search path in $C^-$. Since $f^-$ is still in the chromatic tree in $C$, by Lemma~\ref{hindsight}, $f^-$ is still on $up$'s search path in $C$. By inspection of the chromatic tree transformations, $f^-$ is still a leaf in $C$. Therefore, $f = f^-$. It follows that
	\begin{equation*}
	\begin{aligned}
	H(t, C) &\leq H(f, C) + 12 &&\text{by Lemma~\ref{update_H_to_f}} \\
	&\leq H(f, C^-) + 31 &&\text{by Lemma~\ref{H_ucas}} \\
	&= H(f^-, C^-) + 31 \\
	&\leq H(t^-, C^-) + 31 &&\text{by Lemma~\ref{update_H_to_f}.} \\
	\end{aligned}
	\end{equation*}
	
	\item Suppose $\mathit{focalNode}(up,C^-)$ is removed by $ucas$. So $\mathit{focalNode}(up,C^-) \in R$. Let $\ell = depth(f^-,C^-) - depth(u, C^-)$. By inspection of the chromatic tree transformations, $1 \leq \ell \leq 4$. 
	
	Note that $u$ is on $up$'s search path in $C^-$ since it is a node on the unique path from $entry$ to $f-$. Since $u$ is still in the chromatic tree in $C$, by Lemma~\ref{hindsight}, $u$ is on $up$'s search path in $C$. Therefore, $f$ is some leaf reachable from $u$ in $C$. Let $k = depth(f,C) - depth(u, C) \geq 1$.

	Then an upper bound on $H(t, C) - H(t^-, C^-)$ can be computed as follows:
	\begin{equation*}
	\begin{aligned}
	H(t, C) &\leq H(f, C) + 12 &&\text{by Lemma~\ref{update_H_to_f}} \\
	&\leq H(u, C) + 12 - 2k &&\text{by Corollary~\ref{H_par_coro}} \\
	&\leq H(u, C^-) + 31 - 2k &&\text{by Lemma~\ref{H_ucas}} \\
	&\leq H(f^-, C^-) + 31 - 2k + 4\ell &&\text{by Corollary~\ref{H_par_coro}} \\
	&\leq H(t^-, C^-) + 31 - 2k + 4\ell &&\text{by Lemma~\ref{update_H_to_f}.} \\
	\end{aligned}
	\end{equation*}
	The difference $H(t, C) - H(t^-, C^-)$ is maximized when $\ell = 4$ and $k = 1$, which gives $H(t, C) - H(t^-, C^-) \leq 45$. 
\end{itemize}

Therefore, in either case, $H(t, C) - H(t^-, C^-) \leq 45$. Note that $|targets(up,C^-)| = |targets(up,C)|$, and so
\begin{equation*}
\begin{aligned}
\sum_{t \in targets(up,C)} H(t,C) - \sum_{t^- \in targets(up,C^-)} H(t^-,C) \leq 45|targets(up,C)|.
\end{aligned}
\end{equation*}
Finally, $backup(up,x,C) = backup(up,x,C^-) = 0$ for any node $x$ by definition of $backup$ for an update phase. Therefore, by the definition of $J$,
\begin{equation*}
\begin{aligned}
\Delta J(up, C) &\leq 2|targets(up,C)| \cdot 45 \\
&\leq 450 &&\text{since $|targets(up, C)| \leq 5$.}
\end{aligned}
\end{equation*}
\end{proof}

Unlike $\mathit{focalNode}(up)$, $\mathit{focalNode}(cp)$ may move to a node higher in the chromatic tree as a result of an update CAS. The following lemmas are useful to help determine how $\mathit{focalNode}(cp)$ changes as a result of an update CAS.

\begin{lemma}\label{J_ucas_cleanup_f_change}\normalfont
Consider a cleanup phase $cp$ by process $P$ in a configuration $C$. Let $C^-$ and $C$ be the configurations immediately before and after an update CAS $ucas$. If $\mathit{focalNode}(cp,C) \neq \textsc{Nil}$ and $\mathit{focalNode}(cp,C^-) \neq \textsc{Nil}$, then $\mathit{focalNode}(cp,C)$ and $\mathit{focalNode}(cp,C^-)$ are both nodes in the chromatic tree in $C^-$. 
\end{lemma}

\begin{proof}
Let the node whose pointer is updated by $ucas$ be $u$, and let $A$ be the set of nodes added into the chromatic tree by $ucas$. By definition, $\mathit{focalNode}(cp,C^-)$ is a node in the chromatic tree in $C^-$. Suppose, for contradiction, that $\mathit{focalNode}(cp,C)$ is not in the chromatic tree in $C^-$. The only nodes in the chromatic tree in $C$ that are not in $C^-$ are those in $A$. So $\mathit{focalNode}(cp,C) \in A$, and is not in the chromatic tree in any configuration before $C$. We proceed with cases.
\begin{itemize}
	\item Suppose $currentViol(P,C) \neq \textsc{Nil}$. By Definition~\ref{focalNode}.1, $\mathit{focalNode}(cp,C) = currentViol(P,C)$. By Definition~\ref{currentViol}, there exists a step prior to $C$ in which $P$ visited $currentViol(P,C)$. This contradicts the fact that $\mathit{focalNode}(cp,C)$ is not in the chromatic tree in any configuration before $C$.
	
	\item Suppose $currentViol(P,C) = \textsc{Nil}$ and $location(P,C)$ is in the chromatic tree. By Definition~\ref{focalNode}.2, $\mathit{focalNode}(cp,C) = location(P,C)$. Since the node pointed to by $P$'s local variable does not change by $ucas$, $location(P,C) = location(P,C^-)$. This implies $P$ visited a node in $A$ before $C$, which contradicts the fact that $\mathit{focalNode}(cp,C)$ is not in the chromatic tree before $C$.
	
	\item Suppose $currentViol(P,C) = \textsc{Nil}$, and $location(P,C)$ is not in the chromatic tree. By Definition~\ref{focalNode}.3, $\mathit{focalNode}(cp,C) = \mathit{nextNode}(P,C)$. By Lemma~\ref{nextNode_obs}, in configuration $C$, $\mathit{focalNode}(cp,C)$ is either a node on $cp$'s stack, or the first node in the chromatic tree on $cp$'s search path from $location(P,C)$.
	
	Suppose $\mathit{focalNode}(cp,C)$ is a node on $P$'s stack. This implies that there exists a configuration before $C$ in which $\mathit{focalNode}(cp,C)$ was visited by $P$ and pushed onto the stack. This implies $\mathit{focalNode}(cp,C)$ was in the chromatic tree before $C$, a contradiction.
	
	Then $\mathit{focalNode}(cp,C)$ is the first node in the chromatic tree on $cp$'s search path from $location(P,C)$. The nodes in $A$ are only reachable by following a pointer from $u$ in $C$. Therefore, $cp$ will visit $u$ before visiting a node in $A$ during a solo execution starting from $C$. Thus, $\mathit{nextNode}(P,C) \notin A$, a contradiction.
\end{itemize}
\end{proof}

\begin{lemma}\label{J_ucas_cleanup_f_ancestor}\normalfont
	Consider a cleanup phase $cp$ by process $P$ in a configuration $C$. Let $C^-$ and $C$ be the configurations immediately before and after an update CAS $ucas$. If $\mathit{focalNode}(cp,C) \neq \textsc{Nil}$ and $\mathit{focalNode}(cp,C^-) \neq \textsc{Nil}$, then $\mathit{focalNode}(cp,C)$ and $\mathit{focalNode}(cp,C^-)$ are on $cp$'s search path in $C^-$.
\end{lemma}

\begin{proof}
Suppose $\mathit{focalNode}(cp,C)$ is not on $cp$'s search path in $C^-$.  By Lemma~\ref{J_ucas_cleanup_f_change}, $\mathit{focalNode}(cp,C)$ is a node in the chromatic tree in $C^-$. By Lemma~\ref{focalNode_on_sp}, $focalNode(cp,C)$ is on $cp$'s search path in $C$. Therefore, $ucas$ adds $\mathit{focalNode}(cp,C)$ onto $cp$'s search path. 

Let $u$ be the node whose pointer is modified by $ucas$. Only \textsc{Delete} transformations can add nodes that are already in the chromatic tree onto $cp$'s search path. In this case, $ucas$ adds a path of nodes from $u$ to a leaf onto $cp$'s search path. Lemma~\ref{sp_descendant} implies each node on this path is on $cp$'s search path starting from $\mathit{focalNode}(cp,C)$ in $C$. Therefore, $\mathit{focalNode}(cp,C)$ is a node on the path from $entry$ to $u$ in $C^-$. Since ancestors of $u$ are not modified by $ucas$, this contradicts the fact that $ucas$ adds $\mathit{focalNode}(cp,C)$ onto $cp$'s search path. 
\end{proof}

The proof of the following result is similar to the proof of Lemma~\ref{J_ucas_update}.

\begin{lemma}\label{J_ucas_cleanup}\normalfont
	Consider an update CAS $ucas$ for any transformation. If $C$ is the configuration immediately after $ucas$ is performed, then for any cleanup phase $cp$ active in $C$, $\Delta J(cp, C) \leq 4934$.
\end{lemma}

\begin{proof}
Let $C^-$ be the configuration before $C$. Let $P$ be the process performing $cp$. Suppose $ucas$ changes the pointer of node $u$, removing the set of nodes $R$ from the chromatic tree and adding the set of nodes $A$ to the chromatic tree. To calculate how $J$ changes as a result of $ucas$, we consider how $\mathit{focalNode}(cp)$ changes from $C^-$ to $C$. Let $f^- = \mathit{focalNode}(cp, C^-)$ and $f = \mathit{focalNode}(cp, C)$. 

If $f^- = nextNode(P,C^-) = \textsc{Nil}$, then by Lemma~\ref{nextNode_nil}, $f = nextNode(P,C) = \textsc{Nil}$. In this case, by the definition of $J$, $J(cp,C^-) = J(cp,C) = 0$. Thus, $\Delta J(cp, C) = 0$. So suppose $f^- \neq \textsc{Nil}$ and $f = \textsc{Nil}$. By Observation~\ref{J_positive}, $J(cp,C^-) \geq 0$, and by the definition of $J$, $J(cp,C) = 0$. Thus, $\Delta J(cp, C) \leq 0$.

For the remainder of the proof, suppose $f \neq \textsc{Nil}$ and $f^- \neq \textsc{Nil}$. By Lemma~\ref{J_ucas_cleanup_f_change}, $f$ and $f^-$ are nodes in the chromatic tree in $C^-$, and by Lemma~\ref{J_ucas_cleanup_f_ancestor}, $f$ and $f^-$ are lie on the same path from root to leaf in the chromatic tree. Let $k = depth(f, C^-) - depth(f^-, C^-)$. If $k \geq 0$, then $f^-$ is an ancestor of $f$ and
\begin{equation*}
\begin{aligned}
H(f, C^-) & \leq H(f^-, C^-) - 2k &&\text{by Corollary~\ref{H_par_coro}.} \\
\end{aligned}
\end{equation*}
Otherwise, $k < 0$, so $f$ is an ancestor of $f^-$, and
\begin{equation*}
\begin{aligned}
H(f, C^-) & \leq H(f^-, C^-) - 4k &&\text{by Corollary~\ref{H_par_coro}.} \\
\end{aligned}
\end{equation*}
Thus, in both cases, $H(f, C^-) \leq H(f^-, C^-) - 3k + |k|$. Let $t$ be any node in $targets(cp, C)$, and $t^-$ be any node in $targets(cp, C^-)$. For any value of $k$, 
\begin{equation*}
\begin{aligned}
H(t, C) & \leq H(f, C) + 16 &&\text{by Lemma~\ref{cleanup_H_to_f}} \\
& \leq  H(f, C^-) + 35 &&\text{by Lemma~\ref{H_ucas}} \\
&\leq H(f^-, C^-) + 35 - 3k + |k| \\
&\leq H(t^-, C^-) + 43 - 3k + |k| &&\text{by Lemma~\ref{cleanup_H_to_f}.} \\
\end{aligned}
\end{equation*}	
Therefore, $H(t, C) -  H(t^-, C^-) \leq 43 - 3k + |k|$. Since $|targets(cp,C)| = |targets(cp,C^-)| = 13$, 
\begin{equation*}
\begin{aligned}
\bigg(\sum_{t \in targets(cp,C)} H(t,C)\bigg) - \bigg(\sum_{t^- \in targets(cp,C^-)} H(t^-,C)\bigg) &\leq 13(43 - 3k + |k|). \\
\end{aligned}
\end{equation*}

We consider the following 2 cases, depending on if $k \geq -9$ or $k < -9$. The choice of $k < -9$ was chosen to limit the number of possible ways $focalNode(cp)$ can change from $C^-$ to $C$.
\begin{itemize}
	\item Case 1: Suppose $k \geq -9$. So $\mathit{focalNode}(cp)$ either decreases in depth by at most 4, or increases in depth from $C^-$ to $C$.
	
	By Definition~\ref{backup}, among the nodes $v \in \mathit{focalPath}(cp,C) - \mathit{focalPath}(cp,C^-)$, only those also in the set $A$ may change $\mathit{backup}(cp,v,C)$ to 1. It follows that
	\begin{equation*}
	\begin{aligned}
	\bigg(\sum_{v \in \mathit{focalPath}(cp,C)} \mathit{backup}(cp,v,C)\bigg) - \bigg(\sum_{w \in \mathit{focalPath}(cp,C^-)} \mathit{backup}(cp,w,C^-)\bigg) &\leq |A|. \\
	\end{aligned}
	\end{equation*}
	
	By inspection of the chromatic tree transformations, at most 5 new nodes are added to the chromatic tree as a result of an update CAS, so $|A| \leq 5$. Therefore, by definition of $J$,
	\begin{equation*}
	\begin{aligned}
	\Delta J(cp,C) &\leq 26(43-3k+|k|) + 576|A| \\
	&\leq 3998 - 78k + 26|k| &&\text{since $|A| \leq 5$} \\ 
	&\leq 4934 &&\text{since $k \geq -9$.} \\  
	\end{aligned}
	\end{equation*}
	
	\item Case 2: Suppose $k < -9$. So $\mathit{focalNode}(cp)$ decreases in depth by at least 10 from $C^-$ to $C$. 
	
	Note that $location(P,C^-) = location(P,C)$ since $P$'s local variable $l$ does not change as a result of $ucas$. Let $loc = location(P,C) = location(P,C^-)$. We prove a series of claims that apply in this case.

	
	\paragraph{\textbf{Claim 1.}} In configuration $C$, $currentViol(cp,C) = \textsc{Nil}$.
	
	\paragraph{\normalfont\textit{Proof of Claim 1.}}
	By Definition~\ref{currentViol}, if $ucas$ changes $currentViol(cp)$, it changes it to \textsc{Nil}. Hence, if $currentViol(cp,C) \neq \textsc{Nil}$, then $currentViol(cp,C) = currentViol(cp,C^-)$. In this case, by Definition~\ref{focalNode}.1, $f^- = currentViol(cp,C^-) = currentViol(cp,C) = f$. Since $k = 0$ when $f = f^-$, it follows that $currentViol(cp,C) = \textsc{Nil}$. 
	\qed
	
	\paragraph{\textbf{Claim 2.}} If $loc$ is in the chromatic tree in $C^-$, then it is removed from the chromatic tree by $ucas$.
	
	\paragraph{\normalfont\textit{Proof of Claim 2.}}
	Suppose, for contradiction, that $loc$ is in the chromatic tree in $C^-$ and $C$. Since $currentViol(cp,C) = \textsc{Nil}$, by Definition~\ref{focalNode}.2, $f = loc$. If $f^- = loc$, then $f = f^-$, so $k = depth(f,C^-) - depth(f^-, C^-) = 0$, a contradiction. Therefore, $f^- = currentViol(P,C^-)$. By Lemma~\ref{currentViol_main}.1, $loc$ is one of $currentViol(P,C^-)$ or its 3 closest proper ancestors, so $-3 \leq depth(loc,C^-) - depth(f^-, C^-) \leq 0$. This implies $k = depth(f,C^-) - depth(f^-, C^-) \geq -3$, a contradiction.
	\qed
	
	\paragraph{\textbf{Claim 3.}} If $loc$ is not in the chromatic tree in $C^-$, then $f^-$ is removed from the chromatic tree by $ucas$.
	
	\paragraph{\normalfont\textit{Proof of Claim 3.}}
 	Suppose, for contradiction, that $f^-$ is not removed from the chromatic tree.	By Lemma~\ref{focalNode_nn}, $f^- = nextNode(P,C^-)$ and $f = nextNode(P,C)$. By Lemma~\ref{nextNode_visit}, $P$ only visits nodes that are no longer in the chromatic tree in $C^-$  before visiting $f^-$ in a solo execution starting from $C^-$. Thus, $nextNode(P,C^-) = nextNode(P,C)$. Since $f = f^-$, this implies that $k = 0$, a contradiction.
	\qed
	
	\paragraph{\textbf{Claim 4.}} In configuration $C$, $f = nextNode(P,C)$ and is a node on $P$'s stack.
	
	\paragraph{\normalfont\textit{Proof of Claim 4.}}
	Note that by Claim 2, $loc$ is not in the chromatic tree in $C$, and so, by Lemma~\ref{focalNode_nn}, $f = \mathit{nextNode}(P,C)$. By Lemma~\ref{nextNode_obs}, $f$ is either the first node on $P$'s search path starting from $loc$, or a node on $P$'s stack in $C$. Suppose, for contradiction, that $f$ is the first node on $P$'s search path starting from $loc$ in $C$. 
	
	Suppose $loc$ is in the chromatic tree in $C^-$. Then by Lemma~\ref{focalNode_loc_tree}, the distance between $loc$ and $f^-$ in $C^-$ is at most 3. By  Claim 2, $loc$ is removed by $ucas$, so $u$ is a proper ancestor of $loc$. Therefore, the path between $loc$ and $f^-$ does not change as a result of $ucas$. Thus, since $f$ is a proper ancestor of $f^-$ in $C^-$, $f$ is on the path between $loc$ and $f^-$ in $C^-$. This contradicts the fact that $k < 9$.
	
	So suppose $loc$ is not in the chromatic tree in $C^-$. By Lemma~\ref{focalNode_nn}, $f^- = nextNode(P,C^-)$. By Lemma~\ref{nextNode_obs}, in configuration $C^-$, $f^-$ is either a node on $P$'s stack or the first node on $P$'s search path starting from $loc$ in $C^-$. If $f^-$ is the first node on $P$'s search path starting from $loc$, then since $f$ is a proper ancestor of $f^-$, $f$ is not on $P$'s search path from $loc$, a contradiction.  So $f^-$ is a node on $P$'s stack. In a solo execution starting from $C$, $P$ pops $f^-$ from its stack before visiting $f$. By Lemma~\ref{nextNode_backtracking}, $f$ is a node on $P$'s stack in $C$.
	\qed\\
	
	In configuration $C^-$, there are $|k| + 1$ nodes on the path from $f$ to $f^-$ inclusive. Let $\langle f = w_0, w_1, ..., w_{|k|-1}, w_{|k|} = f^- \rangle$ be the sequence of nodes between $f$ and $f^-$. If $cp$ has executed lines~\ref{ln:cleanup:p_pop} and \ref{ln:cleanup:gp_pop} sometime during its cleanup attempt containing $C^-$ but prior to $C^-$, let $p$ and $gp$ be the nodes popped on these lines respectively. Otherwise, let $p$ and $gp$ be \textsc{Nil}.
	
	\paragraph{\textbf{Claim 5.}} For each $w \in  \{w_1, ..., w_{|k|-7}\} - \{p,gp\}$, $w$ is in the chromatic tree in both $C^-$ and $C$. 
	
	\paragraph{\normalfont\textit{Proof of Claim 5.}}
	By Claim 2 and Claim 3, at least one of $loc$ or $f^- = w_{|k|}$ is removed from the chromatic tree by $ucas$. By Lemma~\ref{focalNode_loc_tree}, if $loc$ is in the chromatic tree in $C$, then the distance between $loc$ and $f^-$ is at most 3. Thus, one of $\{w_{|k|-3}, ..., w_{|k|}\}$ is removed from the chromatic tree by $ucas$. 
	
	Inspection of the chromatic tree transformations shows that along any path starting at $u$, $ucas$ removes at most 4 nodes, which are consecutive proper descendants of $u$ starting at a child of $u$. Therefore, the distance between $u$ and the node removed from $\{w_{|k|-3}, ..., w_{|k|}\}$ is at most 4. So $u \in \{w_{|k|-7}, ..., w_{|k|-1}\}$. The nodes on the path from $w_0$ to $u$ are not removed by $ucas$. Therefore, for each $w \in  \{w_1, ..., w_{|k|-7}\} - \{p,gp\}$, $w$ is in the chromatic tree in both $C^-$ and $C$. 
	\qed
	
	\paragraph{\textbf{Claim 6.}} For each $w \in  \{w_1, ..., w_{|k|-7}\} - \{p,gp\}$, $w$ is not on $P$'s stack in $C^-$. 
	
	\paragraph{\normalfont\textit{Proof of Claim 6.}}
	Suppose $w$ is on $P$'s stack in $C^-$. In configuration $C^-$, $f$ and $w$ are in the chromatic tree and $w$ is on the path from $f$ to $f^-$, so by Lemma~\ref{stack_ancestors}, $w$ will appear above $f$ on $P$'s stack in $C^-$. Therefore, $w$ will be popped by $P$ before $f$ in a solo execution starting from $C$. But this implies $nextNode(P,C) = w \neq f$, a contradiction. Therefore, $w$ is not on $P$'s stack in $C^-$. 
	\qed
	
	\paragraph{\textbf{Claim 7.}} For each $w \in  \{w_1, ..., w_{|k|-7}\} - \{p,gp\}$, $w$ is not on $P$'s search path starting from $loc$ in $C^-$. 
	
	\paragraph{\normalfont\textit{Proof of Claim 7.}}
	If $loc$ is in the chromatic tree in $C^-$, then by Lemma~\ref{focalNode_loc_tree}, $loc \in \{w_{|k|-3}, ..., w_{|k|}\}$, so $w$ is not on $P$'s search path from $f^-$ in $C^-$.
	
	If $loc$ is not in the chromatic tree in $C^-$, then by Lemma~\ref{focalNode_nn}, $f^- = nextNode(P,C^-)$. By Lemma~\ref{nextNode_visit}, $P$ visits $f^-$ before visiting any other node in the tree. Since $w$ is an ancestor of $f^-$ in the tree in $C^-$, $w$ is not on $P$'s search path from $loc$ in $C^-$.
	\qed
	
	\paragraph{\textbf{Claim 8.}} For each $w \in  \{w_1, ..., w_{|k|-7}\} - \{p,gp\}$ and for all configurations $C'$ during $cp$ up to an including $C^-$, $w$ is not on $P$'s search path from $location(P,C')$ in $C'$. 
	
	\paragraph{\normalfont\textit{Proof of Claim 8.}}
	By definition, $w \neq p$ and $w \neq gp$, so it is not one of the nodes popped off $P$'s stack on lines~\ref{ln:cleanup:p_pop} and \ref{ln:cleanup:gp_pop} in $cp$'s attempt containing $C^-$. By Claim 5, Claim 6, Claim 7, and Lemma~\ref{visit_on_stack}, $w$ is not on $cp$'s search path from $location(P,C')$ in $C'$ for all configurations $C'$ during $cp$ prior to $C^-$. 
	\qed
	
	\paragraph{\textbf{Claim 9.}} For each $w \in  \{w_1, ..., w_{|k|-7}\} - \{p,gp\}$ and for all configurations $C'$ during $cp$ up to an including $C^-$, $w$ is not on $P$'s search path from $nextNode(P,C')$ in $C'$.
	
	\paragraph{\normalfont\textit{Proof of Claim 9.}}
	Suppose, for contradiction, that there exists a configuration $C'$ during $cp$ prior to $C^-$ in which $w$ is on $cp$'s search path from $nextNode(P,C')$ in $C'$. In particular, $nextNode(P,C') \neq \textsc{Nil}$. By Claim 8, $location(P,C') \neq nextNode(P,C')$. By Lemma~\ref{focalNode_nn_neq_loc}, $\mathit{focalNode}(cp,C') = nextNode(P,C')$. Since $w$ is in the chromatic tree in all configurations between $C'$ and $C^-$, by Corollary~\ref{sp_no_move_coro}, $w$ is on $P$'s search path from $nextNode(P,C')$ for all configurations between $C'$ and $C^-$.
	
	By Claim 8, for all configurations $C''$ between the start of $cp$ and $C^-$, $w$ is not on $P$'s search path from $location(P,C'')$. So for all configurations $C''$ between $C'$ and $C^-$, $location(P,C'') \neq nextNode(P,C')$, i.e. $P$ does not visit $nextNode(P,C')$ between $C'$ and $C^-$. Therefore, by Lemma~\ref{nextNode_visit}, $location(P,C^-)$ is not in the chromatic tree. By Lemma~\ref{focalNode_nn}, $f^- = nextNode(P,C^-)$. By Lemma~\ref{nextNode_obs}, in configuration $C'$, $nextNode(P,C')$ is either the first node on $P$'s search path starting from $location(P,C')$ that is in the chromatic tree, or it is a node on $P$'s stack. 
	
	Suppose that, in configuration $C'$, $nextNode(P,C')$ is the first node on $P$'s search path from $location(P,C')$ that is in the chromatic tree. Since nodes that are no longer in the chromatic tree do not change, the path between $location(P,C')$ and $nextNode(P,C')$ is the same in $C'$ and $C^-$. Hence, in a solo execution starting from $C'$, the first nodes visited by $P$ are those on the path from $location(P,C')$ to $nextNode(P,C')$. Since $P$ does not visit $nextNode(P,C')$ between $C'$ and $C^-$, it follows that in $C^-$, $loc$ is on the path between $location(P,C')$ and $nextNode(P,C')$. Therefore, $nextNode(P,C')$ is on $P$'s search path from $loc$. Since $w$ is on $P$'s search path from $nextNode(P,C')$ in $C^-$, $w$ is on $P$'s search path from $loc$ in $C^-$. This contradicts Claim 8.
	
	Therefore, $nextNode(P,C')$ is on $P$'s stack in $C'$. When $P$ visits a node during backtracking, it pops that node from its stack. Since $P$ does not visit $nextNode(P,C')$ at or before $C^-$, $nextNode(P,C')$ is still on $P$'s stack in $C^-$. Lemma~\ref{nextNode_visit} implies that, during $P$'s solo execution starting from $C'$, $P$ only visits nodes that are no longer in the chromatic tree in $C'$ before it visits $nextNode(P,C')$. Since $f^-$ is in the chromatic tree in $C^-$, $P$ visits $nextNode(P,C')$ before it visits $f^-$ in its solo execution starting from $C^-$.
	
	Since $P$ is backtracking when it pops $nextNode(P,C')$ during its solo execution in $C^-$, by Lemma~\ref{nextNode_backtracking}, $f^-$ is a node on $P$'s stack from $C^-$. Hence, $f^-$ was pushed onto its stack before $nextNode(P,C')$ was pushed onto its stack. 
	
	Let $\alpha'$ be the prefix of the execution up to and including $C^-$. Since there is a path from $nextNode(P,C')$ to $w$ in $C'$ and a path from $w$ to $f^-$ in $C^-$, by Definition~\ref{global_graph}, in the global graph $G_{\alpha'}$, there is a path from $nextNode(P,C')$ to $f^-$. Consider the configuration $\hat{C}$ immediately after $P$ last pushed $nextNode(P,C')$ onto its stack. From the code, only nodes pointed to by $P$'s local variable $l$ are pushed onto $P$'s stack, so $location(P,\hat{C}) = nextNode(P,C')$. Since $f^-$ is on $P$'s stack in $\hat{C}$, by Lemma~\ref{no_cycle_claim}, there is a path from $f^-$ to $nextNode(P,C')$ in $G_{\alpha'}$. Therefore, there is a cycle in $G_{\alpha'}$ containing $f^-$, contradicting Lemma~\ref{g_alpha}, contradicting Lemma~\ref{g_alpha}.
	\qed
	
	\paragraph{\textbf{Claim 10.}} For each $w \in  \{w_1, ..., w_{|k|-7}\} - \{p,gp\}$, $\mathit{backup}(cp,w,C^-) = 1$.
	
	\paragraph{\normalfont\textit{Proof of Claim 10.}}
	Since $w \in \mathit{focalPath}(cp,C^-)$, then by Lemma~\ref{focalNode_on_sp}, $w$ is on $cp$'s search path in $C^-$. Let $C_s$ be the first configuration during $cp$ in which $w$ is on $cp$'s search path. Let $C_w$ be the first configuration during $cp$ in which $w$ is in the chromatic tree. Suppose $w$ is not on $cp$'s search path in $C_w$. Then there exists an update CAS $ucas$ after the start of $cp$ that adds $w$ onto $cp$'s search path. By Lemma~\ref{sp_descendant}, in $C_s$, $w$ is on $cp$'s search path from $location(P,C_s)$ if $location(P,C_s)$ is in the chromatic tree, or from $nextNode(P,C_s)$ if $location(P,C_s)$ is not in the chromatic tree. This either contradicts Claim 8 or Claim 9. So $w$ is on $cp$'s search path in $C_w$. Since $w$ is in the chromatic tree in all configurations between $C_w$ and $C^-$, by Lemma~\ref{hindsight}, $w$ is on $cp$'s search path in all configurations between $C_w$ and $C^-$. If $w$ is added into the chromatic tree before the start of $cp$, then $C_w$ is the start of $cp$. So $w$ is on $cp$'s search path from $location(P,C_w) = entry$, which contradicts Claim 8. Therefore, $w$ was added into the chromatic tree after the start of $cp$.
	
	Suppose there exists a configuration $\hat{C}$ between $C_w$ and $C^-$ in which $w$ is not in $\mathit{focalPath}(cp,\hat{C})$. Since $w$ is on $cp$'s search path in $\hat{C}$, $w$ is a proper descendant of $\mathit{focalNode}(cp,\hat{C})$. By Lemma~\ref{focalNode_nn} and Lemma~\ref{focalNode_loc_tree}, $\mathit{focalNode}(cp,\hat{C})$ is either $nextNode(P,\hat{C})$ or a descendant of $location(P,\hat{C})$. Thus, $w$ is on $cp$'s search path from either $nextNode(P,\hat{C})$ or $location(P,\hat{C})$ in $\hat{C}$. This either contradicts Claim 9 or Claim 8. 
	
	So for all configurations $\hat{C}$ between $C_w$ and $C^-$, $w$ is in $\mathit{focalPath}(cp,\hat{C})$. In particular, $w$ is in $\mathit{focalPath}(cp,C_w)$, and so $backup(cp,w,C_w) = 1$. Furthermore, there is no step between $C_w$ and $C^-$ in which $w$ is removed from $\mathit{focalPath}(cp)$. Therefore, no step sets $backup(cp,w) = 0$ between $C_w$ and $C^-$. It follows that $backup(cp,w,C^-) = 1$.
	\qed\\

	Note that since $u$ is a proper descendant of $f$, for each node $a \in A$,  $a \notin \mathit{focalPath}(cp,C)$. By Definition~\ref{backup}, $\mathit{backup}(cp,a, C) = 0$. Additionally,  $\mathit{focalPath}(cp,C^-) - \mathit{focalPath}(cp,C) = \{ w_1,\dots,w_{|k|-1},f^-\}$ and $\mathit{focalPath}(cp,C) \subseteq \mathit{focalPath}(cp,C^-)$. By Definition~\ref{backup}, for each $v \in \mathit{focalPath}(cp,C)$, $backup(cp,v,C) = backup(cp,v,C^-)$. By Claim 10, for each $w \in  \{w_1, ..., w_{|k|-7}\} - \{p,gp\}$, $\mathit{backup}(cp,w,C^-) = 1$, so
	\begin{equation*}
	\begin{aligned}
	\bigg(\sum_{u \in \mathit{focalPath}(cp,C)} \mathit{backup}(cp,u,C)\bigg) - \bigg(\sum_{v \in \mathit{focalPath}(cp,C^-)} \mathit{backup}(cp,v,C^-)\bigg) &\leq -(|k| - 9) \\ 
	&= k + 9. \\
	\end{aligned}
	\end{equation*}
	
	By definition of $\Delta J(cp,C)$,	
	\begin{equation*}
	\begin{aligned}
	\Delta J(cp,C) &\leq 26(43-3k+|k|) + 576(k+9) \\
	&= 6302 + 498k + 26|k| \\
	&\leq 2054  &&\text{since $k < -9$.} \\
	\end{aligned}
	\end{equation*}
\end{itemize}

In either case, $\Delta J(cp,C) \leq 4934$.
\end{proof}

The rules for the $B$ account for update phases and cleanup phases to deposit dollars whenever a step causes an increase in $J$. In this manner, we maintain that $B(xp, C) \geq J(xp, C)$. We formalize this in the following lemma.
\begin{lemma}\label{J_positive}\normalfont
For an update or cleanup phase $xp$ and all configurations $C$ in any execution,
\begin{equation*}
B(xp, C) \geq J(xp, C).
\end{equation*}
\end{lemma}

\begin{proof}
Recall that if $xp$ is not active then $J(xp) = 0$. Before $xp$ is active, $B(xp) = 0$, and after $xp$ is no longer active, $B(xp)$ doesn't change. In particular, the lemma holds trivially in the initial configuration with no active operations. Assume the lemma holds in the configuration $C^-$ before $C$. To show $B(xp, C) \geq J(xp, C)$, we show $\Delta B(xp,C) \geq \Delta J(xp,C)$. Let $s$ be the step taking the system from configuration $C^-$ to configuration $C$.

First, we show that the lemma holds immediately after the first step of a new update phase or cleanup phase. Next, we consider steps that may change the value of $J(xp)$ or $B(xp)$. The only types of steps during update phases that cause such changes are successful freezing CASs, successful abort steps, successful update CASs, successful commit steps, and the completion of attempts. In addition to these steps, cleanup phases may change $J$ due to steps that update its local variable $l$. We consider each type of step in turn.
\begin{itemize}
	\item Suppose $s$ is the invocation of new operation, and thus the start of a new update phase $up$. By D1-BUP,  $30h(up) + 120\dot{c}(up) + 120$ dollars are deposited into $B(up)$ at the start of $up$. By Lemma~\ref{H_max}, the maximum value of $H$ is $3h(C) + 12\dot{c}(C) + 12$. It follows that
	\begin{equation*}
	\begin{aligned}
	\sum_{u \in targets(cp,C)} 2H(u, C) &\leq 2|targets(up,C)|(3h(C) + 12\dot{c}(C) + 12) \\
	&\leq 30h(C) + 120\dot{c}(C) + 120 &&\text{since $|targets(up,C)| \leq 5$.}
	\end{aligned}
	\end{equation*}
	Finally, $h(up) = h(C)$ by definition. Therefore, $\Delta B(up,C) = 30h(up) + 120\dot{c}(C) + 120 \geq \Delta J(up,C)$.
	
	Next, we show the lemma holds for each update phase $up' \neq up$ and cleanup phase $cp'$ active in $C$. By D1-BUP, $\Delta B(up',C) = 468$ and $\Delta B(up',C) = 468$. Since there is exactly one more active operation in $C$ compared to $C^-$, $\dot{c}(C) = \dot{c}(C^-) + 1$. Therefore, $\Delta H(x,C) = 18$ for each node $x$ in the chromatic tree. Since $|targets(up',C)| \leq 5$, it follows from the definition of $J$ that $\Delta J(up',C) \leq 180$. Likewise, since  $|targets(cp',C)| = 13$, $\Delta J(cp',C) \leq 468$. Thus, $\Delta B(up',C) \geq \Delta J(up',C)$ and $\Delta B(cp',C) \geq \Delta J(cp',C)$.
	
	\item Suppose $s$ is the start of a new cleanup phase $cp$. Note that $\dot{c}(C) = \dot{c}(C^-)$ since no operations are started or completed. Therefore, only $J(cp,C)$ changes as a result of $s$.  Since $cp$ is not active in $C^-$, $J(cp,C^-) = 0$. By Lemma~\ref{H_max}, the maximum value of $H$ is $3h(C) + 12\dot{c}(C) + 12$. It follows that
	\begin{equation*}
	\begin{aligned}
	\sum_{u \in targets(cp,C)} 2H(u, C) &\leq 2|targets(cp,C)|(3h(C) + 12\dot{c}(C) + 12) \\
	&\leq 78h(C) + 312\dot{c}(C) + 312 &&\text{since $|targets(cp,C)| \leq 13$.}
	\end{aligned}
	\end{equation*}
	When $cp$ is invoked, $\mathit{focalPath}(cp,C) = \{ entry \}$ and $\mathit{backup}(cp,entry,C) = 0$, so
	\begin{equation*}
	\sum_{v \in \mathit{focalPath}(cp,C)} 576 \cdot \mathit{backup}(cp,v,C) = 0.
	\end{equation*}
	Therefore $\Delta J(cp,C) \leq 78h(C) + 312\dot{c}(C) + 312$. By definition, $h(cp) = h(C)$. Therefore, by D1-BCP, $\Delta B(cp,C) = 78h(cp) + 312\dot{c}(C) + 312 \geq \Delta J(cp,C)$.
	
	\item Suppose $s$ is an update to the local variable $l$ on line~\ref{ln:cleanup:leaf_update} of a cleanup phase $cp$. Only the $J(cp)$ account can change due to step $s$. By Lemma~\ref{J_forward_traversal}, $\Delta J(cp,C) \leq 0$. Therefore, $\Delta B(cp, C) = 0 \geq \Delta J(cp,C)$.
	
	\item Suppose $s$ is a step of a cleanup phase $cp$ that pops a node from its stack and updates its local variable $l$ to point to this node. Only $J(cp)$ can change due to step $s$. By Lemma~\ref{J_backtracking}, $\Delta J(cp,C) \leq 0$. Therefore, $\Delta B(cp, C) = 0 \geq \Delta J(cp,C)$.
	
	\item Suppose $s$ is a successful freezing CAS on a downward node $x$. By Lemma~\ref{H_down_freeze}, $H(x, C) \leq H(x, C^-) - 1$, and for each other node $u$, $H(u, C) \leq H(u, C^-)$. No other term in $J$ changes as a result of $s$. Thus for an update or cleanup phase $xp'$ where $x \in target(xp')$, we have $\Delta J(xp',C) \leq -2$ since $J(xp',C)$ includes the term $2H(x,C)$. By T1-B, $\Delta B(xp',C) = -2$, so $\Delta B(xp',C) \geq \Delta J(xp',C)$. 
	
	\item Suppose $s$ is a successful freezing CAS on a cross node $x$. By Lemma~\ref{H_cross_freeze}, $\Delta H(u,C) \leq 0$ for every node $u$ in the chromatic tree. No other term in $J$ changes as a result of $s$. Therefore, $\Delta J(xp',C) \leq 0$ for any update or cleanup phase $xp'$. Since the $B$ accounts do not changes as a result of $s$, $\Delta B(xp',C) = 0 \geq \Delta J(xp',C)$.
	
	\item Suppose $s$ is a successful update CAS. By Lemma~\ref{J_ucas_update}, $\Delta J(up', C) \leq 450$ for each active update phase $up'$ in $C$. By D2-BUP, $\Delta B(up', C) = 4934 > \Delta J(up', C)$.
	
	By Lemma~\ref{J_ucas_cleanup}, $\Delta J(cp', C) \leq 4934$ for each active cleanup phase $cp'$ in $C$. By D2-BCP, $\Delta B(cp', C) = 4934 \geq \Delta J(cp', C)$.
	
	\item Suppose $s$ is a successful commit step of an SCX which changes its SCX-record state from \textsc{InProgress} to \textsc{Committed}. By inspection of Figure~\ref{fig_transformations}, only the root of each transformation is unfrozen. The $isFrozen$ variable of this node changes from 1 to 0. No other variables are changed, so for each node $u$, $H(u)$ increases by at most 1. Therefore, by the definition of $J$, $J(xp)$ increases by at most $2|targets(xp,C)|$ as a result of the increase in $H(u)$ for each node $u \in targets(xp,C)$.
	
	For all active update phases $up'$, $|targets(up',C)| \leq 5$, so $\Delta J(up', C) \leq 10$. By D3-BUP, $\Delta B(up',C) = 26 > \Delta J(up', C)$.
	
	For all active cleanup phases $cp'$, $|targets(cp',C)| = 13$, so $\Delta J(cp', C) \leq 26$. By D3-BCP, $\Delta B(cp',C) = 26 \geq \Delta J(cp', C)$.
	
	\item Suppose $s$ is a successful abort step of an SCX. By Lemma~\ref{H_abort}, $H(u)$ does not change for any node $u$. No other terms in $J$ change as a result of $s$, so $\Delta J(xp') = 0$ for all active update or cleanup phases $xp'$. Since the $B$ accounts do not change as a result of $s$, $\Delta B(xp',C) = \Delta J(xp',C) = 0$.

	\item Suppose $s$ is the completion of an attempt of an update or cleanup phase $xp$ that does not complete an operation. By rule A4 of $abort$ variables, $abort(u)$ may change from 1 to 0 for some nodes $u$ in the chromatic tree. No other variables change as a result, and so for any update or cleanup phase $xp'$, $\Delta J(xp') \leq 0$. The $B$ accounts do not change as a result of $s$, so $\Delta B(xp',C) \geq \Delta J(xp',C)$.
	
	\item Suppose $s$ is the completion of an operation that was in its update or cleanup phase $xp$. Then $J(xp,C) = 0$ since $xp$ is no longer active. By Observation~\ref{obs_J_pos}, $J(xp,C^-) \geq 0$, so $\Delta J(xp,C) \leq 0$. Since $B(xp)$ does not change as a result of $s$, $\Delta B(xp,C) = 0 \geq \Delta J(xp,C)$. Any dollars left over in $B(xp)$ are not needed to account for $xp$'s failed attempts.

	Now consider any other active update or cleanup phase $xp' \neq xp$. By rule A4 of $abort$ variables, $abort(u)$ may change from 1 to 0 for some nodes $u$ in the chromatic tree. Additionally, $H(u)$ may decrease from $\dot{c}(C) = \dot{c}(C^-) - 1$. No other variables change, so $\Delta J(xp',C) \leq 0$. The $B$ accounts do not change as a result of $s$, so $\Delta B(xp',C) \geq \Delta J(xp',C)$.
\end{itemize}
\end{proof}

\begin{lemma}\label{BXP1}\normalfont
	For all update phases or cleanup phases $xp$ and all configurations $C$, $B(xp, C) \geq 0$.
\end{lemma}

\begin{proof} By Observation~\ref{obs_J_pos}, $J(xp,C) \geq 0$. Thus, by Lemma~\ref{J_positive}, $B(xp,C) \geq J(xp,C) \geq 0$.
\end{proof}

Since all bank accounts have non-negative balance, Property P3 of the bank is satisfied. 

\section{Conclusion}\label{section_conclusion}
We have shown that the amortized step complexity of an implementation of a non-blocking chromatic tree is $O(\dot{c}(\alpha) + \log n(op))$. It would be interesting to see if amortized step complexity $O(\dot{c}(op) + \log n(op))$ can be shown for the chromatic tree. This is more challenging because one must argue that an operation with high contention does not perform an excessive amount of rebalancing. This may require a more detailed analysis of the chromatic tree than what is shown by Theorem~\ref{thm_rebal}. Alternatively, one may try to show a lower bound of $\Omega(\dot{c}(\alpha) + \log n(op))$. Finally, we believe that the techniques presented here can be applied to other balanced binary search tree implementations using LLX and SCX.

\printindex

\end{document}